\documentclass[acmsmall]{acmart}

\acmJournal{PACMPL}
\acmVolume{1}
\acmNumber{CONF} % CONF = POPL or ICFP or OOPSLA
\acmArticle{1}
\acmYear{2018}
\acmMonth{1}
\acmDOI{} % \acmDOI{10.1145/nnnnnnn.nnnnnnn}
\startPage{1}

\usepackage{enumitem}
\usepackage{comment}

\setcopyright{none}
\usepackage{booktabs}   %% For formal tables:
\usepackage{subcaption} %% For complex figures with subfigures/subcaptions
\usepackage{graphicx}
\usepackage{float}
\usepackage{math-setup}
\usepackage[section]{placeins}

\newcommand{\later}[1]{{\color{red}\{todo: #1\}}}

\begin{document}

%% Title information
\title[Type-Directed Synthesis of Visualizations from Natural Language Queries]{Type-Directed Synthesis of Visualizations from Natural Language Queries}         %% [Short Title] is optional;
                                        %% when present, will be used in
                                        %% header instead of Full Title.
% \titlenote{with title note}             %% \titlenote is optional;
                                        %% can be repeated if necessary;
                                        %% contents suppressed with 'anonymous'
% \subtitle{Subtitle}                     %% \subtitle is optional
% \subtitlenote{with subtitle note}       %% \subtitlenote is optional;
                                        %% can be repeated if necessary;
                                        %% contents suppressed with 'anonymous'

%% Author information
%% Contents and number of authors suppressed with 'anonymous'.
%% Each author should be introduced by \author, followed by
%% \authornote (optional), \orcid (optional), \affiliation, and
%% \email.
%% An author may have multiple affiliations and/or emails; repeat the
%% appropriate command.
%% Many elements are not rendered, but should be provided for metadata
%% extraction tools.

%% Author with single affiliation.
\author{Qiaochu Chen}
\affiliation{
  \institution{University of Texas at Austin} 
  \city{Austin}
  \state{Texas}
  \country{USA}
}
\email{qchen@cs.utexas.edu}

\author{Shankara Pailoor}
\affiliation{
  \institution{University of Texas at Austin}
  \city{Austin}
  \state{Texas}
  \country{USA} 
}
\email{spailoor@cs.utexas.edu}

%% Author with two affiliations and emails.
\author{Celeste Barnaby}
\affiliation{
  \institution{University of Texas at Austin}
  \city{Austin}
  \state{Texas}
  \country{USA} 
}
\email{celestebarnaby@utexas.edu}

\author{Abby Criswell}
\affiliation{
  \institution{University of Texas at Austin}
  \city{Austin}
  \state{Texas}
  \country{USA} 
}
\email{abbycriswell@utexas.edu}

\author{Chenglong Wang}
\affiliation{
  \institution{Microsoft Research}
  \city{Redmond}
  \state{Washington}
  \country{USA} 
}
\email{chenglong.wang@microsoft.com}

\author{Greg Durrett}
\affiliation{
  \institution{University of Texas at Austin} 
  \city{Austin}
  \state{Texas}
  \country{USA}
}
\email{gdurrett@cs.utexas.edu}

\author{Isil Dillig}
\affiliation{
  \institution{University of Texas at Austin}
  \city{Austin}
  \state{Texas}
  \country{USA}
}
\email{isil@cs.utexas.edu}

%% Abstract
%% Note: \begin{abstract}...\end{abstract} environment must come
%% before \maketitle command
\begin{abstract}
We propose a new technique based on program synthesis for automatically generating visualizations from natural language queries. Our method parses the natural language query into a refinement type specification using the \emph{intents-and-slots paradigm} and leverages type-directed synthesis to generate a set of visualization programs that are most likely to meet the user's intent. Our refinement type system captures useful hints present in the natural language query and allows the synthesis algorithm to reject visualizations that violate well-established design guidelines for the input data set. We have implemented our ideas in a tool called \toolname and evaluated it on {\sc NLVCorpus}, which consists of 3 popular datasets and over 700 real-world natural language queries. Our experiments show that \toolname significantly outperforms state-of-the-art natural language based visualization tools, including transformer and rule-based ones.

\end{abstract}

%% 2012 ACM Computing Classification System (CSS) concepts
%% Generate at 'http://dl.acm.org/ccs/ccs.cfm'.

%% End of generated code

%% Keywords
%% comma separated list
% \keywords{keyword1, keyword2, keyword3}  %% \keywords are mandatory in final camera-ready submission

%% \maketitle
%% Note: \maketitle command must come after title commands, author
%% commands, abstract environment, Computing Classification System
%% environment and commands, and keywords command.
\maketitle

\section{Introduction}
\label{sec:intro}

Natural language interfaces (NLIs)~\cite{articulate,datatone,flowsense} for visualization promise to democratize the visualization authoring process. Given a dataset (often a relational table) and a natural language description, an NLI can generate a set of visualizations that most likely meet the user's intent. For instance, given the dataset shown in~\autoref{fig:intro-table} and a query such as ``give me a scatter plot that shows the fuel economy of all car models'',  an NLI can, in principle, generate the scatter plot shown on the right side of ~\autoref{fig:intro-table}. In this way, even a user with no programming experience can generate  visualizations from large-scale data.

\begin{figure}[t]
\small
\centering
\begin{minipage}{.48\textwidth}
% \vspace*{-0.8cm}
\begin{tabular}{cccc}
    \hline
    \rowcolor{platinum}
     \textbf{Model} & \textbf{Fuel\_economy} & \textbf{Body\_style} & \textbf{Origin} \\
     \hline
     \rowcolor{whitesmoke}
     S-101 & 32 & Sedan & Japan  \\
     \hline
     \rowcolor{whitesmoke}
     S-102 & 39 & SUV & USA \\
     \hline
     \rowcolor{whitesmoke}
     S-103 & 22 & Pickup & USA \\
     \hline
     \rowcolor{whitesmoke}
     S-104 & 39 & Hatchbacks & Japan \\
     \hline
     \rowcolor{whitesmoke}
     \ldots & \ldots & \ldots & \ldots\\
     \hline
\end{tabular}
\end{minipage}
\begin{minipage}{.5\textwidth}
% \vspace*{-0.5cm}
\includegraphics[width=0.9\textwidth, trim=0 560 200 30, clip]{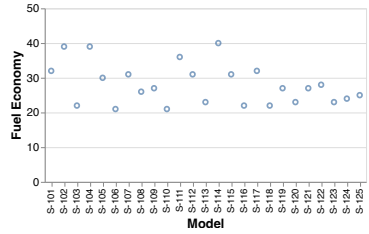}
\end{minipage}
\vspace*{-0.3cm}
\caption{On the left is a cars dataset showing the fuel economy, body style and origin for each model. The plot on the right is for the query ``give me a scatter plot that shows the fuel economy of all car models''.}
\vspace*{-0.5cm}
\label{fig:intro-table}
\end{figure}

\begin{figure}
    \centering
    % \vspace*{-0.1cm}
    \begin{subfigure}[t]{0.50\textwidth}
    \centering
    \includegraphics[width=0.9\textwidth, trim=40 480 250 0, clip]{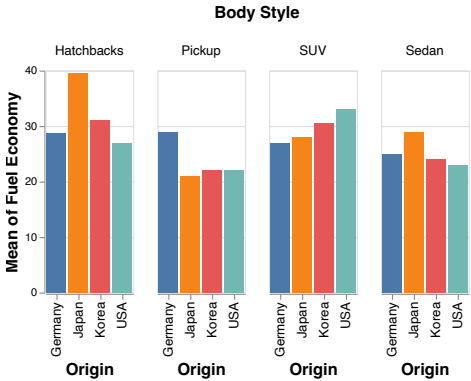}
    \vspace*{-0.4cm}
    \caption{The user-intended plot.}
    \label{fig:me_correct}
    \end{subfigure}
    \hspace*{0.5cm}
    \begin{subfigure}[t]{0.45\textwidth}
    \centering
    % \vspace*{0.5cm}
    \includegraphics[width=0.88\textwidth, trim=30 520 300 30, clip]{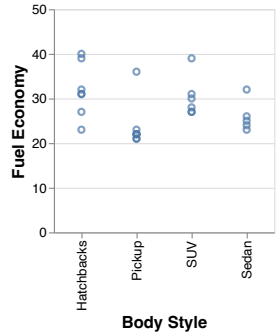}
    \vspace*{-0.4cm}
    \caption{One of the plots returned by {\sc NL4DV}.}
    \label{fig:me_nl4dv}
    \end{subfigure}
    \vspace*{-0.5cm}
    \caption{Figures for ``show the fuel efficiency for cars from different countries segregated based on body style''.}
    % \vspace*{-0.6cm}
\end{figure}

While existing NLIs are effective in producing relatively simple visualizations, a recent study~\cite{nlvcorpus} found that these tools are unable to generate more complex visualizations, such as those that involve subplots or that require performing non-trivial transformations on the input data. 
For example, for the input query ``generate a graph to show the fuel efficiency for cars from different countries segregated based on  body style'' and the dataset from Figure~\ref{fig:intro-table},  state-of-the-art tools  return the plot shown in Figure~\ref{fig:me_nl4dv} as opposed to the ideal plot shown in Figure~\ref{fig:me_correct}.
%which is a grouped bar chart that shows the mean fuel economy for cars from different origins. 
%In particular, the plot in Figure~\ref{fig:me_nl4dv} not only misses one of the columns specified by the user (i.e. ``countries'') but also forgets to make the column ``Body style'' a subplot encoding. 

In this paper, we propose a new technique for generating visualizations from natural language descriptions. Our method combines NLP techniques with program synthesis to address several challenging aspects of data visualization. In particular, our technique can handle fairly ambiguous natural language queries, including those that do not fully specify the desired plot type. In addition, our method can perform transformations and aggregations on the input data, allowing it to handle visualizations that require non-trivial data wrangling.  As an example,  our method can produce the correct plot, shown in Figure~\ref{fig:me_correct}, for the input query mentioned earlier.

At the heart of our technique lies a refinement type system that can be used to express  properties of the desired visualization. At a high level, our method first uses state-of-the-art NLP techniques, namely a BERT-based \cite{bert} intent-and-slots model ~\cite{tur-et-al-intents}, to parse the natural language description into a set of likely refinement type specifications for the desired visualization. Then, for each refinement type specification, our method performs type-directed program synthesis to generate a set of visualization programs of the appropriate type, using a notion of \emph{type compatability} to prune large parts of the search space. Hence, the refinement type system is useful  not only   as a specification mechanism but also for guiding synthesis and reducing the search space.

A distinguishing feature of the synthesis problem in our setting is that a single visualization task typically results in many synthesis problems, one for each  refinement type specification inferred by the parser. However,  because each invocation of the synthesizer can be quite expensive, it is important to reuse information across different synthesis problems. Our approach addresses this concern by learning so-called \emph{synthesis lemmas} that can be used to prove unrealizability of future synthesis tasks. In particular, our approach leverages a novel notion of \emph{refinement type interpolants} to learn useful facts that can be reused across different synthesis goals involving the same data set.

We have implemented our proposed approach as a new tool called \toolname and evaluated it on {\sc NLVCorpus} \cite{nlvcorpus}, which consists of 3 popular visualization datasets and over 700 real-world natural language queries. Our evaluation demonstrates that \toolname yields significantly better results compared to existing state-of-the-art baselines, including transformer and rule-based NLIs. We also perform ablation studies to evaluate the importance of our proposed techniques and show that they all contribute to making our approach practical.

To summarize, this paper makes the following key contributions:

\begin{itemize}[leftmargin=*]
\item We propose a new synthesis-based technique for generating visualizations from an input data set and natural language description.
\item We introduce a refinement type system that is useful both as a specification mechanism and for guiding program synthesis.
\item We describe a technique  based on the intents-and-slots paradigm for parsing natural language descriptions into refinement type specifications.
\item We propose a type-directed synthesis algorithm that uses a notion of \emph{type compatibility} to prune the search space and learns \emph{synthesis lemmas} that are useful across different synthesis attempts.
\item We implement our approach in a tool called \toolname and perform a large-scale evaluation on over 700 real-world visualization tasks.
\end{itemize}

\section{Overview}\label{sec:overview}

We give a high-level overview of our technique with the aid of the motivating example introduced in Section~\ref{sec:intro}. In particular, consider the dataset from Figure~\ref{fig:intro-table} listing the fuel economy of different cars and the following natural language  query:
\[
\emph{"Show the fuel efficiency of cars from different countries segregated based on body style"}
\]
Given this query, our tool, \toolname, generates the visualizations shown in Figure~\ref{fig:me_correct_2}. Among the plots shown here, the top one is the intended one, with corresponding visualization script shown in Figure~\ref{fig:me_prog}. We now explain how \toolname is able to generate these visualizations, highlighting salient features of our approach. 

\begin{figure}
\begin{minipage}[t]{0.4\textwidth}
    \centering
    % \begin{minipage}[t]{0.50\textwidth}
    \vspace*{-0.5cm}
    \includegraphics[width=0.7\textwidth, trim=40 500 250 30, clip]{figures/me-gt2.pdf}
    % \end{minipage}
    % \begin{minipage}[t]{0.45\textwidth}
    \hspace*{0.5cm}
    \includegraphics[width=0.65\textwidth, trim=40 520 350 30, clip]{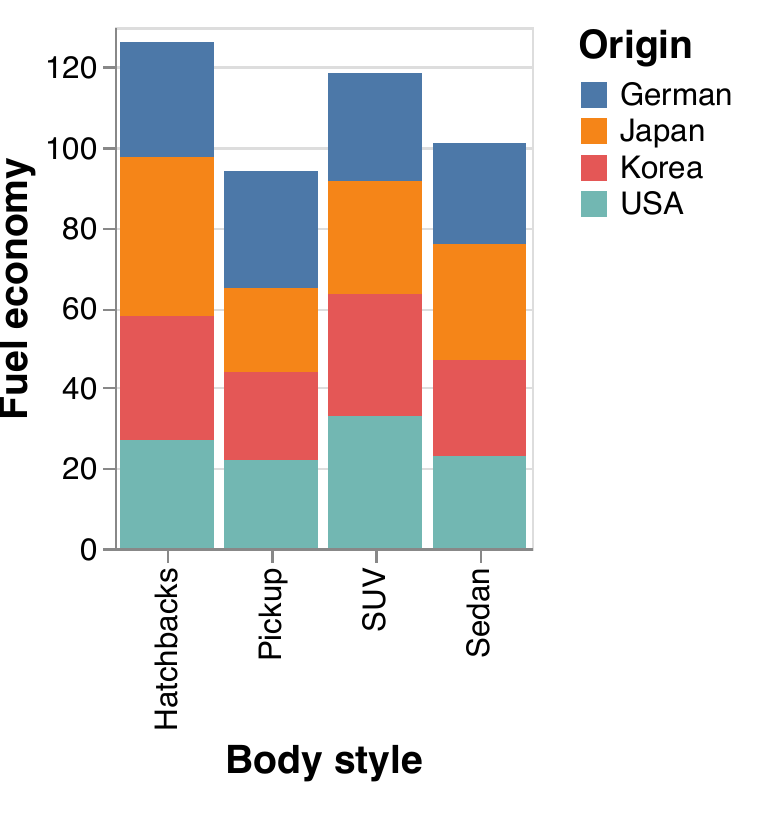}
    % \end{minipage}
    \vspace*{-0.5cm}
    \caption{\small Some plots \toolname returns for the example query. The one on the top is the intended one. The one at the bottom is a plot that is also consistent with the  query.}
    \label{fig:me_correct_2}
\end{minipage}
\hspace{0.5cm}
\begin{minipage}[t]{0.55\textwidth}
    \centering
    \vspace*{0.1cm} 
    \scriptsize
    \begin{align*}
        % & \\
        & \mathsf{let\ } T =  \qquad \qquad \qquad \qquad \qquad \mathsf{ \color{blue} \sslash\ Table\ transformation\ program} \\
        & \qquad {\sf summarize}(  \\
        & \qquad \qquad \qquad {\sf select}( T_{in}, \{{\sf Origin}, {\sf Fuel\_economy}, {\sf Body\_style}\}),\\
        & \qquad \qquad \qquad \{{\sf Origin}, {\sf Body\_style}\}, \\
        & \qquad \qquad \qquad {\sf mean}, \\
        & \qquad \qquad \qquad {\sf Fuel\_economy}) \\
        & \mathsf{ in: } \qquad \qquad \qquad \qquad \qquad \qquad \mathsf{ \color{blue} \sslash\ Plotting\ program} \\
        & \qquad  {\sf Bar}( T, \\
        & \qquad \qquad \colx={\sf Origin}, \\
        & \qquad \qquad \coly={\sf Fuel\_economy}, \\
        & \qquad \qquad \colsub={\sf Body\_style}) \\
        % & \\
    \end{align*}
    \normalsize
    % \vspace*{-0.8cm}
    \caption{\small The visualization program synthesized by \toolname that generates the plot on the top of Figure~\ref{fig:me_correct_2}. The top part is a table transformation program that performs a mean operation on the {\sf Fuel\_economy} column. The bottom portion generates a bar chart from the output of the above table transformation program. }
    \label{fig:me_prog}
    % \includegraphics[width=\textwidth, trim=40 540 300 0, clip]{figures/me-error.pdf}
    % \caption{A bogus plot}
    % \label{fig:me_wrong}

\end{minipage}
\vspace*{-0.5cm}
\end{figure}

\paragraph{\bf Structure of visualization programs.} Similar to prior work~\cite{viser}, the  visualization programs synthesized by \toolname consist of two parts, namely a \emph{table transformation program} $\ptable$ and a \emph{plotting program}  $\pplot$ (see Figure~\ref{fig:me_prog}). Given these programs, \toolname produces a visualization  by first applying $\ptable$ to the input data set to obtain a transformed table $T$ and then applying the plotting program $\pplot$ to $T$.  Since many real-world visualizations tasks require non-trivial data wrangling, synthesis of table transformations is a crucial aspect of the \toolname workflow. We describe the domain-specific language used for visualizations in more detail in Section~\ref{sec:dsl}. 
%Unlike prior work, the inferred goal types  are necessary rather than sufficient conditions for correctness in our context; hence, we use a notion of \emph{type compatibility} rather than \emph{subtyping} to prune infeasible partial programs.
% (Can we move/remove this sentence? This paragraph isn't about typing.)

\paragraph{\bf Motivation for refinement types.} As mentioned in Section~\ref{sec:intro}, \toolname parses the natural language query into a formal specification rather than going directly from natural language to a visualization program. This design choice hinges on two key observations: First, the natural language description often does not contain sufficient information to map it directly to a program. For instance, in our running example, the NL query does not mention anything about a bar graph. Second, the data set to be visualized also contains valuable information for deciding which visualizations make more sense. As an example, looking at the data set, we see that fuel economy is continuous (as opposed to discrete), so it would not be suitable as the x-axis for a bar graph. For these reasons, \toolname parses the NL query into an intermediate specification, which is then supplied as an input to the synthesizer.

In this work, we use refinement types as our specification because both base types and logical qualifiers are useful for guiding synthesis. In particular, record types are useful for distinguishing between different types of tabular data, and  logical qualifiers capture other forms of hints present in the natural language.  For instance, based on the natural language query given above, it is reasonable to conjecture that the color encoding of the plot should be based on country, and our type system allows expressing such information as part of the logical qualifier. We discuss our refinement type system in more detail in Section~\ref{sec:type}. 

\paragraph{\bf From NL queries to refinement types.} As a precursor to synthesis, \toolname first uses state-of-the-art NLP techniques to extract refinement type specifications from the natural language query. In particular, the extracted specifications are of the form $(\rtsym_p, \rtsym_t)$, where $\rtsym_p$ is the output type for the plotting program and $\rtsym_t$ is the output type for the table transformation program. Intuitively, we parse the NL query into two different refinement types as our synthesis procedure generates the plotting and table transformation programs independently.

Our technique for parsing a natural language query to a refinement type consists of two steps. First, we use the technique of \emph{intent classification}~\cite{tur-et-al-intents} to infer some of the base types (e.g. BarPlot, ScatterPlot) as well as which \emph{types} of predicates should be involved in the logical qualifier.  For our running example, the intent classifier is able to predict that the desired plot is a BarPlot based on the training data. In addition, note that the NL query  hints at the fact that the color encoding of the plot is based on country (i.e., ``Origin" column in the data set). Such information is encoded using so-called \emph{{syntactic constraints}} in the logical qualifiers. The  intent classifier can decide whether the NL query contains such {syntactic constraints}, but it cannot decide what the  \emph{arguments} of the  predicate are. Hence, in a second step, we use a {natural-language-processing} technique known as \emph{slot filling}~\cite{jeong-lee-2006-exploiting} to decide the arguments of the inferred predicates. For our running example, the slot-filling technique can infer that {the graph's color is likely to be the Origin field of the input data set and generates a logical qualifier that involves this syntactic constraint}. We describe our technique for parsing the NL query into a refinement type specification in more detail in Section~\ref{sec:parse}.

\begin{figure}[t]
    \centering
    % \vspace{-0.5cm}
    \includegraphics[width=0.8\textwidth, trim=0 0 200 540, clip]{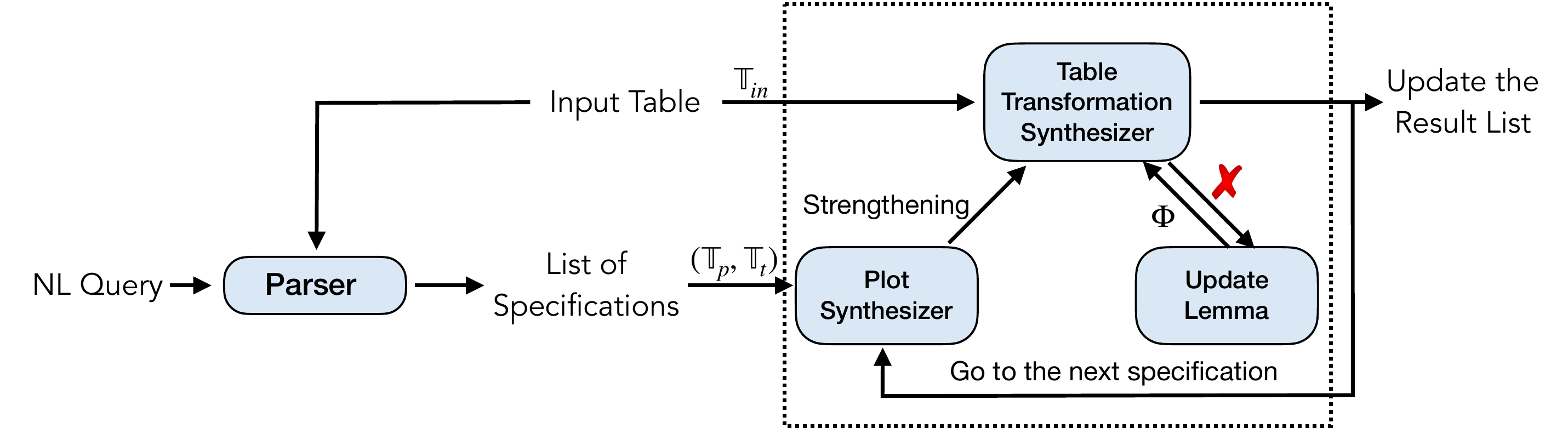}
    \vspace*{-0.3cm}
    \caption{Overview of the workflow}
    \vspace*{-0.3cm}
    \label{fig:overview}
\end{figure}

\begin{figure}
% \vspace*{-0.5cm}
\begin{minipage}[t]{0.3\textwidth}
\includegraphics[width=\textwidth, trim=0 0 1460 840, clip]{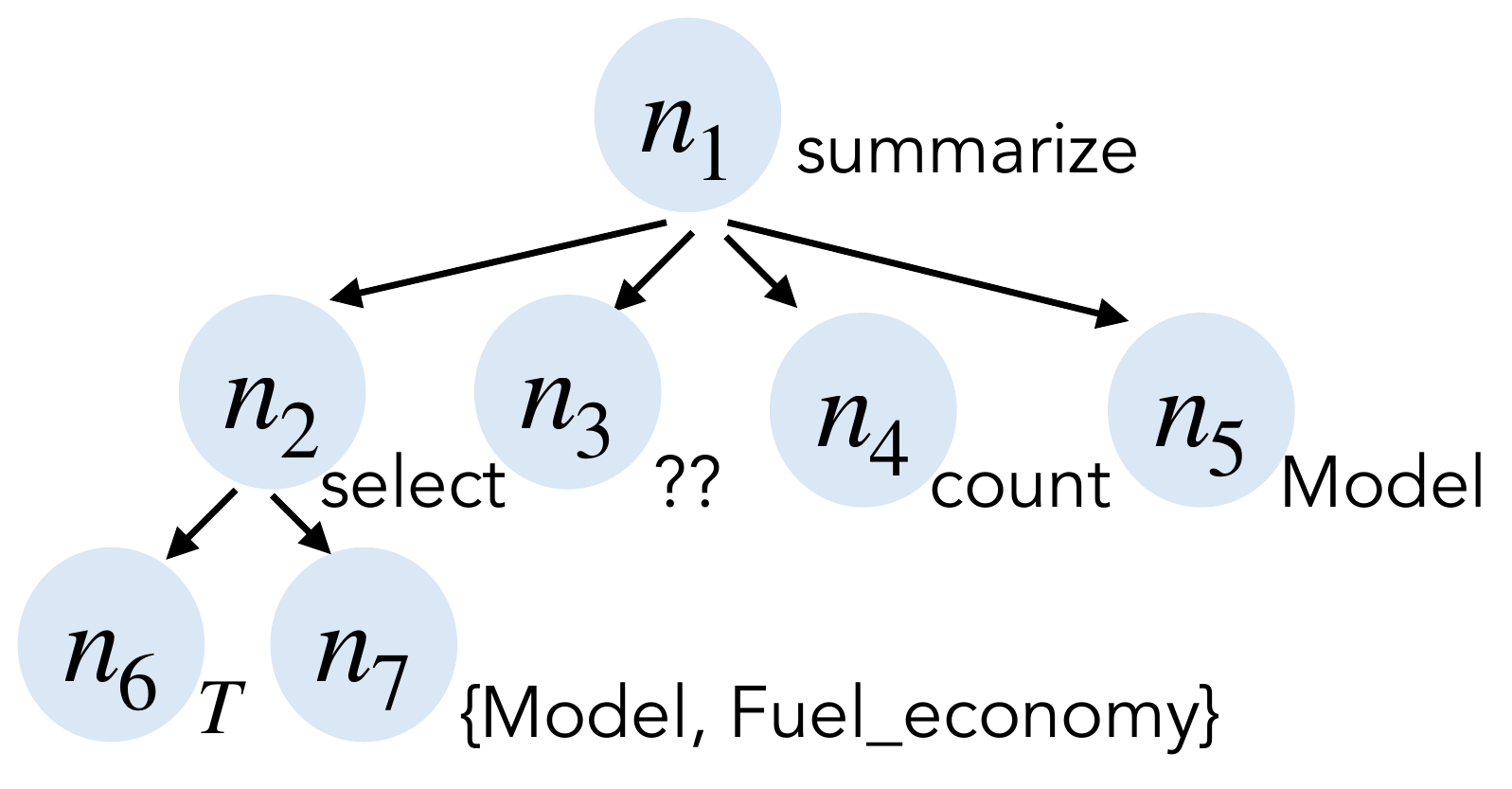}
% \caption{Infeasible Partial Program}
% \label{fig:infeasibleprog}
\end{minipage}
%\begin{subfigure}[t]{\textwidth}
%\begin{equation}
%\{v : {\sf Table}({\sf Body\_style} : {\sf Nominal}, {\sf Model} : {\sf Nominal}, {\sf Fuel\_economy} : {\sf Continuous}, {\sf Origin} : {\sf Nominal})\}
%\end{equation}
%\caption{Input type to program in Figure \ref{fig:infeasibleprog}}
%\end{subfigure}
\begin{minipage}[t]{0.675\textwidth}
\centering
\vspace*{-2.5cm}
\scriptsize
\begin{align*}
    \\
    & \mathtt{Goal\ type\ of\ } n_1: \\
    & \qquad \{\nu : {\sf Table}(\{{\sf Model} : {\sf Discrete}, {\sf Fuel\_economy} : {\sf Qualitative}\}) \ | \ \pi(\nu.{\sf Model}, {\sf count})\} \\
    & \mathtt{Goal\ type\ of\ } n_2: \\
    & \qquad  \{\nu : {\sf Table}(\{{\sf Model} : \top, {\sf Fuel\_economy} : {\sf Qualitative}\}) \ | \ {\sf True}\}
\end{align*}
% \begin{tabular}{|c|c|}
%         \hline
%         Subprogram & Goal Type \\
%         \hline
%         {\sf summarize} &  $\{\nu : {\sf Table}(\{{\sf Model} : {\sf Discrete}, {\sf Fuel\_economy} : {\sf Qualitative}\}) \ | \ \pi(\nu.{\sf Model}, {\sf count})\}$  \\
%         \hline
%         {\sf select} & $\{\nu : {\sf Table}(\{{\sf Model} : \top, {\sf Fuel\_economy} : {\sf Qualitative}\}) \ | \ {\sf True}\}$ \\
%         \hline
%     \end{tabular}
    % \caption{The propagated goal types for the program in Figure \ref{fig:infeasibleprog}.}
    % \label{fig:propgoaltypes}
\end{minipage}
\vspace*{-0.3cm}
\caption{Pruning Example. An abstract syntax tree of a partial program is shown on the left. On the right we show the goal type annotation at node $n_1$ and $n_2$.}
\label{fig:pruningexample}
\vspace*{-0.5cm}
\end{figure}

\paragraph{\bf Synthesis workflow.}  Figure~\ref{fig:overview} shows the high-level workflow of our synthesis algorithm. For \emph{each} specification $(\rtsym_p, \rtsym_t)$ generated by the parser, our synthesizer generates a \emph{set} of visualization programs that satisfy it. At a high level, the synthesis algorithm first generates a plotting program $\pplot$ such that the output type of $\pplot$ is a subtype of $\rtsym_p$. The input type of $\pplot$ is then used to \emph{strengthen} the parsed specification $\rtsym_t$ of the table transformation program to $\rtsym_t'$. For instance, in our running example, suppose we synthesize the following plotting program:
\small
\[
     {\sf Bar}(T, \colx={\sf Body\_style}, \coly={\sf Fuel\_economy}, \colco={\sf Origin})
\]
\normalsize
Such a program only makes sense if there is a unique $y$ value for every $(x, color)$ pair, so our method strengthens the output of the table transformation program with the following constraint:
\small
\[
    |{\sf Proj}(\nu, \{{\sf Body\_style}, {\sf Origin}\})| \geq |{\sf Proj}(\nu, \{{\sf Fuel\_economy}\})|
\]
\normalsize
This constraint states that the cardinality (number of unique tuples) of the output table projected on the $x$ ({\sf Body\_style}) and $color$ ({\sf Origin}) columns should be at least as big as the cardinality when projected onto the $y$ ({\sf Fuel\_economy}) column. This constraint serves as a logical qualifier for the output type of the table transformation program and is used to reduce the search space that the synthesizer needs to explore, as discussed in Section~\ref{sec:synthoverview}. 

\paragraph{\bf Type-directed synthesis.} In addition to using refinement types as the specification, our technique also uses them to guide synthesis as in prior work~\cite{synquid}.  In more detail, our algorithm performs top-down enumerative search, starting with a completely unconstrained program and expanding a non-terminal (i.e., "hole") in the partial program using one of the grammar productions.
Each hole  is annotated with a so-called \emph{goal-type} that is propagated backwards using the type system and the initial specification obtained from the NL query. As explained in more detail in Section~\ref{sec:synthoverview}, the goal type is used to decide (1) which grammar productions are applicable when performing top-down enumeration, and (2) whether a partial program is infeasible, meaning that its actual type is inconsistent with the annotated goal type. However, unlike prior work,  our approach uses a notion of \emph{type compatibility} (Section \ref{sec:incompatibility}) as opposed to \emph{subtyping}  in order to ensure that we do not rule out correct programs. 

For example, suppose we want to synthesize a table transformation that satisfies the following goal type 
\small
\[
  \{\nu : {\sf Table}({\sf Model} : {\sf Discrete}, {\sf Fuel\_economy} : {\sf Qualitative}) \ | \ \pi(\nu.{\sf Model}, {\sf count})\}
\]
\normalsize
for the table in Figure \ref{fig:intro-table}. {Here, the goal type comes in the form of a refinement type that describes the base type ${\sf Table}({\sf Model} : {\sf Discrete}, {\sf Fuel\_economy} : {\sf Qualitative})$ annotated with the predicate  $\pi(\nu.{\sf Model}, {\sf count})$. } The base type describes a table with attributes {\sf Model} and {\sf Fuel\_economy}, whose types are {\sf Discrete} and 
{\sf Qualitative} respectively. The qualifier, $\pi(\nu.{\sf Model}, {\sf count})$ is a {syntactic constraint} that indicates that the \textsf{Model} attribute of the output table is obtained by applying the \textsf{count} operation.  During synthesis, \toolname starts with a program that is a single hole and iteratively expands it using productions in the grammar. Whenever \toolname expands a hole, it propagates the above goal type to newly produced holes in the partial program.  For instance,  Figure \ref{fig:pruningexample} shows a partial program with the annotated goal type of node $n_2$. Using our  type system, \toolname can prove that this partial program is infeasible because the actual type of the term rooted at node $n_2$ is incompatible with its annotated goal type. This is because {\sf Fuel\_economy}  has type {\sf Continuous} in the input table, which  is inconsistent with the goal type labeling node $n_2$, where {\sf Fuel\_economy} is required to be {\sf Qualitative}. 

%Instead of expanding a hole, \toolname deduces that the program is actually infeasible and discards it as the actual type of the {\tt select} subprogram is inconsistent with its goal type (Goal type of $n_2$). In particular, its goal type requires the {\sf Fuel\_economy} column to have type {\sf Qualitative}; however, since {\sf Fuel\_economy} has type {\sf Continuous} in the input table, it will remain {\sf Continuous} after calling {\tt select}.

\paragraph{\bf Type-directed lemma learning.} A unique feature of our approach is its ability to learn \emph{synthesis lemmas} that can be used across different synthesis tasks involving the same data set.\footnote{While there are prior techniques that can learn useful facts during enumeration~\cite{neo,concord}, the facts they learn are not reusable across different synthesis goals.} To see why such lemmas are useful, recall that we want to generate multiple visualizations to show to the user, so we need to explore many different programs that \emph{could} be consistent with  the NL query. In general, there are multiple plausible specifications one can extract from the NL query, and there are multiple programs that satisfy each specification. Hence, our approach needs to explore many different programs during a single visualization session. 
%However, since even finding a single program that satisfies a given specification can be expensive, it is important to build a re-usable database of useful facts across different iterations. 

Our approach addresses this concern by using the  type system to learn synthesis lemmas.  In particular, a \emph{synthesis lemma} is a pair of refinement types $(\lemmaG, \lemmaR)$ such that any program with goal type $\lemmaG$ also needs to be ``consistent'' (in a sense made precise in Section~\ref{sec:type}) with refinement type $\lemmaR$. Hence, if we encounter a synthesis goal (or sub-goal) that is a subtype of $\lemmaG$ but that is inconsistent with $\lemmaR$, we immediately conclude that the synthesis task is infeasible. Our synthesis algorithm learns such lemmas by inferring so-called \emph{refinement type interpolants} whenever it encounters an infeasible partial program. We discuss the algorithm for type-directed lemma learning in Section~\ref{sec:lemma}.

Going back to our running example, consider the same infeasible partial program in Figure \ref{fig:pruningexample}.  The root cause of this failure is that the program was unable to convert the {\sf Fuel\_economy} column from a {\sf Continuous} type to {\sf Qualitative}. In fact, no program in our DSL can achieve this transformation. \toolname automatically captures this fact by generating the following synthesis lemma:
\small
\[
 (\{\nu : {\sf Table}({\sf Fuel\_economy}: {\sf Qualitative})\}, \bot). 
\]
\normalsize
Hence, if we ever encounter a specification such as $\{\nu : {\sf Table}({\sf Body\_style} : {\sf Nominal}, {\sf Fuel\_economy} : {\sf Nominal}) | \ \phi_2\}$ that is a subtype of $\{\nu : {\sf Table}({\sf Fuel\_economy}: {\sf Qualitative})\}$, \toolname can immediately conclude that this  goal  is unrealizable without even attempting synthesis.

\section{Domain-Specific Language for Visualizations}
\label{sec:dsl}

\begin{figure}[t]
\small
% \vspace*{-0.5cm}
\begin{minipage}[t]{0.45\textwidth}
\textbf{Visualization DSL}
\[
\begin{array}{r l }
    % \textbf{Visualization DSL} \\ 
    \prog_v := & \lambda T_{in}. \ \mathsf{let\ } T = \ptable(T_{in}) \mathsf{\ in\ } \pplot \\ 
    % P := & \lambda T. (\pplot \circ \ptable)(T) & \text{visualization program} \\ 
    \\
\end{array}
\]
\textbf{Sub-DSL for plotting} 
\[
\begin{array}{r l }   
    % &\textbf{Sub-DSL for plotting} \\ 
    \pplot := & f(T, \colx, \coly, \colco, \colsub) \\
    f := & \mathsf{Bar} \mid \mathsf{Scatter} \mid \mathsf{Line} \mid \mathsf{Area} \\ 
\end{array}
\]
\end{minipage}
\begin{minipage}[t]{0.45\textwidth}
\textbf{Sub-DSL for table transformations}
\[
\begin{array}{r l }
    \ptable := & \lambda T. \ e  \\
    e := & T   \\ 
    | & \mathsf{bin} (e, n, \coltarg) \\
    | & \mathsf{filter} (e, val_1 \ op\ val_2)  \\
    | & \mathsf{summarize}(e, \overline{\colkey}, \alpha, \coltarg) \\ 
    | & \mathsf{mutate}(e, \coltarg, op, \overline{\colarg}) \\
    | & \mathsf{select}(e, \overline{\colarg}) \\
    val := & const \mid c \\
    \alpha := & \mathsf{mean} \mid \mathsf{sum} \mid \mathsf{count} \\
    \end{array}
 \]
\end{minipage}
\vspace*{-0.3cm}
    \caption{$c$ denotes column names; $const$ are values in the Table; $n$ is an integer; $op$ is user-provided.}
% \vspace*{-0.5cm}
\label{fig:dsl}
\end{figure}

In this section, we introduce the domain-specific language for visualization programs. As in  prior work~\cite{viser}, a visualization program in our setting first performs the necessary table transformations to obtain an intermediate table $T$ and then generates  a plot based on $T$. Hence,  as shown in Figure~\ref{fig:dsl}, a visualization program $\prog_v$ can be expressed as the composition of two programs $\ptable$ and $\pplot$, where $\ptable, \pplot$ are programs expressed in  the so-called \emph{table transformation} and \emph{plotting} DSLs, respectively. In the remainder of this section, we discuss the syntax and (informal) semantics of these two DSLs in more detail.

\paragraph{Plotting DSL}

A program in our plotting DSL takes in an input table $T$ and outputs a plot, which can be one of four types: (1) bar graph, (2) scatter plot, (3) line plot, or (4) area plot. Figure~\ref{fig:plot_example} shows an example of each type of visualization  supported by our plotting DSL. In more detail, a plotting program is  of the form $f(T, \colx, \coly, \colco, \colsub)$ where $f$ specifies the plot type,  $T$ is the input table, and the remaining arguments are attributes of $T$. Specifically, the $\colx$, $\coly$ columns specify the x- and y-axis of the plot and are required for every program in the plotting DSL. The remaining two arguments $\colco$ and $\colsub$ are optional and only make sense for plots with multiple layers or subplots (or both). In particular, the $\colco$ attribute is useful for plots that require multiple layers and specifies that each different color in the plot corresponds to a different value of the $\colco$ column. Finally, the optional fourth argument specifies that each distinct entry in the  $\colsub$ column should be used to generate a different subplot.

\begin{figure}
    \centering
    \small
    % \vspace*{-0.5cm}
\begin{minipage}[t]{0.24\textwidth}
    \centering
    \begin{equation*}
    {\tiny %
    \begin{split}
    \mathsf{Bar}(T, \colx=\texttt{Quarter}, \coly&=\texttt{Value},\\
    \colco &=\texttt{Type}, \\
    \colsub &=\texttt{Area})
    \end{split}
    }
    \end{equation*}
    \vspace*{0.1cm}
    \resizebox{0.9\textwidth}{!}{
\begin{tabular}{|c:c:c:c|}
    \hline
     {\bf Quarter} & {\bf Value} & {\bf Type} & {\bf Area}  \\
     \hline
     Q1 & 5 & T1 & Region A \\
     \hline
     Q1 & 6 & T2 & Region B \\ 
     \hline
     \ldots & \ldots & \ldots & \ldots\\
     \hline
\end{tabular}
}
    \vspace*{0.15cm}
    \includegraphics[width=0.9\textwidth, trim=40 520 300 50, clip]{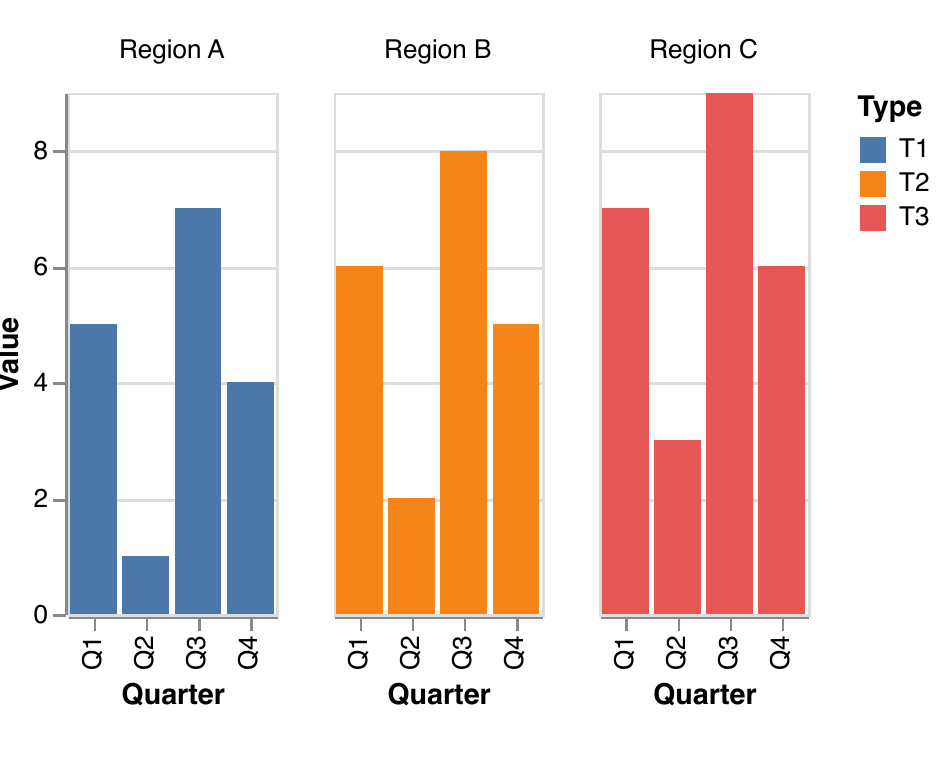}
\end{minipage}
\begin{minipage}[t]{0.24\textwidth}
    \centering
    \begin{equation*}
    {\tiny %
    \begin{split}
    \mathsf{Line}(&T, \colx=\texttt{Year}, \\
    &\coly=\texttt{Revenue},\\
    & \colco=\texttt{Type})
    \end{split}
    }
    \end{equation*}
    \vspace*{0.1cm}
    \resizebox{0.65\textwidth}{!}{
\begin{tabular}{|c:c:c|}
    \hline
     {\bf Year} & {\bf Revenue} & {\bf Type} \\
     \hline
     2000 & 2000 & A  \\
     \hline
     2001 & 1234 & A \\
     \hline
     \ldots & \ldots & \ldots\\
     \hline
\end{tabular}
}
    \vspace*{0.15cm}
    \includegraphics[width=0.8\textwidth, trim=40 560 350 30, clip]{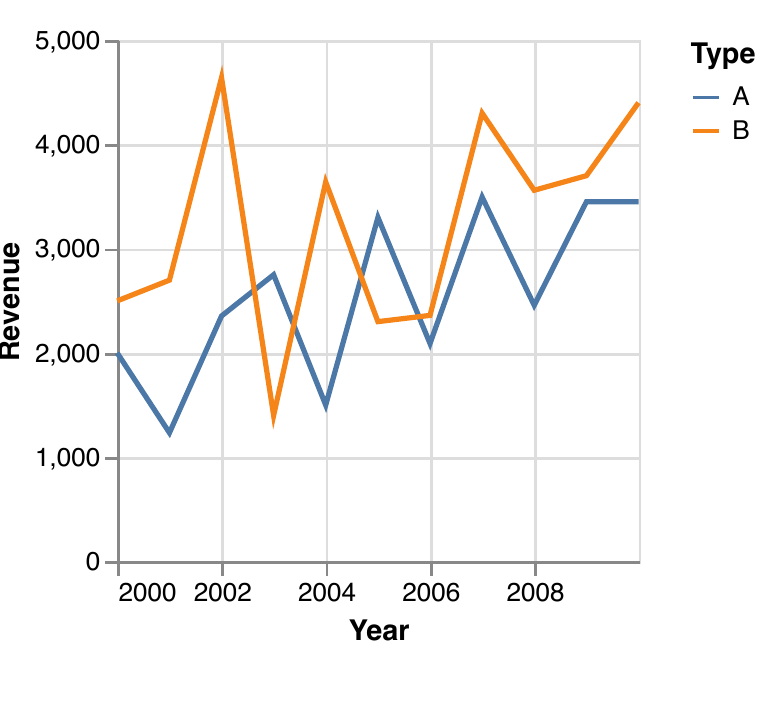}
\end{minipage}
\begin{minipage}[t]{0.24\textwidth}
    \centering
    \begin{equation*}
    {\tiny %
    \begin{split}
    \mathsf{Scatter}(&T, \colx=\texttt{Budget}, \\ &\coly=\texttt{Box\_Office}, \\
    &\colco=\texttt{Rating})
    \end{split}
    }
    \end{equation*}
    \vspace*{0.1cm}
    \resizebox{0.8\textwidth}{!}{
\begin{tabular}{|c:c:c|}
    \hline
     {\bf Box Office} & {\bf Budget} & {\bf Rating} \\
     \hline
     25 & 20 & R  \\
     \hline
     26 & 60 & R \\
     \hline
     \ldots & \ldots & \ldots\\
     \hline
\end{tabular}
}
    \vspace*{0.15cm}
    \includegraphics[width=0.75\textwidth, 
    trim=40 560 350 30, 
    % clip
    ]{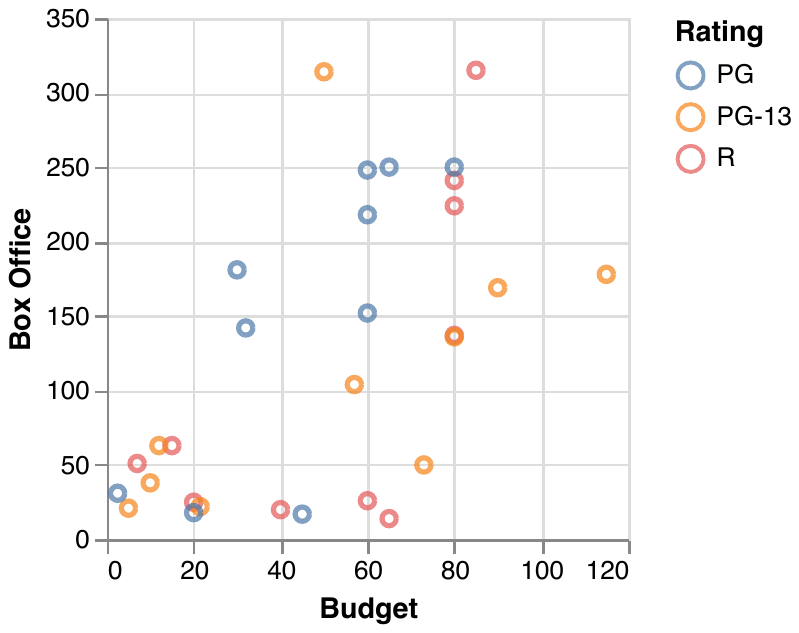}
\end{minipage}
\begin{minipage}[t]{0.24\textwidth}
    \centering
    \begin{equation*}
    {\tiny %
    \begin{split}
    \mathsf{Area}(&T, \colx=\texttt{Year}, \\ & \coly=\texttt{Profit}, \\ 
    & \colco=\texttt{Store})
    \end{split}
    }
    \end{equation*}
    \vspace*{0.1cm}
    \resizebox{0.6\textwidth}{!}{
\begin{tabular}{|c:c:c|}
    \hline
     {\bf Year} & {\bf Profit} & {\bf Store}   \\
     \hline
     2015 & 20 & A  \\
     \hline
     2015 & 10 & B \\ 
     \hline
     \ldots & \ldots & \ldots\\
     \hline
\end{tabular}
}
    \vspace*{0.15cm}
    \includegraphics[width=0.75\textwidth, trim=40 560 350 30, clip]{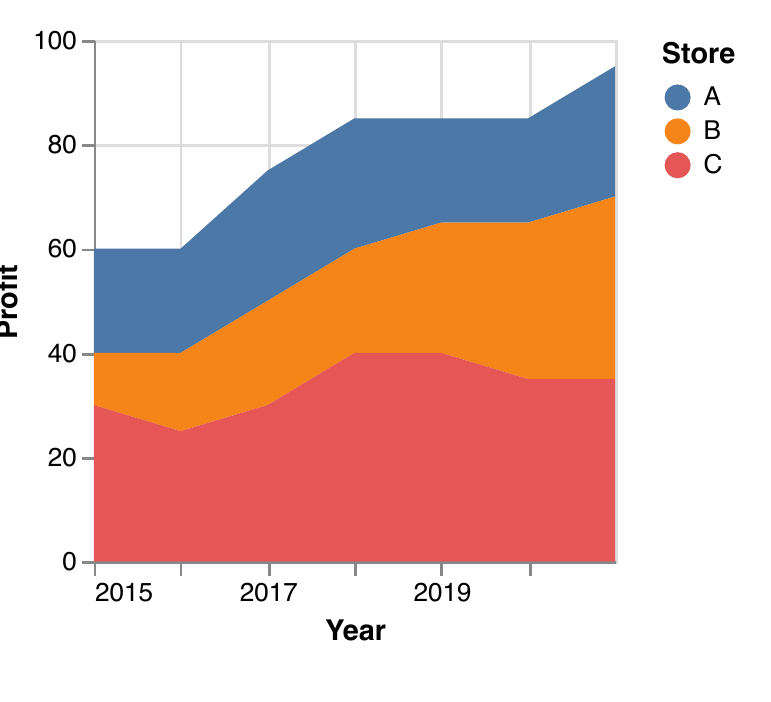}
\end{minipage}
    \vspace*{-0.5cm}
    \caption{Examples of plotting programs and their corresponding visualizations. }
    \label{fig:plot_example}
    % \vspace*{-0.5cm}
\end{figure}

\paragraph{Table transformation DSL} As shown in Figure~\ref{fig:dsl},  a table transformation  program takes in an input table {$T_{in}$}, and outputs a table $T$ by applying a sequence of transformations that are inspired by relational algebra and supported by many popular visualization languages, such as VegaLite and {\tt ggplot2}. In particular, our table transformation DSL includes the following useful constructs:

\begin{itemize}[leftmargin=*]
    \item The $\mathsf{bin}$ operation discretizes a numeric column in the table $\coltarg$ into a set of bins. Here, the argument $n$ specifies the number of bins that the entries in $\coltarg$ should be split into. For example, in the first input table shown in Figure ~\ref{fig:table_example}, column $c_2$ is binned.
    
    \item The $\mathsf{filter}$ construct corresponds to the standard selection operation in relational algebra. Given a table $T_{in}$ and predicate $\phi$ of the form  $val \ op \ val$, it produces a subset of $T_{in}$ consisting of all tuples  that satisfy $\phi$. The second illustration in Figure ~\ref{fig:table_example} offers an example of the filter operation.

    \item The $\mathsf{summarize}$ construct performs an aggregation operation specified by $\alpha$ on a specified column $\coltarg$. In more detail, given an input table $t$ and  "keys" (i.e., columns) $\colkey = [c_1, \ldots, c_k]$, it produces a new table that has columns $c_1, \ldots, c_k, \coltarg$ such that for each value of the tuple $(c_1, \ldots, c_k)$, the corresponding value of $\coltarg$ is obtained by applying the aggregation operator $\alpha$ to those entries that have the same value for $(c_1, \ldots, c_k)$. In the third illustration in Figure ~\ref{fig:table_example}, column $c_2$ is summarized by the \texttt{count} operator.
    
    \item The $\mathsf{mutate}$ construct produces a table that has one more column $\coltarg$ than its input table. In particular, the value stored in $\coltarg$ is obtained by applying operator $op$ to the corresponding values stored in columns $\overline{\colarg}$. In the fourth illustration in Figure ~\ref{fig:table_example}, the \texttt{mutate} operator creates column $c_3$ by taking the max of columns $c_1$ and $c_2$.
    
    \item The $\mathsf{select}$ construct corresponds to the standard projection operation in relational algebra. In particular, $\mathsf{select}(t, \overline{\colarg})$ yields a table containing only the columns $\overline{\colarg}$. 
\end{itemize}
Observe that the first argument of each operator is a term $e$ in the table transformation DSL; thus, these transformations can be arbitrarily nested within one another. Hence, the table transformation DSL allows performing non-trivial data wrangling tasks that require applying many different operations to the input table.

\begin{figure}
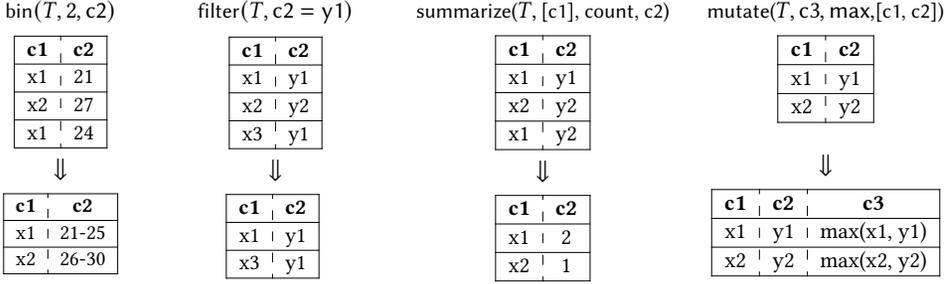

% \vspace*{-0.1cm}
    \centering
    \small
\begin{minipage}[t]{0.20\textwidth}
    \centering
    \begin{equation*}
    {\footnotesize %
    % \begin{split}
        \textsf{bin}(T, 2, {\sf c2})
    % \end{split}
    }
    % \vspace*{-0.2cm}
    \end{equation*}
    \resizebox{0.5\textwidth}{!}{
\begin{tabular}{|c:c|}
    \hline
     {\bf c1} & {\bf c2} \\
     \hline
     x1 & 21   \\
     \hline
     x2 & 27 \\
     \hline
     %\hline
    x1 & 24 \\
     \hline 
    % x2 & 29 \\ 
    %  \hline
\end{tabular}
} \\
    \vspace*{-0.2cm}
    \[
    \Downarrow
    \]
    \vspace*{-0.6cm}
    \resizebox{0.6\textwidth}{!}{
\begin{tabular}{|c:c|}
    \hline
     {\bf c1} & {\bf c2} \\
     \hline
     x1 & 21-25   \\
     \hline
     x2 & 26-30 \\
     \hline
     %\hline
     %x3 & [y2 - 5, y2 + 5) \\
     %\hline 
     %x4 & [y1 - 5, y2 + 5) \\ 
     %\hline
\end{tabular}
}
\end{minipage}
\begin{minipage}[t]{0.20\textwidth}
    \centering
    \begin{equation*}
    {\footnotesize %
    % \begin{split}
    \textsf{filter}(T, {\sf c2 = y1})
    % \end{split}
    }
    \end{equation*}
    % \vspace*{-0.2cm}
    \resizebox{0.5\textwidth}{!}{
\begin{tabular}{|c:c|}
    \hline
     {\bf c1} & {\bf c2} \\
     \hline
     x1 & y1   \\
     \hline
     x2 & y2 \\
     \hline
     x3 & y1 \\
     %\hline 
     %x4 & y3 \\ 
     \hline
\end{tabular}
} \\
    \vspace*{-0.2cm}
    \[
    \Downarrow
    \]
    % \vspace*{-0.5cm}
    \resizebox{0.5\textwidth}{!}{
\begin{tabular}{|c:c|}
    \hline
     {\bf c1} & {\bf c2} \\
     \hline
     x1 & y1   \\
     \hline 
     x3 & y1 \\ 
     \hline
\end{tabular}
}
\end{minipage}
\begin{minipage}[t]{0.30\textwidth}
    \centering
    \begin{equation*}
    {\footnotesize %
    % \begin{split}
      \textsf{summarize(} T, \textsf{[c1], } \textsf{count, c2)}
    % \end{split}
    }
    \end{equation*}
    % \vspace*{-0.2cm}
    \resizebox{0.34\textwidth}{!}{
\begin{tabular}{|c:c|}
    \hline
     {\bf c1} & {\bf c2} \\
     \hline
     x1 & y1   \\
     \hline
     x2 & y2 \\
     \hline
     x1 & y2 \\
     %\hline 
     %x2 & y3 \\ 
     \hline
\end{tabular}
} \\
    \vspace*{-0.2cm}
    \[
    \Downarrow
    \]
    \vspace*{-0.6cm}
    \resizebox{0.34\textwidth}{!}{
\begin{tabular}{|c:c|}
    \hline
     {\bf c1} & {\bf c2} \\
     \hline
     x1 & 2   \\
     \hline
     x2 & 1 \\
     \hline
\end{tabular}
}
\end{minipage}
\begin{minipage}[t]{0.23\textwidth}
    \centering
    \begin{equation*}
    {\footnotesize %
    % \begin{split}
    \textsf{mutate(} T, {\sf c3, max,} \textsf{[c1, c2])}
    % \end{split}
    }
    \end{equation*}
    \resizebox{0.45\textwidth}{!}{
\begin{tabular}{|c:c|}
    \hline
     {\bf c1} & {\bf c2} \\
     \hline
     x1 & y1   \\
     \hline
     x2 & y2 \\
     \hline
\end{tabular}
} \\
    \vspace*{0.08cm}
    \[
    \Downarrow
    \]
    \resizebox{1\textwidth}{!}{
\begin{tabular}{|c:c:c|}
    \hline
     {\bf c1} & {\bf c2} & {\bf c3} \\
     \hline
     x1 & y1 & max(x1, y1)   \\
     \hline
     x2 & y2 & max(x2, y2) \\
     \hline
\end{tabular}
}
\end{minipage}
\vspace{-0.3cm}
    \caption{Examples of table transformation programs. }
    \label{fig:table_example}
    \vspace{-0.5cm}
\end{figure}

% \begin{figure}
% \small
% \begin{minipage}[t]{0.45\textwidth}
% \textbf{Visualization DSL}
% \[
% \begin{array}{r l }
%     % \textbf{Visualization DSL} \\ 
%     P := & \lambda T_{in}. \ \mathsf{let\ } T = \ptable(T_{in}) \mathsf{\ in\ } \pplot \\ 
%     % P := & \lambda T. (\pplot \circ \ptable)(T) & \text{visualization program} \\ 
%     \\
% \end{array}
% \]
% \textbf{Sub-DSL for plotting} 
% \[
% \begin{array}{r l }   
%     % &\textbf{Sub-DSL for plotting} \\ 
%     \pplot := & f(T, \colx, \coly, \colco, \colsub) \\
%     f := & \mathsf{Bar} \mid \mathsf{Scatter} \mid \mathsf{Line} \mid \mathsf{Area} \\ 
% \end{array}
% \]
% \end{minipage}
% \begin{minipage}[t]{0.45\textwidth}
% \textbf{Sub-DSL for table transformations}
% \[
% \begin{array}{r l }
%     \ptable := & \lambda T. \ e  \\
%     e := & T   \\ 
%     | & \mathsf{bin} (e, n, \coltarg) \\
%     | & \mathsf{filter} (e, val_1 \ op\ val_2)  \\
%     | & \mathsf{summarize}(e, \overline{\colkey}, \alpha, \coltarg) \\ 
%     | & \mathsf{mutate}(e, \coltarg, op, \overline{\colarg}) \\
%     | & \mathsf{select}(e, \overline{\colarg}) \\
%     val := & const \mid c \\
%     \alpha := & \mathsf{mean} \mid \mathsf{sum} \mid \mathsf{count} \\
%     \end{array}
%  \]
% \end{minipage}
%     \caption{Language Syntax. $c$ represents the column name; $const$ are values in the Table; $n$ is an integer; $op$ is provided by the user.}
% \vspace*{-0.5cm}
% \label{fig:dsl}
% \end{figure}
\section{Overview of Refinement Type System}\label{sec:type}

While the visualization DSL introduced in Section~\ref{sec:dsl} does not have explicit type annotations, our approach leverages a   refinement type system that facilitates effective synthesis. The design of our type system is based on two pragmatic considerations: first, we want our refinement types to serve as useful specifications, meaning that they should capture the clues that are commonly found in natural language descriptions of visualization tasks. Second, we want our type system to be useful for pruning infeasible parts of the search space during synthesis. 
With these considerations in mind, we introduce those aspects of our refinement type system that are necessary for understanding the overall synthesis approach.

\subsection{Type Syntax } \label{sec:syntax}

As standard~\cite{liquid}, a refinement type is of the form 
$\{\nu: \btsym \ | \ \phi\}$ where $\btsym$ is a base type and $\phi$ is a logical qualifier. As shown in Figure~\ref{fig:type},  base types include strings, integers, four different types of plots, and tables. A table type $\mathsf{Table}(\sigma)$ denotes a table with schema $\sigma$, which maps each column name (attribute) to its column type, which indicates the type of values stored under that column. The column type $\top$ indicates \emph{any} type of of data, whereas Quantitative and Qualitative indicate whether the entry is associated with a quantity or quality respectively.  Quantitative data can be further divided into Continuous and Discrete, and Qualitative data can be divided into Nominal, Ordinal, and Temporal (e.g., date or year).

In contrast to base types, logical qualifiers are formulas formed from atomic predicates using the standard logical connectives $\neg, \land$, and $\lor$. We differentiate between two types of atomic predicates, namely \emph{{syntactic constraints}}  $\pi(..)$ and \emph{table property predicates} of the form $\term \between \term$. We discuss both types of atomic predicates in more detail below.

\paragraph{{{\bf Syntactic constraints.}}} Given a table or plot $x$ with attribute $c$, the predicate $\prov(x.c, \provop)$ expresses that $\provop$ was used in the derivation of $x.c$. Here, $\provop$ is either a  built-in function $f$ in our DSL (e.g., \textsf{count, mutate}) indicating that function $f$ was involved in the computation of $x.c$, or a term of the form $x'.c'$ indicating data flow from $x'.c'$ to $x.c$.   Intuitively, the {syntactic constraints} in our type system allow encoding useful hints present in the natural language description about the origin of the data used in the  visualization task. 

% \begin{example}
% Using the motivating example we shown in Section~\ref{sec:intro}, suppose the user specifies that the column ``Origin'' should be used as the subplot encoding, then we can derive  the provenance constraint $\prov(\nu.{\tt subplot}, x.{\tt Origin})$, which means the subplot encoding of variable $\nu$ (recall the base type is a plot type) is associated with the {\tt Origin} column from table represented by the variable $x$. 
% \end{example}

\begin{example}
For the running example from Section~\ref{sec:overview}, our NL parser generates the following type for the output of the plotting program:
\small
\[
\{ \nu:\mathsf{BarPlot} \ | \ \pi(\nu.{\sf color}, x.{\sf Origin}) \}
\]
\normalsize
where $x$ refers to the input of the plotting program. Here, the base type indicates that we want a bar graph, and the {syntactic constraint} indicates that the color encoding is bound to the Origin field of the input table. In other words, it indicates that colors in the plot correspond to values of the Origin column.

\end{example}

\begin{figure}[!t]
% \centering
% \vspace*{-0.3cm}
\small
\begin{minipage}[t]{0.45\textwidth}
\ \ \ \ \ \ \ \ \ \textbf{Base Type}
\[
\begin{array}{r l l}
    
    % \textbf{Base Type} \\ 
    \btsym := & \tablet(\sigma) \mid \plott \mid \mathsf{str} \mid \mathsf{int} \\
    \plott := & \mathsf{BarPlot} \mid \mathsf{ScatterPlot} \\
    \mid & \mathsf{LinePlot} \mid \mathsf{AreaPlot} \\
    \sigma := & \{c_1: \tau_{c_1}, ..., c_n: \tau_{c_n}\} \\
    \columnt := & \top \mid \mathsf{Qualitative} \mid \mathsf{Quantitative} \\
    | & \mathsf{Nominal} \mid \mathsf{Ordinal} \mid \mathsf{Temporal} \\
    | & \mathsf{Discrete} \mid \mathsf{Continuous} \\
    \\
\end{array}
\]
\end{minipage}
\begin{minipage}[t]{0.45\textwidth}
\textbf{Refinement Type}
\[
 \begin{array}{r l l}  
    % \textbf{Refinement Type}\\
    \rtsym := & \rtype{\btsym}{\phi} \mid x: \rtsym \rightarrow \rtsym \\
    \phi := & \prov(x.\enc, \provop) \mid  \term \ \between \ \term \ {\rm where} \     \between  \in \{=, \geq, \leq\} \\
    | & \neg \phi \mid \phi \wedge \phi \mid \phi \vee \phi \\
    \enc := & \mathsf{x} \mid \mathsf{y} \mid \mathsf{color} \mid \mathsf{subplot} \mid c \\ 
    \provop := & \mathsf{mean} \mid \mathsf{sum} \mid \mathsf{count} \mid \mathsf{bin} \mid \mathsf{filter} \mid \mathsf{mutate} \mid x.c \\
    \term := &  |\tablesym| \mid \aggr(\tablesym) \mid n \mid c \mid x \\
    \tablesym := & x \mid \mathsf{Proj}(\tablesym, \overline{c}) \mid \mathsf{Filter}(\tablesym, val_1 \ op \ val_2) \\
    \aggr := & \mathsf{max} \mid \mathsf{min} &
\end{array}
\]
\end{minipage}
    \vspace*{-0.5cm}
    \caption{Type Syntax. $c$ is a column name; $const$ are values in the Table;  $n$ is an integer; $x$ is a variable.}
    \label{fig:type}
    \vspace*{-0.3cm}
\end{figure}

\paragraph{{\bf Table properties.}} In addition to the {syntactic} requirements, our type system  allows expressing properties of tables using predicates of the form $\term \between \term$ where $\term$ is a term and $\between$ is a relation symbol (e.g., $\leq$).  In more detail, terms $\theta$ can be formed using the following constructs:

\begin{itemize}[leftmargin=*]
    \item Given a variable $x$ of type \textsf{Table}, $|x|$ represents the cardinality of $x$ (i.e., number of unique tuples).
    \item The functions \textsf{Proj} and \textsf{Filter} have the same semantics as the corresponding constructs in our table transformation DSL.
    \item Given a column $x$, the aggregation operators $\mathsf{max}(x)$ and $\mathsf{min}(x)$ return the maximum (resp. minimum) value in $x$.
\end{itemize}

Intuitively, table property predicates are useful for specifying the table transformation component of the visualization task and provide significant pruning power during synthesis.

\begin{example}
Consider the following refinement type:
\vspace{-0.05cm}
\small
\[
\rtype{\texttt{Table}({\sf Price} : {\sf Discrete}, {\sf Origin} : {\sf Nominal} )}{|\nu| = 3 \wedge {\sf max}({\sf Proj}(\nu, \{{\sf Price}\})) = 8}
\]
\normalsize
This type describes a table that (1) has two attributes, Price and Origin, of types Discrete and Nominal respectively, (2) contains three unique tuples, and (3) has a maximum value of 8 in its Price column.
\end{example}

\subsection{Subtyping} 
Given a refinement type specification $\rtsym$, the goal of our approach is to synthesize a visualization program of type $\rtsym'$  such that $\rtsym'$ is a subtype of $\rtsym$. Thus, we start by formalizing the subtyping relation for our type system using judgments of the following form:
\[
\small
\subty{\Env}{\rtsym_1}{\rtsym_2}
\]
where $\Env$ is a type environment mapping variables (and built-in DSL functions) to their corresponding types. As standard, the meaning of this judgment is that  $\rtsym_1$ is a subtype of $\rtsym_2$ under type environment $\Env$.  Since deciding subtyping  between base types does not require the type environment, we omit the type environment for base types.

\begin{figure}[!t]
    % \vspace*{-0.1cm}
    \centering
    \small
    \[
    \begin{array}{llll}
        & \vdash \sf{Quantitative} \subtype \top \ \ \ \ 
        & \vdash \sf{Qualitative} \subtype \top \\
        & \vdash \sf{Continuous}  \subtype \sf{Quantitative} \ \ \ \ 
        & \vdash \sf{Discrete} \subtype \sf{Quantitative} \\
        & \vdash \sf{Nominal}  \subtype \sf{Qualitative} \ \ \ \
        & \vdash \sf{Ordinal}  \subtype \sf{Qualitative} \ \ \ \
        & \vdash \sf{Temporal}  \subtype \sf{Qualitative} \\ 
    \end{array}
    \]
    \begin{mathpar}
    \inferrule*[Left=Base-Trans]{\vdash \btsym'' \subtype \btsym'  \ \ \ \vdash \btsym ' \subtype \btsym}{\vdash \btsym'' \subtype \btsym} \and
    \inferrule*[Left=Base-Ref]{}{\vdash \rtype{\btsym}{\phi}  \subtype \btsym}
    \\
    \inferrule*[Left=Table-Width]{}{\vdash \tablet(\{c_i:{\btsym_i}^{ \ i \in 1...n+k}\}) \subtype \tablet(\{c_i:{\btsym_i}^{ \ i \ \in 1...n}\})}
    \\
    \inferrule*[Left=Table-Permutation]{\vdash \tablet(\{c_i:{\btsym_i}^{ \ i \in 1...n}\}) \ \text{ is a permutation of } \tablet(\{c'_i:{\btsym_i}^{ \ i \ \in 1...n}\})}{\vdash \tablet(\{c_i:{\btsym_i}^{ \ i \in 1...n}\}) \subtype \tablet(\{c'_i:{\btsym_i}^{ \ i \ \in 1...n}\})}
    \\
    \inferrule*[Left=Table-Depth]{\forall i. \vdash \btsym_i \subtype \btsym_i'}{\vdash \tablet(\{c_i : \btsym_i^{ \ i \in 1 ... n}\}) \subtype \tablet(\{c_i : \btsym_{i}'^{ \ i \in 1 ... n}\})} \\ 
    \inferrule*[Left=Ref]{
    \vdash \btsym_1 \subtype \btsym_2 \\
    \mathsf{Encode}(\Env) \wedge \mathsf{Encode}(\phi_1) \Rightarrow \mathsf{Encode}(\phi_2) %\ \ \ \ \mathsf{Valid}
    }{\Env \vdash {\rtype{\btsym_1}{\phi_1}} \subtype {\rtype{\btsym_2} {\phi_2}}} \and
     \inferrule*[Left=Func]{\Env \vdash \rtsym_1' \subtype \rtsym_1  \ \ \ \Env \vdash \rtsym_2 \subtype \rtsym_2'}{\Env \vdash x: \ftype{\rtsym_1}{\rtsym_2} \subtype x: \ftype{\rtsym_1'}{\rtsym_2'}}
    \end{mathpar}
    \vspace*{-0.5cm}
    \caption{Base and refinement type subtyping relation.}
    \label{fig:subtyping}
    \vspace*{-0.5cm}
\end{figure}

 Figure~\ref{fig:subtyping} presents our subtyping rules. The first several rules are straightforward and show the subtyping relation between primitive types like \textsf{Discrete} and \textsf{Quantitative}. The subtyping rules for tables are essentially standard subtyping rules for records~\cite{tapl}. The last two rules for refinement types are also standard and require (1) checking the subtyping relation between base types and (2) checking the validity of a logical formula for the logical qualifiers. In particular, these rules make use of a function called \textsf{Encode} that converts the logical qualifier of a refinement type into an SMT formula. The interested reader can find details of the SMT encoding in the appendix.
 %To make our subtyping check decidable, we follow prior work \cite{liquid} and encode qualifiers as formulas in the combined theory of equality with uninterpreted functions and integers. As such, our subtyping checks are sound but not complete.

\subsection{Type compatibility} \label{sec:incompatibility}

\begin{figure}[t]
    \centering
    % \vspace*{-0.5cm}
    \small
    \begin{mathpar}
    \inferrule*[Left=Symmetry]{\vdash \btsym \sim \btsym'}{\vdash \btsym' \sim \btsym} \and
    \inferrule*[Left=Data]{\btsym, \btsym' \in \tau_c \ \ \ \ \vdash \btsym \subtype \btsym' \ \lor \ \vdash \btsym' \subtype \btsym }{\vdash \btsym \sim \btsym'} \\
    \inferrule*[Left=Table]{\forall i, j. (c_i = c_j')  \to \  \vdash (\btsym_i \sim \btsym_j')}{\vdash \tablet(\{c_i : \btsym_{i}^{ \ i \in 1 ... n}\}) \sim \tablet(\{c_j' : \btsym_{j}'^{ \ j \in 1 ... m}\}) } \and
    \inferrule*[Left=Func]{\Env \vdash \rtsym_1 \sim \rtsym_1' \ \ \ \  \  \vdash \rtsym_2 \sim \rtsym_2'}{\Env \vdash x : \ftype{\rtsym_1}{\rtsym_2} \sim x : \ftype{\rtsym_1'}{\rtsym_2'}} \\
    \inferrule*[Left=Refinement-Comp]{
    \vdash \btsym_1 \sim \btsym_2 \\\\
    \mathsf{SAT}(\mathsf{Encode}(\Env) \wedge \mathsf{Encode}(\phi_1) \wedge \mathsf{Encode}(\phi_2)) }{\Env \compJudg{\rtype{\btsym_1}{\phi_1}}{\rtype{\btsym_2}{\phi_2}}}
    \end{mathpar}
    \vspace{-0.5cm}
    \caption{Base and refinement type compatibility relation}
    \label{fig:incomp}
    \vspace{-0.5cm}
\end{figure}

While our synthesis algorithm ensures that the type of the synthesized program is a subtype of the specification, we utilize a weaker notion of \emph{type compatibility} for pruning during synthesis. In particular, because our synthesis algorithm needs to reason about the feasibility of incomplete programs (where some parts are yet to be determined), we  introduce a notion of type compatibility that is much weaker than subtyping. Intuitively, two types $\rtsym_1$ and $\rtsym_2$ are \emph{compatible} with each other if there exists a subtype $\rtsym$ of $\rtsym_1$ that is also a subtype of $\rtsym_2$. Conversely, if two types $\rtsym_1$ and $\rtsym_2$ are incompatible, there is no refinement of $\rtsym_1$ that will make it a subtype of $\rtsym_2$. As we will see in Section~\ref{sec:synthesis}, the notion of type (in)compatibility is very useful for pruning during synthesis. In this section, we formalize this notion and present rules for checking type compatibility.
We define the compatibility relation for our type system using judgments of the form:
\small
\[
% \vspace*{-0.25cm}
\Env \vdash \rtsym_1 \comp \rtsym_2
% \vspace*{-0.25cm}
\]
\normalsize
stating that $\rtsym_1$ is compatible with $\rtsym_2$ under environment $\Env$, as shown in Figure~\ref{fig:incomp}. Unlike the subtyping relation, the compatibility relation is  symmetric (first rule in Figure~\ref{fig:incomp}); however it is \emph{not} transitive. The second rule in Figure~\ref{fig:incomp} defines the type compatibility relation for primitive types and states that they are compatible if one is a  subtype of the other or vice versa. The third rule (\textsc{Table}) asserts that two table types are compatible when all their shared columns are compatible. The intuition is that if all  shared columns  are type compatible, then we can construct a new table type that is a refinement of both by taking the union of their schemas. Finally, \textsc{Refinement-Comp} and \textsc{Func} are similar to their subtyping counterparts in that they reduce the compatibility check to an SMT query. However, there are two key differences. First, the encoded formula is a conjunction of the qualifiers as opposed to an implication. Second, we check that the encoding is satisfiable as opposed to valid.  The intuition behind this rule is that, if the resulting formula is satisfiable, then  $\{\nu : \rtsym \ | \ \phi_1 \land \phi_2\}$ is a well defined type in our type system that has at least one inhabitant, and it refines both $\{\nu : \rtsym \ | \ \phi_1\}$ and $\{\nu : \rtsym \ | \ \phi_2\}$.

\subsection{Typing Rules} \label{sec:typingrules}

In this section, we give an overview of our typing rules for assigning types to DSL terms. In particular, our typing rules derive judgments of the  form 
$
\small
\Env \vdash t: \rtsym
$
to indicate that term $t$ has type $\rtsym$ under environment $\Env$. Since the typing rules are not the primary focus of this paper, we only discuss two representative rules and leave the rest to the appendix.

\begin{figure}[!t]
    \centering
    \small
    % \vspace{-0.5cm}
    \begin{mathpar}
    \inferrule*[Left=Bar]{
    \Env(T) = \rtype{\tau_T}{\phi_T}\\\\
    \tau_T  = \mathsf{Table}(\{ \colx: \sf{Discrete}, \coly: \sf{Quantitative}, \colco: \sf{Discrete}, \colsub: \sf{Discrete} \})\\\\
    \textsf{Encode}(\Env) \land \textsf{Encode}(\phi_T) \Rightarrow |(\nu, \{\colx, \colco, \colsub\})| \geq |(\nu, \{\coly\})|
    }{\tjudg{\Env}{\textsf{Bar}(T, \colx, \coly, \colco, \colsub): \rtype{ \textsf{BarPlot}}{\bigwedge_{e \in \{{\sf x}, {\sf y}, {\sf color}, {\sf subplot}\} } \prov(\nu.e, T.c_{e})}}} \\
    \end{mathpar}
    \vspace*{-1.2cm}
    \caption{Typing Rule for a Bar Plot. We use notation $(\nu, \{c_1, \dots, c_n\})$ as a shorthand for ${\sf Proj}(\nu, \{c_1, \dots, c_n\})$. }
    \label{fig:rulebar}
    % \vspace*{-0.1cm}
    \end{figure}

\paragraph{Typing rules for the plotting sub-DSL} To illustrate the typing rules for the plotting sub-DSL, Figure~\ref{fig:rulebar} shows the rule for the \textsf{Bar} construct, which generates a bar graph given table $T$. At a high level, this rule states that if $T$'s type satisfies two constraints, then the output type will be a refinement of {\sf BarPlot}. The first constraint is that $T$'s schema must be suitable for generating bar graphs, meaning that $\colx$ is {\sf Discrete} and  $\coly$ is {\sf Quantitative}. This requirement is captured by the second premise. In addition to having a suitable schema, another important requirement for a bar graph is that it should not have overlapping bars, meaning that the x-label in each subplot must correspond to a unique y-value. This requirement is captured through the cardinality constraint in the third premise, which checks that the logical qualifier for $T$ implies that there is unique $y$ for each $x$. If these premises hold, then the entire term is well-typed with base type \textsf{BarPlot} and a logical qualifier stating the {syntactic constraint} for the return value of \textsf{Bar}. In particular, the logical qualifier in the conclusion states, for example, that the \textsf{x} attribute of the plot is derived from the $\colx$ attribute of the input table $T$.

    \begin{figure}[!t]
    % \vspace*{-0.5cm}
    \centering
    \small
    \begin{mathpar}
    \inferrule*[Left=Summ-Mean]{
    \Env \vdash e: \rtype{\tau_t}{\phi} \ \ \ \ \text{where $\tau_t = \tt{Table}(\{\ldots, \coltarg: \tau_{\tt tgt}, \ldots\})$} \\\\
    \coltarg \not\in \overline{\colkey} \ \ \  \vdash \tau_{\sf tgt} : \sf{Quantitative} \\\\
    \tau' = {\sf Table}(\{c_0': \tau_0', \ldots, c_k': \tau_k', \coltarg : \btsym_{\tt tgt} \}) \ \ c_i' \in \overline{\colkey}  \ \ \ \ \ \ \tau' = \tau'[\coltarg \mapsto \sf{Continuous}] \\\\
    \phi_1 = \phi \ \forget \ \textsf{Terms}(\phi, \coltarg) \ \ \ \ \ \ \ \phi_2 = \phi_1 \  \forget \ \prov(\nu.\coltarg, \mathsf{mean}) \\\\
    \phi' = \phi_2 \land |(\nu, \{\coltarg\})| \leq |(\nu, \overline{\colkey})| \land \prov(\nu.\coltarg, \mathsf{mean})
    }{\tjudg{\Env}{\textsf{summarize}(e, \overline{\colkey}, \mathsf{mean}, \coltarg)} : \rtype{\tau'}{\phi'}}\\\\
    
    \end{mathpar}
    \vspace*{-1.5cm}
    \caption{Typing rule for Summarize instantiated with a Mean operation. We use notation $(\nu, \{c_1, \dots, c_n\})$ as a shorthand for ${\sf Proj}(\nu, \{c_1, \dots, c_n\})$, and the $\forget$ operator is defined in the text.}
    \label{fig:rulemean}
    \vspace*{-0.5cm}
\end{figure}

\paragraph{Typing rules for table transformation sub-DSL} Figure~\ref{fig:rulemean} shows the typing rule for the \textsf{summarize} construct in our table transformation DSL. Recall that \textsf{summarize} takes as input an aggregation operator, and the type depends on  which aggregation operator \textsf{summarize} is invoked with. In Figure~\ref{fig:rulemean}, we consider instantiating \textsf{summarize} with \textsf{mean} as a representative example. To understand this typing rule, let us first recall the semantics of \textsf{summarize}, which associates each unique value of the specified key columns with the mean of the values in the specified target column (see Figure~\ref{fig:table_example}). The first two premises in the typing rule {\sc Summ-Mean} impose some requirements on the input table. In particular, because it only makes sense to take the mean of quantitative values, the second premise ensures that the target column has a suitable type. Furthermore, since the mean operation produces a value of type \textsf{Continuous}, the column $c_{tgt}$ has type \textsf{Continuous} in the output table with base type $\tau'$. The fourth and fifth lines in Figure~\ref{fig:rulemean} state the relationship between the logical qualifiers of the input and output tables. To that end, given a logical qualifier $\phi$ and a set of terms $S$, we use the notation $\phi \forget S $ to denote the strongest logical qualifier $\phi'$ that is implied by $\phi$ and that does not imply anything about any term $t \in S$.~\footnote{One way to obtain $\phi \forget t$ is to replace all occurrences of $t$ with a fresh existentially quantified variable $x$ and then eliminate the quantifier. We formalize the $\forget$ operator in the appendix.} Thus, according to our typing rule, the new logical qualifier $\phi'$ for the output table differs from the qualifier $\phi$ for the input table in the following ways: First, it "removes" from $\phi$ any knowledge about the terms that involve $\coltarg$ which are affected by the \textsf{summarize} operation. Second, it asserts that the number of unique tuples over $\overline{c_{key}}$ is greater than or equal to the number of unique values in $\coltarg$. This is because the cardinality of the output table is equal to the number of unique $(c_1, \ldots, c_k)$ values where each $c_i \in \overline{c_{keys}}$. However, as two distinct $(c_1, \dots, c_k)$ values could have the same value for $\coltarg$, we cannot infer a stronger constraint. Finally, since the values of $\coltarg$ were produced by the $\textsf{mean}$ operation, $\phi'$ includes the {syntactic constraint} $\pi(\nu.\coltarg, {\sf mean})$.

\begin{comment}
\subsection{Type Safety}

We prove the safety of our type system in a standard way by proving standard Progress and Preservation \cite{tapl} theorems for closed terms in our DSL. We write $t \eval t'$ to indicate that $t$ evaluates to $t'$.

\begin{proposition}[{\bf Progress}]
If $\ \emptyset \vdash t : \rtsym$, then either $t$ is a value or there is some term $t'$ such that $t \eval t'$.
\end{proposition}

\begin{proposition}[{\bf Preservation}]
If $\ \emptyset \vdash t : \rtsym$ and $t \eval t'$, then $\emptyset \vdash t' : \rtsym$.
\end{proposition}

Jocelyn TODO: (1)we need to define progress and preservation here. (2) show that our decidable system is conservative (overapproximate) 
\end{comment}
% \input{problem_statement}
\section{From Natural Language to Refinement Types}\label{sec:parse}

In this section, we describe a technique for generating refinement type specifications from natural language queries. At a high level, we frame this problem as an instance of the intents-and-slots problem \cite{jeong-lee-2006-exploiting,tur-et-al-intents} and build a parser that combines intent detection and slot filling on top of the BERT language model~\cite{bert}. %The benefit of having such an architecture is that we only have to focus on the specific visualization and data transformation properties mentioned in the natural language and do not require the natural language to fully specify the plot. 

\subsection{Background on Intents-and-Slots-Paradigm}

The intents-and-slots paradigm is a classical paradigm in the NLP literature on task-oriented dialog systems \cite{hemphill-etal-1990-atis,dahl-etal-1994-expanding,tur-et-al-intents} and flexibly supports many types of user interactions, as evidenced by its adoption on Amazon Alexa and other dialog platforms. At a high level, \emph{intent classification} is the problem of determining the topic of a query from a natural language utterance. For instance, given a set of topics such as ``flights", "movies", "restaurants", intent classification can be used to determine which of these topics a sentence is about. Once the intent of the utterance is identified, \emph{slot filling} determines pre-defined properties of that topic. For example, if the topic of a query is ``flights", relevant parameters include airline, destination city, flight number, etc., and slot-filling techniques aim to identify these parameters. 

As a concrete example, consider the query \emph{"What flights are available from San Francisco to New York?"} Here, an intent classifier aims to determine that the topic of the query belongs to the \emph{flight} category as opposed to \emph{movies} or \emph{restaurants}. Then, assuming that the flight category has attributes such as departure and destination city, a slot-filling technique can be used to determine that the departure city of the query is San Francisco and that the destination is New York.

The intents-and-slots paradigm is a good fit for our setting for two main reasons. First, compared to conventional semantic parsing \cite{zelle:aaai96,zettlemoyer-collins-05}, the intents-and-slots framework does not make strict assumptions about the grammatical structure of inputs. As a result, it can more flexibly handle user inputs that do not conform to a pre-defined syntax. Second, user queries in our setting can be naturally classified into different intents based on (1) the type of the plot (e.g., bar graph, scatter plot, etc.) they refer to and (2) which predicates in our refinement type system they involve. Furthermore, the arguments of these predicates can be determined using the slot-filling paradigm.

%Its ability to represent a core user goal (the plot) and arguments associated with that goal makes it a natural choice for this kind of interaction. Moreover, it enables us to leverage pre-trained language models, which recent work has shown are effective for intent detection and slot filling \cite{hardalov-et-al}.

%Compared to conventional semantic parsing \cite{zelle:aaai96,zettlemoyer-collins-05}, the intents-and-slots framework does not make strict assumptions about the grammatical structure of inputs. As a result, it can more flexibly handle user inputs that do not conform to English syntax. Recent work has explored purely sequence-to-sequence approaches \cite{jia-liang-2016-data, ncnet, lewis-etal-2020-bart}, including GPT-3 \cite{gpt3}, but these models are not guaranteed to produce well-formed output. With intents and slots, the output space already accommodates the task and does not require complex inferential constraints to yield valid output \cite{shin-constrained}.

\subsection{Parsing Technique}

\begin{figure}
    \centering
    % \vspace*{-0.5cm}
    \includegraphics[width=1.0\textwidth, trim=0 0 0 400, clip]{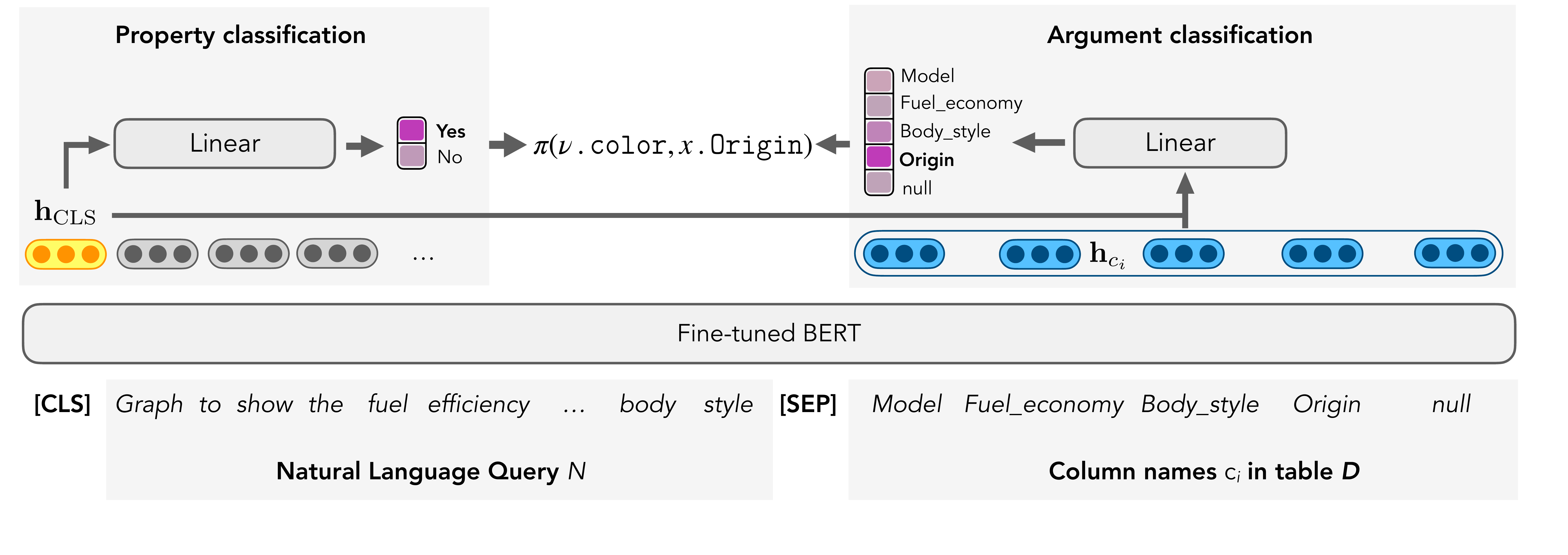}
    \vspace*{-1.0cm}
    \caption{Model architecture for the color encoding property in the running example. The yellow box represents the contextualized embedding of the query. Blue boxes represent the BERT embedding for each of the columns names in the table. After applying a linear layer to the BERT encoding in each task, we obtain a probability distribution across all possibles classes. The one with the highest probability is highlighted with pink.}
    \label{fig:parser}
    \vspace*{-0.5cm}
\end{figure}

In this section, we explain our  instantiation of the intents-and-slots paradigm for our setting.

\paragraph{Overview.} We have identified six types of properties that are typically mentioned in natural language queries and that are useful to the synthesizer. These include the following:
\begin{itemize}[leftmargin=*]
    \item {\bf Plot type:} According to the query, what is the most likely plot type desired by the user?
    \item {\bf Color:} Does this query mention anything about the color encoding of the plot?
    \item {\bf Subplot:} Does the query indicate that the visualization has subplots?
    \item {\bf Mean:} Does the query indicate that the visualization requires computing the mean of  values?
    \item {\bf Sum:} According to the query, does the visualization require summing values?
    \item {\bf Count:} Does the query indicate that the visualization involves the use of the {\tt count} operator?
\end{itemize}

Our parser uses  six different intent classifiers, one for each category listed above. With the exception of plot type, all intent classifiers are binary and yield a yes/no prediction. In the case of a ``yes'' prediction,  our parser uses slot filling to predict which attribute in the source data set this operation is associated with. For plot types, the intent classifier predicts whether the query refers to a bar chart, scatter plot, line graph, or area plot. 

\paragraph{Input to BERT} Our method performs the predictions outlined above using a fine-tuned BERT model, with a shared encoder for both intent classification and slot filling as illustrated in Figure~\ref{fig:parser}. The input to  BERT  is of the following form:
\small
\[
\mathrm{[CLS]} \ \nl \ \mathrm{[SEP]} \  c_1,\ldots,c_n
\]
\normalsize
where $N$ is the natural language query, $c_1,\ldots,c_n$ denote  column names from the input table $\data$, and $\mathrm{[CLS]}$ and $\mathrm{[SEP]}$ are standard placeholder tokens. We include the column names as part of the input for two reasons. First, when performing slot filling, the model needs to predict attributes of the source table, so we will need access to embedded representations of the column names. Second, even for intent classification,  information about the input table provides useful context that the BERT model can condition on. Given this input, BERT generates a \textit{contextualized encoding} of each of its input tokens, where each token in the input is mapped to a dense vector representation informed by all the other tokens through BERT's attention mechanism \cite{vaswani-et-al-2017}. 

\paragraph{Intent Classification} Our intent classifier is a model of the form $p(C_i \mid \nl, A_\data; \mathbf{w}_{c_i})$, where $A_\data$ is the set of column names for input table $D$, and $C_i \in \{0,1\}$ is a binary label for classifiers other than plot type, and $C_i \in \{{\sf Bar}, {\sf Scatter}, {\sf Line}, {\sf Area}\}$ for the plot type classifier. Additionally, $\nl$ denotes the NL input, $i$ is an index denoting the property type, and $w_i$ are the model weights. As standard practice, we take the vector $\mathbf{h}_\mathrm{CLS}$ as the representation of the sentence, and we use $p(C_i \mid \nl, A_\data, i; w_i) = \sigma(\mathbf{w}_{c_i}^\top \mathbf{h}_{\mathrm{CLS}}(\nl, A_\data))$, where $\sigma$ denotes the logistic function, making this a standard logistic regression layer with the CLS token's contextualized embedding as input.

\paragraph{Slot-filling model} As shown in Figure~\ref{fig:parser}, a property like {\sf color} requires a parameter in the form of one of the column names. Crucially, such arguments cannot be predicted with a standard classification model since the column names change with each table $\data$ being plotted. As such, the model must be able to place a distribution over an arbitrary set of column tokens. To address this issue, we use a pointer mechanism similar to implementations of the attention mechanism in settings like machine translation \cite{bahdanau-et-al-2015} and document summarization \cite{see-etal-2017-get}. Specifically, our slot-filling model places a distribution $p(c_i \mid \nl, A_\data, i; w_i)$ over column names. We use the same BERT encodings as the intent classifiers and let
%, which crucially provide a representation of both the user's intent ($\mathbf{h}_\mathrm{CLS}$) as well as the column names ($\mathbf{h}_{c_i}$).\footnote{For column names that are tokenized as two words by the BERT tokenizer, we do XXX FILL IN...} We therefore let
\small
\[
p(c_i \mid \nl, A_\data, i; \mathbf{W}) =  \mathrm{softmax}_i(\mathbf{h}_\mathrm{CLS}(\nl, A_\data)^\top  \mathbf{W}  \mathbf{h}_{c_i}(\nl, A_\data))
\]
\normalsize
where $\mathbf{W}$ is a square weight matrix and softmax denotes the standard softmax operation, which exponentiates and normalizes the arguments to form a probability distribution.
% \todo{wrote h as a fcn of N, D}
\vspace*{-0.1cm}
\paragraph{From BERT predictions to specifications.} Recall that the input to our synthesizer is a pair of refinement types of the form $(\rtsym_p, \rtsym_t)$ where $\rtsym_p$ is the output type of the plotting program and $\rtsym_t$ is the output type of the table transformation program. We now explain how to map the predictions made by the BERT model to specifications of this form. 

To generate $\rtsym_p$ for the plotting program, we assign the prediction of the plot type classifier to be the base type of $\rtsym_p$. The qualifier of $\rtsym_p$ consists of a conjunction of {syntactic constraints} output by the {\sf color} and {\sf subplot} models. In particular, the logical qualifier of $\rtsym_p$ includes a {syntactic constraint} $\prov(\nu.{\sf color}, x.f)$ if the intent classifier for the color property predicts "yes" and the slot-filling model outputs column name $f$. 

In particular, for each model, we generate the {syntactic constraint} for this specific property if the intent classifier returns ``yes'' and populate the arguments of the predicate with the output of the slot-filling model. For instance, in Figure~\ref{fig:parser} where we focus on the color property, we first generate the predicate template $\prov(\nu.{\sf color}, x.?)$ where $?$ is to be determined by the argument classifier. Then, $?$ is filled by the output of the argument classifier, which in this case is ${\sf Origin}$ and the model returns the predicate $\prov(\nu.{\sf color}, x.{\sf Origin})$ as its final output. 

The base type of $\rtsym_t$ is obtained using a pre-trained model~\cite{lin-etal-2020-bridging} that outputs the set of likely columns in the table mentioned in the query. For the logical qualifier of $\rtsym_t$, we use the predictions made by the {\sf mean}, {\sf sum} and {\sf count} intent classifiers and their corresponding slot-filling models. For example, the logical qualifier includes a predicate $\prov(\nu.f, {\sf sum})$ if the {\sf sum} model predicts ``yes'' and the slot-filling mechanism predicts column name $f$ for the first argument. 

%Since the natural language query does not indicate anything related to the type of the columns of the table, we assign the types of all the fields in the base types to be $\top$. 

%we do not predict any information about the base type. The qualifiers of $\rtsym_t$ is obtained through the {\sf mean} , {\sf sum} and {\sf count} models through a similar procedure as we described above for $\rtsym_p$.

%\paragraph{Running example} We can now look at the rest of Figure~\ref{fig:parser} in more detail. This figure presents how the model works on the color encoding property. Given the input, the BERT model outputs a sequence of contextualized encodings for the token. With the output of BERT, we can perform the two classification tasks introduced above. The left side shows the property classifier based on $\mathbf{h}_\textrm{CLS}$, applying the linear layer to obtain the binary classification probability. If the model returns yes, we proceed to the second step, which we call argument classification. Here we take BERT encodings for all the tokens representing the column names in the table and apply a linear layer to obtain which column should be associated with the color encoding. Finally, the model returns the provenance predicate $\prov(\nu.{\tt color}, x.{\tt Origin})$ as its highest prediction.
\vspace*{-0.2cm}
\paragraph{Distribution over specifications.} Our parser assigns a probability to each specification by using  the probabilities output by each model. In particular, let us view a refinement type $\rtsym$ as a set of tuples $\{C_i,c_i\}_i$ where $C_i$ is a intent and $c_i$ is an attribute predicted by the slot-filling model. Then, our method assigns a probability to these sets of tuples as $p(\{C_i,c_i\}_i) = \prod_i p(C_i \mid \nl, A_\data, i; \mathbf{w}_{c_i}) p(c_i \mid \nl, A_\data, i; \mathbf{W})$. Hence, we can rank all possible specifications from highest to lowest probability.
\vspace*{-0.2cm}
% \subsection{Training}

% Our model is jointly trained on a collection of examples $(\nl^*,\{C_i^*,c_i^*\}_i)$ where we observe the gold properties and values for each natural language utterance. Our training loss is:
% $$
% \mathcal{L} = - \log p(C_i^* \mid ...) - \log p(c_i^* \mid ...)
% $$
% \later{fix this notation too, I'm just being lazy until we confirm we want this} which is the standard negative log likelihood objective (equivalent to maximizing the log probability of the data). We can optimize this with standard stochastic gradient descent, simultaneously training the shared BERT encoder parameters as well as the weight vectors and matrices comprising the classification layers.

\section{Synthesis from Refinement Type Specifications}\label{sec:synthesis}
In this section, we describe our synthesis algorithm which takes as input  a visualization specification $(\rtsym_p, \rtsym_t)$ and   an input table $D$ and generates \emph{all} visualization programs $\prog_v = \prog_p \circ \prog_t$ such that (1) $\prog_t(D)$ is an inhabitant of $\rtsym_t$ (written $\prog_t(D) \vDash \rtsym_t$) and (2) $\prog_v(D)$ is an inhabitant of $\rtsym_p$.  At a high level, the synthesis algorithm  is based on type-directed top-down enumerative search and uses the refinement type system from Section~\ref{sec:type} to significantly reduce the search space.  We first start by explaining the basic synthesis algorithm (Section~\ref{sec:synthoverview}) and then introduce the concept of \emph{type-directed lemma learning} to improve the scalability of our approach (Section~\ref{sec:lemma}). 
Our algorithms frequently use the typing judgements from Section \ref{sec:type}. While these typing judgments make use of a type environment, we treat the type environment as implicit and drop it to simplify presentation. %(e.g., instead of $\Gamma \vdash t : \btsym'$, we write $\vdash t : \btsym'$ )

\subsection{Overview of Synthesis Algorithm}\label{sec:synthoverview}

Our top-level synthesis algorithm is presented in Figure~\ref{fig:synthesizevis} and works as follows:  Given a table $D$ of type $\rtsym_{in}$ and specification $(\rtsym_p, \rtsym_t)$,  it first synthesizes a set of plotting programs $\progs_p$ whose output type is a subtype of the goal type (line 3). In more detail, each plotting program $\prog_p \in \progs_p$ of type $\rtsym^p_{in} \rightarrow \rtsym^p_{out}$ satisfies the following two properties:
(1)  $\rtsym^p_{out} \subtype \rtsym_p$ and (2) $\rtsym^p_{in} \comp \rtsym_t$. The first constraint ensures that the generated visualization  satisfies the user's specification, and the second constraint ensures that there is \emph{at least one} input table to the plotting program that is consistent with $\rtsym_t$.  Then, for each synthesized plotting program $\prog_p$ of type $\rtsym^p_{in} \rightarrow \rtsym^p_{out}$, the algorithm synthesizes (at line 6) a set of corresponding table transformation programs $\progs_t$ of type  $\rtsym^t_{in} \rightarrow \rtsym^t_{out}$ such that (1) $\rtsym^t_{out} \comp \rtsym_t \land \rtsym^p_{in}$  and (2)  $\rtsym_{in} \subtype \rtsym^t_{in} $. Note that the first  condition \emph{strengthens} the original specification using $\rtsym_{in}^p$ (via intersection types) and
ensures that there is at least one output of the table transformation program that is a valid input to the plotting program. 

The key part of the algorithm is the {\sc SynthesizeGoal} procedure, presented in Figure~\ref{fig:goal_synthesis}, that is used to synthesize \emph{both} table transformation and plotting programs. To unify presentation,  {\sc SynthesizeGoal}  takes a few additional arguments:

\begin{itemize}[leftmargin=*]
    \item $\mathcal{G}$, the grammar for the DSL in which we synthesize programs
    \item The correctness checking condition $\rhd$ for the \emph{input} type (either $\subtype$ or $\comp$)
    \item The correctness checking condition $\lhd$ for the \emph{output} type (either $\subtype$ or $\comp$)
\end{itemize}

  At a high level, {\sc SynthesizeGoal} is a top-down enumeration procedure which starts from the root symbol of the grammar and keeps expanding  non-terminals until it generates a complete program. We represent the syntax of the underlying DSL as a context-free grammar  $\mathcal{G} = (V, \Sigma, R, S)$, where $V, \Sigma$ denote non-terminals and terminals respectively, $R$ is a set of productions, and $S$ is the start symbol. As standard~\cite{neo}, we formalize our top-down enumeration procedure using the notion of \emph{partial programs}:

\begin{definition}[{\bf{Partial program}}] A partial program $\prog$ is a sequence $\prog \in (\Sigma \cup V)*$ such that $S \xRightarrow[]{*} \prog$ (i.e. $\prog$ can be derived from $S$ via a sequence of productions). We refer to any non-terminal in $\prog$ as a hole, and we say that $\prog$ is complete if it does not contain any holes. 
\end{definition}

In the remainder of this section, we represent each partial program $\prog$ as an abstract syntax tree (AST) $(N, E)$ with nodes $N$ and edges $E$. Each node $n \in N$ is represented as a pair  $(l, \goaltype)$ where $l$ is a node label (either a terminal or non-terminal symbol in $\grammar$) and $\goaltype$ is the  \emph{goal type} of the subprogram rooted at $\node$. The goal type $\goaltype$ of a node $n$ serves as a necessary correctness condition such that if the sub-program rooted at $n$ does not satisfy $\goaltype$, then the whole program cannot satisfy its specification.  For a node $n$, we use the notation $P(n)$ to denote the subtree of $P$ rooted at $n$, and use ${\sf Label}(n)$ and ${\sf GoalType}(n)$ to refer to the label and goal type of  $\node$, respectively. Finally, we refer to a node as \emph{complete} if the subtree rooted at $n$ is a complete program.

With this notation in place, we now describe the basic version (everything not underlined) of {\sc SynthesizeGoal} in more detail. Our algorithm maintains a worklist $\worklist$ of partial programs and iteratively grows it. At the beginning, $\worklist$ is initialized to be the empty program $P_0$  with a single node $\node_0$ annotated with the grammar start symbol $S_{\grammar}$ and top-level goal  $\outputtype$. The loop in lines $4$-$18$ dequeues a program $P$ from the worklist with type $\ftype{\inputtype^P}{\outputtype^P}$ and checks if it is complete  and whether it satisfies the correctness conditions. If so,  this program is added to the set $\mathcal{S}$ containing all synthesis results.
Otherwise, {\sc SynthesizeGoal} calls $\textsf{Expand}$  at line 12 to generate a new set of partial programs by expanding a hole $h$ in $P$. Similar to prior work~\cite{synquid,lambda2}, when $\textsf{Expand}$ generates a new partial program $P'$, it propagates the goal type at $h$ to its children. However, the goal types we infer are necessary conditions for correctness with respect to \emph{type compatibility} (as opposed to subtyping) and are derived based on the premises of the typing rules from  Section \ref{sec:typingrules}. In other words, the types of all subprograms must be compatible with their propagated goal type in order for the overall program to be compatible with its goal type.

% \begin{figure}[!t]
% \small
% \vspace{-10pt}
% \begin{algorithm}[H]
% \begin{algorithmic}[1]
% \Procedure{SynthesizeGoal}{$\grammar,\inputtype, \outputtype, \rhd, \lhd$}
% \vspace{0.05in}
% \Statex\Input{The grammar $\mathcal{G}$ of the synthesis DSL }
% \Statex\Input{Input type $\inputtype$ of the target function }
% \Statex\Input{Output type $\outputtype$ of the target function}
% \Statex\Input{Correctness checking operators $\rhd, \lhd \in \{\subtype \ , \sim\}$ for input and output type respectively}
% \Statex\Output{A set $R$ of complete programs that satisfy correctness operators $\rhd$ and $\lhd$}
% \vspace{0.05in}
% \State $\res \assign \{\}$
% \State $\prog_0 \assign \{ (S_{\grammar}, \outputtype), \emptyset)$; $\worklist \assign \{\prog_0\}$
% \While{$\worklist \neq \emptyset$}
% \State $\prog \assign \worklist.remove()$;
% \If{$\textsf{IsComplete}(\prog)$}
% \If{$\vdash {\sf InputType}(P) \rhd \inputtype \  \wedge \vdash {\sf OutputType}(P) \lhd \outputtype$} $\res \assign \res \cup \{\prog\}$;
% \EndIf
% \ForAll{$P' \in \textsf{Expand}(\mathcal{G}, \prog)$}
% \If{\underline{$\textsc{ViolatesLemma}(P', \Phi)$}} 
% \State {\bf continue};
% \ElsIf{$\textsc{TypeIncompatible}(P')$} 
% \State  \underline{$\lemmas \assign \lemmas \cup \textsc{InferLemmas}(\prog', \inputtype);$}
% \Else
% \State $\worklist \assign \worklist \cup \{P\}$;
% \EndIf
% \EndFor
% \EndIf
% \EndWhile
% \State \Return $\res$;
% \EndProcedure
% \end{algorithmic}
% \end{algorithm}
% \vspace*{-1.0cm}
% \caption{Algorithm for synthesizing programs in grammar $\mathcal{G}$.}
% \label{fig:goal_synthesis}
% \end{figure}

Next, for each expansion $P'$ of $P$, the  {\sc TypeIncompatible} procedure (presented in Figure~\ref{fig:typeinfeasible}) uses our refinement type system to check whether $P'$ is infeasible. To do so, it iterates over all nodes and checks whether the subtree rooted at that node is a complete program (line 3). If so, it infers the type  $\rtsym$ of this sub-program using our type system (line 4) and queries whether $\rtsym$ is type-compatible with the goal type of $n$. 
%At the end, it returns true iff this type compatibility checks fails for any of the nodes. 
%Finally, going back to the top-level {\sc SynthesizeGoal} procedure, only those partial programs that pass the type compatibility check at line 9 are added to the worklist (line 12). 

\begin{theorem}
Let $P$ be a partial program with input type $\inputtype$ and top level goal type $\outputtype$. If $\textsc{TypeIncompatible}(P)$ returns true, then for any completion $P'$ of $P$, $P' \not \comp (x : \ftype{\inputtype}{\outputtype})$.
\end{theorem}

\begin{figure}[!t]
% \vspace*{-0.5cm}
\begin{minipage}[t]{0.51\textwidth}
\begin{subfigure}[t]{1.0\textwidth}
    \begin{algorithm}[H]
    \small
    \begin{algorithmic}[1]
    \Procedure{SynthesizeVis}{$(\rtsym_p, \rtsym_t), \data$}
    \Statex \Input{A specification $(\rtsym_t, \rtsym_p)$}
    \Statex \Input{The input table $\data$}
    \Statex \Output{A set of visualization programs.}
    \vspace{0.05in}
    \State $\res \assign \emptyset$; $\inputtype \assign {\sf GetType}(D)$
    \State $\progs_p \assign \textsc{SynthesizeGoal}(\vlang, \rtsym_t, \rtsym_p, \sim, \subtype)$
    \ForAll{$\prog_p : \ftype{\rtsym^p_{in}}{\rtsym^p_{out}} \in \progs_p$}
    \State $\rtsym_{s} \gets \rtsym_t \wedge \rtsym^p_{in}$; 
    \State $\progs_t \assign \textsc{SynthesizeGoal}(\tlang, \inputtype,\rtsym_s, \subtype, \sim)$;
    \ForAll{$\prog_t \in \progs_t$}
    \If{$\prog_t(D) \vDash \rtsym_t \land  \prog_p(\prog_t(D)) \vDash \rtsym_p$}
    \State $\res \assign \res \cup \{\prog_v \circ \prog_t\}$;
    \EndIf
    \EndFor
    \EndFor
    \State \Return $\res$
    \EndProcedure
    \end{algorithmic}
    \end{algorithm}
    \vspace*{-0.8cm}
    \caption{Top-level synthesis algorithm. }
    \label{fig:synthesizevis}
\end{subfigure}

\begin{subfigure}[t]{1.0\textwidth}
    \small
    % \vspace{1.4cm}
    \begin{algorithm}[H]
    \begin{algorithmic}[1]
    \Procedure{{\sc TypeIncompatible}}{$\prog$}
    \Statex\Input{A partial program $\prog$}
    \Statex\Output{True if  type-incompatible}
    \ForAll{$\node \in {\sf Nodes}(P)$}
    \If{${\sf IsComplete}(\prog(\node))$}
    \State $\rtsym \assign {\sf TypeOf}(\prog(\node))$;
    \If{$\vdash \rtsym \incomp {\sf GoalType}(\node)$} 
    \State \Return $\mathsf{true}$;
    \EndIf
    \EndIf
    \EndFor
    \Return $\mathsf{false}$;
    \EndProcedure
    \end{algorithmic}
    \end{algorithm}
    \vspace*{-0.8cm}
    \caption{Procedure for checking program infeasibility .}
    \label{fig:typeinfeasible}
\end{subfigure}
\end{minipage}%
\begin{minipage}[t]{0.49\textwidth}
\begin{subfigure}[t]{1.0\textwidth}
    \begin{algorithm}[H]
    \small
    \vspace{0.5cm}
    \begin{algorithmic}[1]
    \Procedure{SynthesizeGoal}{$\grammar,\inputtype, \outputtype, \rhd, \lhd$}
    \vspace{0.05in}
    \Statex\Input{Grammar $\mathcal{G}$, Specification ($\inputtype$, $\outputtype$) }
    % \Statex\Input{Input type $\inputtype$ of the target function }
    % \Statex\Input{Output type $\outputtype$ of the target function}
    \Statex\Input{Operators $\rhd, \lhd \in \{\subtype \ , \sim\}$ to check correctness for input and output type respectively}
    % \Statex\Output{A set $R$ of complete programs that satisfy correctness operators $\rhd$ and $\lhd$}
    \vspace{0.05in}
    \State $\res \assign \{\}$
    \State $\prog_0 \assign \{ (S_{\grammar}, \outputtype), \emptyset)$; $\worklist \assign \{\prog_0\}$
    \While{$\worklist \neq \emptyset$}
    \State $\prog \assign \worklist.remove()$;
    \State $\inputtype^P \assign {\sf InputType}(\prog)$
    \State $\outputtype^P \assign {\sf OutputType}(\prog)$
    % \State $\ftype{\inputtype^P}{\outputtype^P} \assign \textsf{GetType}(\prog)$;
    \If{$\textsf{IsComplete}(\prog)$}
    \If{$\vdash \inputtype \rhd \inputtype^P \  \wedge \vdash \outputtype^P \lhd \outputtype$} 
    \State $\res \assign \res \cup \{\prog\}$;
    \EndIf
    \State {\bf continue};
    \EndIf
    \ForAll{$P' \in \textsf{Expand}(\mathcal{G}, \prog)$}
    \If{\underline{$\textsc{ViolatesLemma}(P', \Phi)$}} 
    \State {\bf continue};
    \ElsIf{$\textsc{TypeIncompatible}(P')$} 
    \State  \underline{$\lemmas \assign \lemmas \cup \textsc{InferLemmas}(\prog', \inputtype);$}
    \Else
    \State $\worklist \assign \worklist \cup \{P'\}$;
    \EndIf
    \EndFor
    \EndWhile
    \State \Return $\res$;
    \EndProcedure
    \end{algorithmic}
    \end{algorithm}
    \vspace*{-0.8cm}
    \caption{Goal type synthesis algorithm.}
    \label{fig:goal_synthesis}
\end{subfigure}
\end{minipage}
\vspace*{-0.3cm}
\caption{Procedures for program synthesis. In \textsc{SynthesizeVis}, $\vlang$ is the grammar for the plotting sub-DSL and $\tlang$ is the grammar for the table transformation sub-DSL. $\rtsym_t \wedge \rtsym^p_{in}$ stands for the intersection type of $\rtsym_t$ and $\rtsym^p_{in}$. We provide the procedure for computing type intersection  in the appendix.}
% \vspace*{-0.25cm}
\end{figure}

\subsection{Type-Directed Learning}\label{sec:lemma}
We now describe our \emph{type-directed learning} technique that refines the basic synthesis algorithm from the previous subsection. The motivation for this technique is that we need to synthesize \emph{many} programs during each visualization session.
%: First, since our NLP model produces many possible refinement types for a given natural language utterance, we perform synthesis for many different specifications. Second, even for a single specification, there are multiple visualization programs consistent with that specification, and our algorithm needs to explore all of them. However, since all of these synthesis problems target the same data set, there are valuable learning opportunities across different calls to {\sc SynthesizeGoal}. 
To leverage the similarities across all these synthesis tasks, our algorithm learns so-called \emph{synthesis lemmas} that capture inferred constraints for the input data set. While this idea is somewhat similar to the notion of \emph{conflict-driven learning} in prior synthesis work~\cite{neo}, there  are  two key differences. First, our learned lemmas can be reused across different specifications as long as the input data set is the same. Second, the learning of the synthesis lemmas is type-directed and leverages our refinement type system. 
%In the remainder of this section, we discuss what these synthesis lemmas in more detail.

\begin{definition}[\bf{Synthesis lemma}] A synthesis lemma for an input table $\data$ is a pair of refinement types $(\lemmaG, \lemmaR)$ such that, for any program $P$ and type $\rtsym$ satisfying $ \rtsym \subtype \lemmaG$, if $P(D)$ is an inhabitant of $\rtsym$, then we have  $\rtsym \comp \lemmaR$.
\end{definition}

In other words, a synthesis lemma captures \emph{additional} (learned) constraints $\lemmaR$ that the synthesized program must satisfy if its output type is to be a subtype of $\lemmaG$. Given a lemma $(\lemmaG, \lemmaR)$ and partial program $\prog$, the basic idea is to use $\lemmaR$ for pruning as follows: If the desired goal type $\rtsym$ of  $\prog$ is a subtype of $\lemmaG$ but $\rtsym$ is \emph{not} type-compatible with $\lemmaR$, then we can prune $\prog$ without even attempting synthesis. Hence, such lemmas can be useful  both for proving  the unrealizability of a top-level synthesis goal as well as pruning the search space 
during synthesis.

Given a lemma $(\lemmaG, \lemmaR)$, we refer to $\lemmaG$ as the \emph{guard} of the lemma and $\lemmaR$ as the \emph{requirement}. We also say that a lemma is \emph{activated} if the goal type of the synthesis task is a subtype of $\lemmaG$. Clearly, the more general the guard of the lemma, the more pruning opportunities that lemma provides. 
%Hence, for these lemmas to be useful in practice, we need to learn lemmas whose goals are not very specific. For instance, if the guard of a  learned lemma is  exactly the same synthesis goal we have previously encountered, then it would not be useful for proving the unrealizability of future synthesis tasks. In what follows, we first explain how to utilize these lemmas during synthesis and then describe a type-directed technique for learning such useful lemmas.

\paragraph{{\bf Pruning with lemmas.}} To understand how our synthesis procedure utilizes such lemmas, observe that the {\sc SynthesizeGoal} algorithm from Figure~\ref{fig:goal_synthesis} invokes the {\sc ViolatesLemma} procedure shown  in Figure~\ref{fig:violatelemma}. Given a partial program $\prog$ and a  lemma $(\lemmaG, \lemmaR) \in \lemmas$, this procedure checks if there exists some hole in $\prog$ that violates that lemma. In particular, a hole $h$ with annotated goal type $\rtsym$ violates the lemma if $\lemmaG$ is activated (i.e., $\rtsym \subtype \lemmaG)$ and $\rtsym$ is incompatible with requirement $\lemmaR$. If this type compatibility check fails for \emph{any} of the holes, then $\prog$  guaranteed to be infeasible. 

%\begin{figure}[t]
%    \centering
%    \small
%    \vspace*{-.4cm}
%    \begin{algorithm}[H]
%    \begin{algorithmic}[1]
%    \Procedure{ViolatesLemma}{$\prog, \lemmas$}
%    \Statex \Input{A partial program $\prog$ and a set of lemmas $\lemmas$}
%    \Statex \Output{{\sf true} if $\prog$ is infeasible, {\sf false} otherwise}
%    \ForAll{$h \in \textsf{Holes}(\prog)$}
%    \ForAll{$(\lemmaG, \lemmaR) \in \lemmas$}
%    \If{$\vdash {\sf GoalType}(h) \subtype \lemmaG$}
%    \If{$\vdash {\sf GoalType}(h) \incomp \lemmaR$}
%    \State \Return {\sf true};
%    \EndIf
%    \EndIf
%%    \EndFor
%    \EndFor
%    \State \Return {\sf false};
%%    \EndProcedure
%    \end{algorithmic}
%    \end{algorithm}
%    \vspace*{-1.0cm}
%    \caption{Procedure for \textsc{ViolatesLemma}.}
%    \label{fig:violatelemma}

\begin{figure}
\vspace*{-0.5cm}
    \begin{subfigure}[t]{0.45\textwidth}
    \small
    \begin{algorithm}[H]
    \begin{algorithmic}[1]
    \Procedure{ViolatesLemma}{$\prog, \lemmas$}
    \Statex \Input{A partial program $\prog$}
    \Statex \Input{A set of lemmas $\lemmas$}
    \Statex \Output{{\sf true} if $\prog$ is infeasible, {\sf false} otherwise}
    \ForAll{$h \in \textsf{Holes}(\prog)$}
    \ForAll{$(\lemmaG, \lemmaR) \in \lemmas$}
    \If{$\vdash {\sf GoalType}(h) \subtype \lemmaG$}
    \If{$\vdash {\sf GoalType}(h) \incomp \lemmaR$}
    \State \Return {\sf true};
    \EndIf
    \EndIf
    \EndFor
    \EndFor
    \State \Return {\sf false};
    \EndProcedure
    \end{algorithmic}
    \end{algorithm}
    \vspace*{-0.6cm}
    \caption{Procedure for checking violation of lemmas.}
    \label{fig:violatelemma}
    \end{subfigure}
    \begin{subfigure}[t]{0.5\textwidth}
    \small
    \begin{algorithm}[H]
    \begin{algorithmic}[1]
    \Procedure{InferLemmas}{$\prog, \inputtype$}
    \Statex \Input{A failed partial program $\prog$}
    \Statex \Input{Input type $\inputtype$ of $\prog$}
    \Statex \Output{A set of learned lemmas $\lemmas$}
    \State $\lemmas \assign \{\}$;
    \ForAll{$\node \in \textsf{CompleteNodes}(\prog)$}
    \State $\rtsym \assign {\sf TypeOf}(\prog(\node))$;
    \If{$\vdash {\sf GoalType}(\node) \incomp \rtsym$}
    \State $\lemmaG_{\node} \assign {\sf GetInterpolant}(n)$
    \State $\btsym_{in}, \btsym_{G_n} \assign {\sf GetBaseTypes}(\inputtype, \lemmaG_{\node})$
    \State $\lemmaR_{\node} \assign \textsc{GenReq}(\btsym_{in}, \btsym_{G_n}, \mathtt{max\_depth})$;
    \State $\lemmas \assign \lemmas \cup (\lemmaG_\node, \lemmaR_\node)$;
    \EndIf
    \EndFor
    \State \Return $\lemmas$;
    \EndProcedure
    \end{algorithmic}
    \end{algorithm}
    \vspace*{-0.8cm}
    \caption{Procedure for inferring lemmas.}
    \label{fig:inferlemma}
    \end{subfigure}
    \vspace*{-0.3cm}
    \caption{Core type-directed lemma learning procedures}
    % \vspace*{-0.4cm}
\end{figure}

\begin{theorem}   Let $P$ be a partial program with input type $\inputtype$ for table $D$ and whose top level goal type is $\outputtype$. If $\textsc{ViolatesLemma}(P, \Phi)$ returns true, then $P(D)$ is not an inhabitant of $\outputtype$.
\end{theorem}

\paragraph{{\bf Learning lemmas.}} Next, we discuss how to use our refinement type system to infer these synthesis lemmas. As shown in line 16 of Figure~\ref{fig:goal_synthesis}, our synthesis technique invokes a procedure called {\sc InferLemmas} (presented in Figure \ref{fig:inferlemma})  every time it encounters an infeasible partial program. 
%Specifically, given an infeasible partial program $\prog$ and the type of the input table $\rtsym_{in}$, {\sc InferLemmas} derives a set of synthesis lemmas that can be used for pruning in future synthesis attempts. %Recall that a synthesis lemma $(\lemmaG, \lemmaR)$ consists of a guard $\lemmaG$ and inferred requirements $\lemmaR$ such that for any partial program $\prog'$ whose goal type is a subtype of $\lemmaG$, its goal type also needs to satisfy the additional requirements encoded by $\lemmaR$. Thus, for these lemmas to be useful in practice, the context needs to be as reusable as possible. 
In order to generate useful lemmas, we introduce the notion of \emph{type interpolants} that are inspired by Craig interpolation~\cite{craig} in logic. Intuitively,  type interpolants allow our algorithm to learn lemmas with generalizable guards that can be activated in many contexts. 

\begin{definition}[\bf{Base type interpolant}] Given two incompatible base types $\btsym_1$ and $\btsym_2$,  we say that $\btsym$ is a base type interpolant for $\btsym_1$ and $\btsym_2$ if (1) $\btsym_1 \subtype \btsym$, (2) ${\btsym} \incomp {\btsym_2}$, and (3) for any $\btsym'$ such that $\btsym \subtype \btsym'$, we have ${\btsym'} \sim {\btsym_2}$.
\end{definition}

\begin{example}
If $\btsym_1 = {\sf Table}(\{{\sf colA} : {\sf Discrete}, {\sf colB} : {\sf Qualitative}, {\sf colC} : {\sf Continuous}\})$ and $\btsym_2 = {\sf Table}(\{{\sf colA} : {\sf Qualitative}, {\sf colB}: {\sf Qualitative}, {\sf colC}: {\sf Continuous}\})$ then the base type interpolant for $\btsym_1$ and $\btsym_2$ is ${\sf Table}(\{{\sf colA}: {\sf Quantitative}\})$. Note that the type interpolant isolates the incompatibility; namely {\sf colA} in $\btsym_1$’s schema is a {\sf Quantitative} data type, but {\sf colA} in $\btsym_2$’s schema is {\sf Qualitative}.
\end{example}

Next, we generalize this notion from base types to refinement types:

\begin{definition}[\bf{Refinement type interpolant}] Given two refinement types $\rtsym_1 = \rtype{\btsym_1}{\phi_1}$ and $\rtsym_2 = \rtype{\btsym_2}{\phi_2}$, we say that $\rtsym$ is a type interpolant for $\rtsym_1$ and $\rtsym_2$ if:

\begin{itemize}
    \item $\btsym_1 \incomp \btsym_2$, then $\rtsym$ is the base type interpolant for $\btsym_1$ and $\btsym_2$
    \item $\btsym_1 \sim \btsym_2$, then $\rtsym = \rtype{\btsym_1}{\phi}$ and $\phi$ is a Craig interpolant for $\phi_1$ and $\phi_2$
\end{itemize}
\end{definition}

\begin{example}
Let $\rtsym_1 =  \{\nu : {\sf Table}(\{{\sf colA}:{\sf Discrete}, {\sf colB}: {\sf Discrete}\}) \mid |(\nu, \{\sf colA\})| \leq |(\nu, \{\sf colB\})| \leq 20 \}$ and $\rtsym_2 = \{\nu : {\sf Table}(\{{\sf colA} : {\sf Discrete}\}) \mid  |(\nu, \{\sf colA\})| = 30\}$. Then $\{\nu : {\sf Table}(\{{\sf colA} : {\sf Discrete}, {\sf colB} : {\sf Discrete}\}) \mid  |(\nu, \{\sf colA\})| \leq 20\}$ is a refinement type interpolant for $\rtsym_1$ and $\rtsym_2$.
\end{example}

With these definitions in place, we now describe {\sc InferLemma} in more detail. Given an infeasible partial program $\prog$, {\sc InferLemma} first iterates over every complete node $n$ in $P$ and checks whether $n$'s goal type and actual type are incompatible (lines 4-5). If they are, it proceeds to generate a lemma $(\lemmaG_n, \lemmaR_n)$ where $\lemmaG_n$ is a type interpolant between $n$'s goal and actual types (lines 6) and the requirement $\lemmaR_n$ is generated using the  call {\sc GenReq}  (line 8). Intuitively, the use of type interpolants allows learning lemmas whose guards are as general possible so that they are frequently activated.

The {\sc GenReq} procedure for generating a requirement is presented as inference rules in Figure \ref{fig:genreq}. At a high level, {\sc GenReq} infers DSL constructs that must be used in order to satisfy the goal type and expresses these as {syntactic} constraints.
In more detail, this procedure takes three inputs: (1) the base type  $\btsrc$ for input table $D$, (2) the base type $\btdst$ of the lemma guard, and (3) a synthesis depth $k$ which serves as an upper-bound on the AST depth of the program to be synthesized. The output of {\sc GenReq} is a refinement type $\lemmaR$ such that all programs of maximum AST depth $k$ and with base type $\btsrc \rightarrow \btsym$ where $\btsym \comp \btdst$ must have an output type that is compatible with $\lemmaR$. 

We now explain the two inference rules from Figure~\ref{fig:genreq} in more detail. The first rule, labeled {\sc Base}, is the base case for the recursive {\sc GenReq} procedure. In the case where $k=1$, {\sc GenReq}  finds the set of $F$ of all DSL operators $f$ such that $f$ takes as input a value of base type $\btsrc$ and produces an output whose base type is compatible with $\btdst$. Then, the generated requirement is that the synthesized function must contain one of the operators in $F$: this is expressed as a disjunction of {syntactic constraints}, where each formula is of the form $\bigvee_{c_i}\prov(\nu.c_i, f)$ and $c_i$ is an index over the attributes of $\btdst$. Intuitively, this formula says that $f$ could be used to derive any of the columns in the target table's schema.
\begin{example}
Suppose $\btsrc = {\sf Table}(\{{\sf colA} : {\sf Qualitative}\})$ and $\btdst = {\sf Table}(\{{\sf colA} : {\sf Discrete}\})$. When $k = 1$, {\sc GenReq} returns $\{\nu : {\sf Table}(\{{\sf colA} : {\sf Discrete}\}) \mid \prov(\nu.{\sf colA}, {\sf count})\}$ as {\sf count} is the only operation which directly transforms a {\sf Qualitative} column to a {\sf Discrete} one.
\end{example}

The second rule  from Figure~\ref{fig:genreq} handles the case for $k>1$. To compute a suitable requirement, it first utilizes the base case to get an encoding $\phi_{R_1}$ of all programs of depth 1 whose input type is  $\btsrc$ to whose output is compatible with $\btdst$. Next, it computes an encoding $\phi_{R_2}$ of all programs of depth $k \geq 2$ of the form $P^{k-1} \circ f^1$ where $f^1$ is a function from $\btsrc$ to an intermediate type $\btsym_t$, and $P^{k-1}$ is a program of depth at most $k-1$ whose input type is $\btsym_t$ and output type is compatible with $\btdst$.  Thus, the constraint $\phi_{R_1} \lor \phi_{R_2}$ encodes the requirement for all programs \emph{up to} depth $k$.

\begin{example}
Suppose $\btsrc = {\sf Table}(\{{\sf colA} : {\sf Qualitative}\})$  and $\btdst = {\sf Table}(\{{\sf colA} : {\sf Continuous}\})$. Then $\req_{R_1}$ is $\bot$ (false) because there is no operation that can directly transform a {\sf Qualitative} column to a {\sf Continuous} one. However, a {\sf Qualitative} column can only be converted to a {\sf Continuous} one via the {\sf count} operation followed by a {\sf mean} or {\sf sum}. As such, $\req_{R_2} =  \prov(\nu.{\sf colA}, {\sf count}) \land (\prov(\nu.{\sf colA}, {\sf mean}) \lor \prov(\nu.{\sf colA}, {\sf sum}))$. Thus, {\sc GenReq} returns $\{\nu : {\sf Table}(\{{\sf colA} : {\sf Continuous}\}) \mid \req_{R_2} \}$
\end{example}

\begin{figure}
    \centering
    \tiny
    % \vspace*{-0.6cm}
    \begin{minipage}[c]{0.5\textwidth}
    \begin{mathpar}
    \inferrule*[Left=Base]{\\\\\\\\\\\\\\\\
    k = 1 \\\\ F = \{f \in \mathsf{Ops} \mid f: \ftype{\btsrc}{\btdst'}, \ \vdash \btdst' \sim \btdst\} 
     }{k \vdash (\btsrc, \btdst) \leadsto \rtype{\btdst}{\bigvee_{f \in F} \bigvee_{c_i \in \btdst}\pi(\nu.c_i, f)}} \and
     \end{mathpar} 
     \end{minipage}
    \begin{minipage}[c]{0.49\textwidth}
     \begin{mathpar}
    \inferrule*[Left=Rec]{k > 1 \and
    1 \vdash (\btsrc, \btdst) \leadsto \rtype{\btdst}{{\req_R}_1} \\\\
    F = \{f \mid f \in \mathsf{Ops} \land f: \ftype{\btsrc}{\btsym_t}\} \\\\
    \forall f: \ftype{\btsrc}{\btsym_t} \in F. \ k-1 \vdash (\btsym_t, \btdst) \leadsto \rtype{\btdst}{{\req_R}_f} \\\\
    {\req_R}_2 = \bigvee_{f \in F}\left(\bigvee_{c_i \in \btsym_t}\pi(\nu.c_i, f) \land {{\req_R}_f}\right)
    }{k \vdash (\btsrc, \btdst) \leadsto \rtype{\btdst}{{\req_R}_1 \lor {\req_R}_2}}
    \end{mathpar}
    \end{minipage}
    \vspace*{-0.2cm}
    \caption{ \textsc{GenReq} procedure where $k$ describes an upper bound on the maximum AST depth of the function to be synthesized. $\btsrc$ is the desired input type and $\btdst$ is a base type that the output must be compatible with.}
    \vspace*{-0.2cm}
    \label{fig:genreq}
\end{figure} %\todo{Need to write something about Encode} 

We now state and prove theorems about our main synthesis procedure {\sc SynthesizeVis}.

\begin{theorem}{\bf{(Soundness)}}
Suppose $\textsc{SynthesizeVis}((\rtsym_p, \rtsym_t), D)$ returns a set of programs $\mathcal{S}$. Then for each visualization program $\prog_v = \prog_p \circ \prog_t \in \mathcal{S}$, $\prog_t(D) \vDash \rtsym_t$ and $\prog_v(D) \vDash \rtsym_p$.
\end{theorem}

\begin{theorem}{\bf{(Completeness)}}
Given a specification $(\rtsym_p, \rtsym_t)$ and input table $D$, if there is a visualization program $\prog_v = \prog_p \circ \prog_t$ such that $\prog_t(D) \vDash \rtsym_t$ and $\prog_v(D) \vDash \rtsym_p$, then $\prog_v \in \textsc{SynthesizeVis}((\rtsym_p, \rtsym_t), D)$
\end{theorem}

\section{Implementation}
We have implemented the proposed algorithm as a new tool called \toolname written in Python. In what follows, we describe key implementation details that are not covered in the technical sections.

\paragraph{{\bf Parser implementation and training}} Our NL parser is based on the BERT implementation and pre-trained weights provided by HuggingFace \cite{wolf-etal-2020-transformers}. For training, our model is jointly trained on a collection of examples $(\nl,\{C_i^*,c_i^*\}_i)$ where we observe the goal properties and values for each natural language utterance. Our training loss for an example is:
\small
\[
\vspace*{-0.1cm}
\mathcal{L}(\nl,\{C_i^*,c_i^*\}_i) = \sum_{i} - \log p(C_i^* \mid \nl, \data, i; \mathbf{w}_{c_i}) - \log p(c_i^* \mid \nl, \data, i; \mathbf{W})
\]
\normalsize
which is the standard negative log likelihood objective. % (equivalent to maximizing the log probability of the data). 
We can optimize this objective with standard stochastic gradient descent, simultaneously training the shared BERT encoder parameters as well as the weight vectors and matrices comprising the classification layers. To implement the training procedure, we use the AdamW optimizer~\cite{loshchilov2018decoupled}, and train the models for 20 epochs with a batch size of 16. We provide the detailed hyperparameters including the learning rates in the appendix.

\paragraph{{\bf Type interpolants.}} Recall that our lemma generation technique from Section~\ref{sec:lemma} utilizes the notion of type interpolants to generate guards. While computing interpolants for base types is quite straightforward, we sometimes also need to compute Craig interpolants for the logical qualifiers. In order to ensure that the overhead of this procedure does not outweigh its benefits, we use a simple template-based approach to generate interpolants. In particular, we only generate interpolants that are conjunctions of predicates and enumerate them in increasing number of atomic predicates up to a small bound. The atomic predicates are generated from a pre-defined family of predicates and instantiated with terms and constants that appear in the input formulas. 

\paragraph{{\bf Ranking visualizations.}} As described in Section~\ref{sec:parse}, \toolname ranks specifications based on the probabilities produced by the intent classifier and slot filling model. To break ties between programs associated with the same specification, we use the order in which they were explored during the synthesis process, which has the effect of ranking simpler programs above more complicated ones.

\section{Evaluation}\label{sec:eval}
We now describe a series of experiments designed to answer the following research questions:
\begin{itemize}[leftmargin=*]
    \item {\bf RQ1.} How do the results produced by \toolname compare against those of existing tools?
    \item {\bf RQ2.} How long does \toolname take to synthesize visualizations?
    \item {\bf RQ3.} How important are the refinement type system and lemma learning for performance?
    \item {\bf RQ4.} How effective do users find \toolname in generating visualizations?
    
\end{itemize}

\begin{figure}
    \centering
    \vspace*{-0.5cm}
    \hspace*{-2.0cm}
    \begin{minipage}[t]{0.44\textwidth}
    \begin{table}[H]
    \vspace*{-0.2cm}
\footnotesize
\centering
\begin{tabular}{|c|c|c|}
    \hline
     {\bf Domain}  & {\bf \# of Columns} & {\bf NL Queries} \\
     \hline
     Cars & 9  &  278 \\
     \hline
     Movies  & 10 & 243  \\
     \hline
     Superstore  & 27 & 209 \\
     \hline
     \multicolumn{2}{c|}{} & {\bf 730} \\
     \cline{3-3}
     %\hline
    % x3 & y3 \\
     %\hline 
    % x4 & y4 \\ 
     %\hline
\end{tabular}
\caption{Summary of datasets from {\sc NLVCorpus}.}
\label{tab:dataset}
% \vspace*{-0.5cm}
\end{table}

    \end{minipage}
    \hspace{0.1cm}
    \begin{minipage}[t]{0.45\textwidth}
    \begin{table}[H]
    \vspace*{-0.2cm}
\footnotesize
\centering
\begin{tabular}{|c|c|c|c|}
    \hline
     {\bf Domain} & {\bf Parsing Time} & {\bf Synthesis Time} & {\bf Total Time} \\
     \hline
     Cars & 4.81 & 0.67 & 5.48  \\
     \hline
     Movies & 4.89 & 0.94 & 5.83 \\
     \hline
     Superstore & 5.82 & 0.69 & 6.51 \\
     \hline 
     \multicolumn{1}{c|}{} & {\bf 5.13} & {\bf 0.77} &  {\bf 5.89} \\
     \cline{2-4}
\end{tabular}
\caption{Average \toolname runtime in seconds.}
\label{tab:runtime}

\end{table}
    \end{minipage}
\vspace*{-0.5cm}
\end{figure}

\paragraph{{\bf Benchmarks}} To answer these questions, we perform an evaluation on the {\sc NLVCorpus} benchmark suite \cite{nlvcorpus}. {\sc NLVCorpus} contains a large collection of  {\emph{real-world}} natural language queries and their corresponding ground truth visualizations for three domains, namely Cars, Movies, and Superstore. In total, the corpus contains over 700 queries, gathered from around 200 users and specifying a variety of visualizations. {We also note that the queries in this benchmark set are quite diverse as they are \emph{syntactically unrestricted}, and each user was only allowed to specify the visualizations from one of the domains.} Table~\ref{tab:dataset} gives a high level summary of the {\sc NLVCorpus} benchmarks. 

\paragraph{\bf{Training set for the parser}} In order to use our parser, recall that we  first need to train it. Hence, to evaluate \toolname on one of the domains (e.g., Cars) of {\sc NLVCorpus}, we train it on the other two domains (e.g., Movies and Superstore). {The training data for each of the domain is automatically generated from the corresponding ground truth visualization programs provided by {\sc NLVCorpus} with no manual effort required.}

%Since NLV{\sc Corpus} is the only dataset in this domain, we do the evaluation in a cross-validation style. In particular, when we evaluate on the Cars dataset, we train the parser using the data from Movies and Superstore. The setting for evaluating on the Movies and Superstore dataset is similar. 

\paragraph{{\bf Experimental Setup}} All of our experiments are conducted on a machine with Intel Xeon(R) W-3275 2.50 GHz CPU and 32 GB of physical memory, running the Ubuntu 18.04 operating system with a NVIDIA Quadro RTX8000 GPU.

\paragraph{{\bf \toolname Configuration}} In all of our experiments, we configure \toolname to terminate after it finds ten visualization programs. We sort these programs by the score obtained from the parser  and  break ties by prioritizing programs with smaller AST sizes.

\subsection{Comparison with Other Tools}
To answer our first research question, we compare \toolname~ against the following existing tools:

\begin{itemize}[leftmargin=*]
    \item {\sc NL4DV} \cite{nl4dv}: A state-of-the-art \emph{rule-based} technique for generating visualizations from natural language.   
    \item {\sc Draco-NL} \cite{draco}: A variant of {\sc Draco}, which is a visualization recommendation system that generates visualizations from a partial specification. Even though {\sc Draco} does not support natural language queries by default, we implemented a custom translator that converts the output of our NL parser to partial specifications in {\sc Draco}'s query language. 
    \item {\sc NcNet-Original} \cite{ncnet}: A transformer based encoder-decoder model which translates natural language queries to visualizations. This variant of {\sc NcNet} was trained only on the {\sc NL2VIS} dataset \cite{nl2vis}. %\footnote{{\sc NL2VIS} is a large-scale machine-generated dataset that is translated from an existing NL2SQL dataset\cite{yu-etal-2018-spider}. We do not evaluate on this dataset because the queries are synthetic and do not reflect real-world settings.}
    \item {\sc NcNet-Augmented}: A variant of {\sc NcNet} that was trained on an augmented dataset that combines both {\sc NL2VIS} and {\sc NLVCorpus}. To ensure there is no overlap between the training and test data,  we test on one of the domains (e.g., Cars) from {\sc NLVCorpus} and train on the other two (e.g., Movies and Superstore)  when performing our evaluation.
    
   % More specifically, when we evaluate this model on a dataset in {\sc NLVCorpus}, the training data we used from {\sc NLVCorpus} is from the other two datasets.
    % \item {\sc NcNet-Augmented}: A variant of {\sc NcNet} that was trained on both {\sf nl2vis} and the datasets used in evaluation. When evaluating this model on a dataset, we used a variant that was trained on the other two and {\sf nl2vis}. 
    \item {\sc Bart-Vis} \cite{lewis-etal-2020-bart}: A translation-based approach that uses a fine-tuned {\sc BART} language model to directly generate visualizations. We adopt a similar training and testing set-up as  {\sc NcNet-Augmented}. 
\end{itemize}

\begin{table}[]
% \vspace{-0.3cm}
    \scriptsize
    \centering
    \begin{tabular}{|cc|ccc|ccc|ccc|}
    \hline
    \multicolumn{2}{|c|}{\multirow{2}{*}{{\bf Tool}}} & \multicolumn{3}{c|}{{\bf Cars}} & \multicolumn{3}{c|}{{\bf Movies}} & \multicolumn{3}{c|}{{\bf Superstore}}\\
    & & top-1 & top-5 & top-10 & top-1 & top-5 & top-10 & top-1 & top-5 & top-10 \\
    \hline
    \multirow{2}{*}{Rule-based} & {\sc NL4DV} & 0.43 & 0.49 & 0.49 & 0.43 & 0.48 & 0.49 & 0.05 & 0.46 & 0.51  \\
    & {\sc Draco-NL} & 0.36 & 0.53 & 0.57 & 0.26 & 0.40 & 0.41 & 0.40 & 0.59 & 0.59  \\
    \hline
    \multirow{3}{*}{Translation-based} & {\sc NcNet-Original} & 0.08 & 0.08 & 0.08 & 0.09 & 0.09 & 0.09 & 0.07 & 0.07 & 0.07 \\
    & {\sc NcNet-Augmented} & 0.10 & 0.11 & 0.11 & 0.12 & 0.12 & 0.12 & 0.07 & 0.07 & 0.07 \\
    & {\sc Bart-Vis}\tablefootnote{While the number for this baseline is low, we did confirm that {\sc Bart-Vis} gave 31\% accuracy on the test set of {\sc NL2VIS}. However it does not seem to train well with the scale of data we have in the real-world setting. }  & 0.09 & 0.11 & 0.12 & 0 & 0.08 & 0.15 & 0 & 0 & 0\\
    \hline
    \multicolumn{2}{|c|}{\toolname}  & 0.58 & 0.77 & 0.85 & 0.48 & 0.64 & 0.71 & 0.54 & 0.81 & 0.84 \\
    \hline
    \end{tabular}
    \caption{Comparison between \toolname\ and other tools in terms of accuracy on the {\sc NLVCorpus}.}
    \label{tab:main_res}
    \vspace*{-0.7cm}
\end{table}
\paragraph{Main results}
We evaluate the performance of all tools in terms of their average top-1, top-5 and top-10 accuracy with respect to the ground truth label. As summarized in Table~\ref{tab:main_res}, \toolname outperforms all other tools in terms of accuracy. Among these tools, {\sc NL4DV} and {\sc Draco-NL} are the closest competitors to \toolname; however, they both have significantly lower top-10 results compared to \toolname, and their performance fluctuates across different domains. 

\paragraph{Running time} As demonstrated in Table~\ref{tab:main_res}, \toolname has better overall accuracy across the board; however, the reader may wonder if this accuracy comes at the cost of significantly longer running times. The last column in Table~\ref{tab:runtime} shows the average end-to-end running time of \toolname across the three different domains. As we can see from this table, \toolname is quite fast despite performing enumerative synthesis, taking an average 5.89 seconds to complete each benchmark.

%While achieving higher accuracy, \toolname can finishes all queries in the benchmarks within seconds. \later{Should we say something more detail about the synthesis time here? It is a bit difficult because it is hard to measure...}

\paragraph{Failure analysis for the baselines} As we can see from Table~\ref{tab:main_res}, some of the baseline tools  perform quite poorly on the {\sc NLVCorpus} data set, so we try to provide some intuition about why this is the case.  At a high level,  machine translation-style approaches do not do well for two main reasons. First, they require the natural language query to have a complete specification of the intended plot; however, many of the queries in {\sc NLVCorpus} only have \emph{partial} specifications, similar to the working example from Section \ref{sec:overview}. Second, machine translation approaches do not take any logical constraints into account and may therefore end up generating non-sensical visualizations, such as a bar chart where the x-axis is associated with a continuous variable. 
Among the rule-based techniques, {\sc NL4DV} heavily relies on the both the parsing and visualization recommendation rules encoded in the tool and therefore fails to achieve good results when the dataset becomes more complicated and rules cannot generalize. Finally,  while {\sc Draco-NL}  utilizes the output of our parser, we observed that it produces low-quality results when the ground truth visualization requires performing non-trivial aggregation operations over the input data set.

%its hand-crafted hard and soft constraints are not sufficient for reasoning about visualizations that require data aggregations. As such, the plots it produces are not as high-quality as \toolname's. 

\paragraph{Failure analysis for \toolname} We also analyzed the cases where \toolname does not produce the intended visualization among its top-$k$ results. In some cases, the ground truth visualization was not ranked sufficiently high, but \toolname can generate the intended visualization as we increase the value of $k$. In most cases, however, \toolname fails to generate the correct visualization because the natural language query does not contain enough hints about the attributes that should be used to generate  the visualization. In such cases, the parser is not  able to infer even the base type for the output of the table transformation program, resulting in a very imprecise specification. 

%and this causes our parser to generate a list of specifications that does not capture the ground truth program. For example, in the cars dataset, we cannot synthesize the visualization corresponding to the query "How does mileage relate to cylinder count?" as mileage is not enough of a hint to infer the required column "mpg".

\subsection{Ablation Study}
In this section, we present the results of an ablation study to justify some of the design choices underlying \toolname. 

\begin{figure}
  % 
% \vspace*{-0.2cm}
    \centering
    \begin{minipage}[t]{0.38\textwidth}
    % \vspace*{-0.2cm}
    % \hspace{-0.5cm}
    \begin{table}[H]
    \scriptsize
    \centering
    % \vspace*{1.5cm}
    % \hspace*{-0.5cm}
    \begin{tabular}{|c|c|c|}
    \hline
     {\bf Tool}  & {\bf Avg Time (s)} & {\bf \%  Completed} \\
     \hline
     {\sc BaseOnly} & 19.61  &  84.2  \\
     \hline
     {\sc TableOnly}  & 20.71 & 91.0  \\
     \hline
     {\sc SynOnly}  & 2.11 & 91.4  \\
     \hline
     \toolname & 0.70 & 100.0 \\
     \hline
\end{tabular}
\vspace{0.5cm}
\caption{Refinement type ablation results.}
\label{tab:qualifier_eval}
\end{table}
    \end{minipage}
\hspace*{0.55cm}
    \begin{minipage}[t]{0.5\textwidth}
    \centering
    \vspace*{0.2cm}
    \input{figures/lemma_eval}
    \vspace*{-0.5cm}
    \caption{Completed benchmarks over time.}\label{fig:lemma_eval}
    \end{minipage}
\vspace*{-0.7cm}
    
\end{figure}

\subsubsection{Importance of refinement types}

First, we quantify the impact that each component of our refinement type system has on \toolname's ability to prune infeasible programs. We do so by comparing the following three variants of \toolname:

\begin{itemize}[leftmargin=*]
    \item {\sc \toolname-BaseOnly}: This is a variant of \toolname that only uses base types, but no logical qualifiers.
    \item {\sc \toolname-SynOnly}: This is a variant of \toolname that uses the base types and the {syntactic} constraints in the type system, but no table property constraints. 
    \item {\sc \toolname-TableOnly}: This is a variant of \toolname that uses the base types and the table property predicates in the type system, but no {syntactic} constraints. 
\end{itemize}

Given unbounded time, all of these variants will have the same accuracy as \toolname because they all check that a candidate program satisfies the specification by running it on the input table. As such, we compare the time they take to complete each benchmark (i.e., return ten visualization programs). In particular, we run each variant over all the benchmarks (with a 60 second timeout), and record the number of benchmarks each one completes along with the average time taken.

The results of this ablation study are shown in Table~\ref{tab:qualifier_eval}, where we report both the average synthesis time in seconds (excluding the time to parse the natural language description into a refinement type) as well as the percentage of benchmarks solved within the 60 second time limit.  Compared to {\sc BaseOnly}, {\sc SynOnly} is almost 10$\times$ faster and {\sc TableOnly} completes nearly 50 more benchmarks within the time limit. Finally, having both {syntactic} and table property constraints allows Graphy to complete all the benchmarks (66 more benchmarks than {\sc SynOnly}) and provides an overall speedup of 28$\times$ compared to base types alone.

\subsubsection{Importance of lemma learning}

We also perform a second ablation study to evaluate the importance of the type-directed lemma learning technique presented in Section~\ref{sec:lemma}. To perform this study, we consider  a variant of \toolname called {\sc \toolname-NoLemma} that is the same as \toolname except that it does not perform type-directed lemma learning.

The results of this ablation study are presented in Figure~\ref{fig:lemma_eval}, which shows the number of benchmarks completed (x-axis) within a given time limit (y-axis) when generating top-10 visualizations. As we can see from the gap between the two lines, \toolname is significantly faster than {\sc \toolname-NoLemma} and achieves an overall speed of $2.1\times$ across all benchmarks. For example, \toolname can complete 97\% of the benchmarks within 2 seconds, whereas {\sc \toolname-NoLemma}  completes 77\%.

\subsection{User Study}
We conducted a small user study to evaluate whether \toolname is helpful to end users. We recruited 12 participants, consisting of a mix of undergraduate and graduate students in computer science, math, and business. We  asked each participant to reproduce 2 plots \footnote{the plots were of different styles and required different aggregation operations to derive them.} from the Cars domain using both \toolname and Excel. For each plot and tool, we gave the participants 15 minutes to reproduce the plot using the tool.

\paragraph{Graphy Setup} To facilitate this user study, we developed a UI on top of \toolname. The UI allows users to enter a natural language query and presents the top 10 visualizations generated by \toolname. To avoid biasing the type of natural language query, no examples were given; participants were only told to keep the query high-level and under twenty words.

\paragraph{Excel Setup} We gave each participant a 10 minute tutorial demonstrating  how to generate a scatter plot and perform table transformations using PivotTable. Participants were also allowed to use any online resource of their choice.

%To further evaluate if \toolname is applicable at helping end-users solve real-world tasks, we conducted a user study involving 12 participants, who are either undergraduate and graduate students in computer science, math and business. Each participant was provided with 2 plots in the Cars dataset that are of different plot types and requires different aggregation operation regardless of if we can always synthesize the plot for any natural language in the benchmark or not. Then, we asked the user to reproduce the plot in both \toolname and Excel. For both tool, the user had 15 minutes to work on reproducing each of the plot. More details about the user study can find found in the appendix. 

%For the set-up involving our tool, participants were just provided with the tool and introduced to the interface. In order to not bias the user with the type natural language they provide, we does not give user any natural language query example but just tell them to keep the description high-level and within 20 words. We gave them 5 minutes to try out the tool.  For the set-up for Excel, we first introduced the user to the interface, showed how to do a scatter plot and a mean aggregation using pivot table, and the user had maximum 10 minutes to get themselves familiar with the tool. 

\paragraph{{\bf Results}} When using Excel, participants could only finish 92\% of the tasks within the time limit, and took 443 seconds on average to solve a task. On the other hand, when using \toolname, the participants could solve \emph{all} the tasks and were, on average, nearly $12.7\times$ faster. We give a more thorough presentation of our user study in the appendix. 

%both plots well under the timelimit (33 seconds on average), and, on average, were able to do so in 33 seconds. On the other hand, when using Excel, participants could onl

%\paragraph{{\bf Results.}} In the set up where participants reproduced the plots through Excel, they finished 79\% of the tasks in the given time limit with the average time of 516 seconds. In contrast, when they had access to \toolname, all the participants finished the tasks within the time limit and on average it took them 33 seconds to finish a task. As standard in user study, to check whether Tool (i.e. using \toolname or Excel to produce a plot) is a significant factor contributes to the finishing time, we did a linear mixed effect model with the Time as the dependent variable and Tool as the independent variable. We also included Tool * Task interaction and Task as a random effect. The p-value for the Time factor is less than $0.000000001$ ($1\mathrm{e}{-9}$). Thus, our user study provides firm evidence that the proposed technique makes it easier for user to produce plot in a real-world setting. 
\section{Related Work}

\paragraph{NLIs for data visualization} Many visualization NLIs are  powered by (1) ruled-based translation engines~\cite{articulate,datatone,DBLP:journals/tvcg/YuS20} that pattern-match keywords in the user input and translate them into visualization constructs, or (2) neural translation engines that leverage encoder-decoder models~\cite{nl2vis, ncnet} or pre-trained language models~\cite{poesia2022synchromesh} to directly generate a visualization program. Because existing systems expect the input NL to be complete specifications of the visualization task, they do not perform well in complex tasks where the query is incomplete or the task requires data transformation~\cite{nl4dv}. \toolname addresses this issue by formulating the visualization task as a recommendation task: it first extracts an incomplete user specification from the NL query and then generates diverse recommendations from it. As shown in our evaluation, \toolname generalizes better to complex tasks and improves user experience. 
%In the future, \toolname can work with interactive data analysis tools~\cite{eviza} for conversational visualization design.  

%One line of work focuses on automatically generating the visualizations, and the state-of-the-art tools are either rule-based~\cite{articulate} or transformer-based ~\cite{ncnet}. Another line of works leverages additional user interactions to generate high quality visualizations \cite{datatone, askdata, powerbi}. For instance, DataTone~\cite{datatone} learns to resolve the inherent ambiguity in natural language queries by using user preference histories, and Evizeon~\cite{eviza} allows users to interact with the system by asking follow-up questions. Our approach falls into the former category as it does not require additional user-provided information to generate high quality visualizations. 

% \vspace{-0.2cm}
\paragraph{Visualization recommendation systems} Visualization recommendation systems are built to help the user  explore the visualization design space from incomplete specifications. For example, {\sc Draco}~\cite{draco} leverages a constraint solver to recommend visualizations based on the user's design and data constraints written in answer set programs; Voyager~\cite{voyager} and ShowMe~\cite{showme} use heuristics to recommend visualization chart type and axes based on the user's fields of interests and the data statistics; DeepEye~\cite{deepeye} is similar to Voyager but with a statistical learning-to-rank model to the rank visualizations. Unlike existing systems that require formal specifications (e.g., constraints, concrete fields) as input, \toolname supports visualization recommendation from natural language. 
%In future, \toolname can be integrated to existing visualization exploration systems~\cite{DBLP:journals/tvcg/YuS20, voyager} to reduce users' specification effort.

%Similar to our tool, {\sc Draco} and  recommend visualizations based on partial specifications, but they only support incomplete programs and cannot handle natural language specifications. In our tool, we give visualization recommendation based on both users' natural language specification and visualization design guidelines.

% \paragraph{Program synthesis for table transformation} 
% \vspace{-0.2cm}
\paragraph{Type-directed program synthesis} Since refinement types \cite{dependent,liquid} were introduced, there have been a number of proposed techniques for synthesizing programs from refinement type specifications \cite{synquid, plpv09, tyde19, myth, myth2, resyn}. In particular, {\sc Myth2} takes a function type, along with input-output examples and generates a refinement type specification by combining the type signature and examples. {\sc Synquid}~\cite{synquid} synthesizes programs using polymorphic refinement types. At the heart of their procedures is a round-trip type checking mechanism that interleaves top-down and bottom-up propagation of type information. \toolname employs a similar approach but with two key differences: First our approach propagates necessary conditions to ensure \emph{type-compatibility} as opposed to subtyping, and second, we apply type-directed lemma learning to further speedup synthesis over multiple specifications. 

%With the recent advances of refinement type~\cite{dependent,liquid}, people have been proposing techniques for generating programs from high-level specifications expressed as refinement types \cite{synquid, plpv09, tyde19, myth, resyn}. In particular, {\sc Myth} relies on both the type information and the input-output examples to direct the synthesis procedure \cite{myth}. The follow-up work, {\sc Myth2}~\cite{myth2} further extends the expressiveness of type specifications to by interpreting examples as refinement types. Another closely related work is {\sc Synquid}~\cite{synquid}, which synthesizes programs using polymorphic refinement types. The core of the synthesis procedure relies on the notion of round-trip type checking mechanism that interleaves top-down and bottom-up propagation of type information. In the paper we deploy the similar idea under the visualization type system and we further extend the technique with type-directed lemma learning. However, we differentiate from prior work by  emphasizing using the type system to prune out infeasible program rather than to guarantee that the program synthesized indeed subtypes the goal type.

% \vspace{-0.2cm}
\paragraph{Program synthesis from NL} Beyond data visualization, there have also been proposals for performing program synthesis directly from natural language \cite{gpt3, sqlizer, lin-etal-2020-bridging, opsynth}. These techniques can mainly divided into two categories: end-to-end parsing  vs parse-then-synthesize techniques. Most of the recent work from the NLP community focuses on end-to-end parsing, using either powerful generic language models \cite{lewis-etal-2020-bart, gpt3} or domain-specific techniques targeting SQL \cite{lin-etal-2020-bridging, wang-etal-2020-rat}, spreadsheet formulas \cite{nlyze}, and bash commands \cite{lin-etal-2018-nl2bash}. Parse-then-synthesize approaches, of which \toolname is an instance,  first parse the natural language into an intermediate specification such as a sketch ~\cite{sqlizer} or a function declaration~\cite{anycode} and then synthesize programs from these intermediate specifications. 
%\toolname belongs to the parse-then-synthesize category. However, in order to solve the recommendation challenge in the visualization domain, \toolname needs to efficiently explore the large design space to make recommended visualizations diverse. As the result, \toolname uses refinement type as the intermediate specification (as opposed to program sketches) and solves the synthesis problem by finding the type inhabitant of the specification.

%\paragraph{Multi-modal synthesis} Multi-modal specifications have the potential to scale up program synthesis for more challenging tasks. For example, prior work uses a combination of natural language and input-output examples to solve challenging regular expressions~\cite{regel}, web question-answering~\cite{webqa}, data wrangling~\cite{mars} and string manipulation programs~\cite{cps}, SQL queries~\cite{duoquest}, and temporal logic formulas~\cite{Ltltalk} and more. \toolname can also be used in a multi-modal synthesis setting to work with visualization by example tools~\cite{falx} or direct manipulation tools~\cite{DBLP:conf/uist/HempelLC19} to tackle more challenging problems.~\later{this paragraph is not that important, remove/shrink if we are running out of space} %Our technique can also be viewed as an instance of multi-modal synthesis that combines natural language queries, visualization guidelines, and relational data.

\paragraph{Lemma learning.} {\toolname's type directed lemma learning strategy is most similar to the conflict analysis procedures in CDCL-based synthesizers such as Neo and Concord \cite{neo, concord}. In particular, these conflict analysis procedures generate a conflict clause whenever the synthesizer determines a partial program is infeasible, and this clause is used prune many other infeasible partial programs. However, unlike \toolname's lemmas, these conflict clauses  are only useful for a {\emph single} synthesis task whereas \toolname's lemmas can be reused across \emph{multiple} synthesis tasks so long as they use the same input dataset.}

%\paragraph{Lemma learning in program synthesis} {\color{red} Synthesis techniques that uses Conflict-driven Clause Learning (CDCL) are closely related to Graphy as both of the works speed up the synthesis procedure by learning lemmas. Specifically, CDCL-based synthesizer, such as Neo and Concord \cite{neo, concord}, learn conflict clauses to avoid enumerating programs that cause the same conflict in the future during a single task. In contrast, the lemma learned in \toolname can be reused across multiple synthesis tasks as long as they use the same input dataset.}
\section{Conclusion}

We have presented \toolname, a new synthesis-based NLI for visualizations. We evaluated \toolname on 3 datasets with over 700 natural language queries and found it significantly outperforms prior state-of-the-art approaches in top-1, top-5, and top-10 accuracy.

%% Acknowledgments
\begin{acks}                            %% acks environment is optional
                                        %% contents suppressed with 'anonymous'
  %% Commands \grantsponsor{<sponsorID>}{<name>}{<url>} and
  %% \grantnum[<url>]{<sponsorID>}{<number>} should be used to
  %% acknowledge financial support and will be used by metadata
  %% extraction tools.
  We would like to thank Anders Miltner, Ben Mariano, Xi Ye, fellow graduate students on GDC 5S, and the anonymous reviewers for their help and feedback for this paper. This material is based upon work supported by the \grantsponsor{GS100000001}{National Science Foundation}{http://dx.doi.org/10.13039/100000001} under grant number \grantnum{GS100000001}{CCF-1811865}, \grantnum{GS100000001}{CCF-1712067}, \grantnum{GS100000001}{CCF-1762299},  \grantnum{GS100000001}{CCF-1918889}, Google under the Google Faculty Research Grant, as well as Facebook, Amazon and RelationalAI.
\end{acks}

%% Bibliography
\bibliography{main}

%% Appendix
\pagebreak
\appendix
\section{Proof}

% \begin{lemma}{\bf{(Correctness of \textsf{PruneGoal})}}\label{lemma:prune} TODO: Whatever goal type that is pruned out is guarantee not to able to synthesize any program that type checks up to the enumeration depth. (since this is depend on the enumeration depth, how should we write this?)
% 
% \end{lemma}

\setcounter{theorem}{0}

\begin{theorem} {\bf (Soundness of {\sc TypeIncompatible})}
Let $P$ be a partial program with input type $\inputtype$ and top level goal-type $\outputtype$. If $\textsc{TypeIncompatible}(P)$ returns true, then for any completion $P'$ of $P$ with type $\ftype{\inputtype'}{\outputtype'}$, $ \vdash \ftype{\inputtype'}{\outputtype'} \not\sim  \ftype{\inputtype}{\outputtype}$. 
\label{theorem:typeincompatible}
\end{theorem}

\begin{proof}
 Suppose \textsc{TypeIncompatible}($\prog$) returns true. Then by lines 3-5 of Figure ~\ref{fig:typeinfeasible}, there exists a node $n \in$ \textsf{Nodes}($\prog$) such that the type of the complete program $P(n)$ is incompatible with $\textsf{Goal}(n)$. Since $\prog'$ is a completion of $\prog$, $\textsf{Nodes}(\prog) \subseteq \textsf{Nodes}(\prog')$. Hence $n \in \textsf{Nodes}(\prog')$. Given the subprogram rooted at $n$ does not satisfy its goal type, we know that $\prog'$ does not satisfy its top level goal type $\outputtype$, i.e. $\not\vdash \outputtype' \comp \outputtype$. Therefore, $\vdash \ftype{\inputtype'}{\outputtype'} \not\sim  \ftype{\inputtype}{\outputtype}$.

\end{proof}

\begin{theorem} {\bf (Soundness of {\sc ViolatesLemma})}  Let $\prog$ be a partial program whose top-level goal type is $\outputtype$, $D$ an input table, and $\Phi$ a set of learned lemmas. If $\textsc{ViolatesLemma}(\prog, \Phi)$ returns true, then for any completion $\prog'$ of $\prog$, $\prog'(D)$ is not an inhabitant of $\outputtype$.
\label{theorem:violateslemma}
\end{theorem}

\begin{proof}
Suppose {\sc ViolatesLemma}($P$,$\Phi$) returns true. Then by lines 2-5 of Figure ~\ref{fig:violatelemma}, there exists some hole $h \in {\sf Holes }(P)$ and some $(G, R) \in \Phi$ such that $\vdash {\sf GoalType}(h) <: G \land \vdash {\sf GoalType}(h) \not \comp R$. 
Given a completion $\prog'$ of $\prog$, let node $n \in {\sf Nodes}(P')$ be the node instantiated from $h$ in $\prog$ with a terminal symbol. We note the goal type of $n$ to be $\textsf{GoalType}(n)$. Since node $n$ is instantiated from $h$, $\textsf{GoalType}(n) = \textsf{GoalType}(h)$.

Given $\vdash {\sf GoalType}(h) <: G$, ${\sf GoalType}(h) \not\sim R$, following the definition of the synthesis lemma, we know that $\not\vdash \prog'(n)(D) : {\sf GoalType}(h)$. 
% Since $\vdash {\sf GoalType}(h) \not \comp R$, it must be the case that $\not \vdash \prog'(n)(D) : {\sf GoalType}(h)$. \later{JC: these two sentence seems to be contradicting each other. can you clarify it bit?}
Then $\prog'(n)$ is not an inhabitant of its goal type. Since $\prog'(n)$ is a subprogram of $\prog'$, this means that $\prog'$ is not an inhabitant of its top level goal type $\outputtype$.
\end{proof}

\begin{theorem}{\bf{(Soundness of {\sc SynthesizeVis})}}
Suppose $\textsc{SynthesizeVis}((\rtsym_p, \rtsym_t), D)$ returns a set of programs $\mathcal{S}$. Then for each visualization program $\prog_v = \prog_p \circ \prog_t \in \mathcal{S}$, $\prog_t(D) \vDash \rtsym_t$ and $\prog_v(D) \vDash \rtsym_p$.
\end{theorem}

\begin{proof}
It follows from line 8 of Figure ~\ref{fig:synthesizevis} that a program $\prog_p \circ \prog_t$ is only appended to $R$ if $\prog_t(D) \vDash \rtsym_t$ and $\prog_p(\prog_t(D)) \vDash \rtsym_p$. 
\end{proof}

\begin{lemma}

{
Let $\mathcal{G}, \rtsym_{in}, \rtsym_{out}, \rhd, \lhd$ be inputs to $\textsc{SynthesizeGoal}$, and let $D$ be an input table with $D \vDash \inputtype$.
Let $\prog$ be a program such that there exists a complete program $\prog'$ with most precise type $\ftype{\rtsym_{in}'}{\rtsym_{out}'}$ that can be derived from $\prog$ with $\prog'(D) \vDash \outputtype$.Then \textsc{SynthesizeGoal} will add $\prog$ to the worklist $\worklist$.

}\label{lemma:synthgoal}

\end{lemma}

\begin{proof}

By induction on the number of terminals $m$ in the AST of program $P$.

\textbf{Base Case:} $m = 0$. The only such program $\prog_0$ with $0$ terminals  is a partial program with one hole that is annotated with the goal output type $\rtsym_{out}$. This program is added to $\worklist$ on line 3 of Figure ~\ref{fig:goal_synthesis}.

\textbf{Inductive Hypothesis:} Assume this lemma holds for all programs whose ASTs have less than $m$ terminals, where $m \geq 0$.

\textbf{Inductive Case:} Suppose $\prog_{m+1}$ has $m+1$ terminals. Then there is some program $P_{m'}$ with $m' \leq m$ terminals and some production $\alpha$ such that expanding $P_{m'}$ with $\alpha$ produces $\prog_{m+1}$ .

Since $\prog'$ can be derived from $\prog_{m+1}$ by $\prog_{m+1} \xRightarrow[]{*} \prog'$, $\prog'$ can also be derived from $\prog_{m'}$ by $\prog_{m'} \xRightarrow[]{\alpha} \prog_{m+1} \xRightarrow[]{*} \prog'$. Thus, by inductive hypothesis, $\prog_{m'}$ is added to $\worklist$. Then at some point $\prog_{m'}$ will be dequeued from $\worklist$ on line 5 of Figure ~\ref{fig:goal_synthesis}. The \textsc{Expand} procedure on line 8 will identify $\alpha$ as a possible production, and will expand $\prog_{m'}$ to  $\prog_{m+1}$. 

Note that $\prog'(D) \vDash \outputtype$ and $D \vDash \inputtype$ imply $\vdash \inputtype' \comp \inputtype \wedge \vdash \outputtype' \comp \outputtype$, meaning there exists a completion of $\prog_m$ such that $\vdash \ftype{\inputtype'}{\outputtype'} \comp (\ftype{\inputtype}{\outputtype})$. Then, by contrapositive of Theorem ~\ref{theorem:typeincompatible}, \textsc{TypeIncompatible}($\prog_{m+1}$) will return false.

Similarly, since there exists a completion of $\prog_{m+1}$ that inhabits its goal type, by contrapositive of Theorem ~\ref{theorem:violateslemma}, \textsc{ViolatesLemma}($\prog_{m+1}$, $\Phi$) will return false.
$\prog_{m+1}$ will thus be added to $\worklist$ on line 16.

\end{proof}

\begin{lemma}
{\bf (Completeness of } $\textsc{SynthesizeGoal}$)\label{lemma:synthesizegoal} 
Let $\mathcal{G}, \rtsym_{in}, \rtsym_{out}, \rhd, \lhd$ be inputs to $\textsc{SynthesizeGoal}$, let $D$ be an input table with type $D \vDash \inputtype$, and let $\mathcal{S}$ be the set of programs with respect to the typing environment $\Env$ returned by $\textsc{SynthesizeGoal}$. Then for any complete program $\prog$ with the most precise type $\ftype{\rtsym_{in}'}{\rtsym_{out}'} $ in the grammar $\mathcal{G}$ such that 
% $\subty{\Env}{\rtsym_{out}'}{\rtsym_{out}}$ and
% $\subty{\Env}{\rtsym_{in}'}{\rtsym_{in}}$
$\prog(D) \vDash \outputtype$
, $\prog \in \mathcal{S}$.

\end{lemma}

\begin{proof}
By Lemma ~\ref{lemma:synthgoal}, $\prog$ is added to $\worklist$. Note that the only termination condition for the while loop on line 4 of Figure ~\ref{fig:goal_synthesis} is that we exhaust $\worklist$. Thus, $\prog$ will be dequeued on line 5 of Figure ~\ref{fig:goal_synthesis} at some point. Since $\prog$ is complete, we check if $\Env \vdash \inputtype \rhd \inputtype' \wedge \Env \vdash \outputtype' \lhd \outputtype$ (line 7). Since $\prog(D) \vDash \outputtype$, we have 
$\subty{\Env}{\rtsym_{out}'}{\rtsym_{out}}$ and $\subty{\Env}{\rtsym_{in}}{\rtsym_{in'}}$. It then follows that $\Env \vdash \outputtype' \lhd \outputtype$ and $\Env \vdash \inputtype \lhd \inputtype'$, so $\prog$ will be added to $\res$ on line 8 of Figure ~\ref{fig:goal_synthesis}.
\end{proof}

\begin{theorem}{\bf{(Completeness of {\sc SynthesizeVis})}}
Given a specification $(\rtsym_p, \rtsym_t)$ and input table $D$, if there is a visualization program $\prog_p \circ \prog_t$ such that $\prog_t(D) \vDash \rtsym_t$ and $\prog_p(\prog_t(D) \vDash \rtsym_p$, then $\prog_p \circ \prog_t \in \textsc{SynthesizeVis}((\rtsym_p, \rtsym_t), D)$
\end{theorem}

\begin{proof}
Let $P_p$ and $P_t$ be programs such that $\prog_t(D) \vDash \rtsym_t$ and $\prog_p(\prog_t(D)) \vDash \rtsym_p$. From Lemma ~\ref{lemma:synthesizegoal}, we know that $\prog_p$ is in the set of programs returned by \textsc{SynthesizeGoal} on line 3 of Figure ~\ref{fig:synthesizevis}. Also by Lemma ~\ref{lemma:synthesizegoal}, we know that $\prog_t$ is in the set of programs returned by \textsc{SynthesizeGoal} on line 6 of Figure ~\ref{fig:synthesizevis}. Since $\prog_t(D) \vDash \rtsym_t$ and $\prog_v(D) \vDash \rtsym_p$, it follows from line 9 of Figure~\ref{fig:synthesizevis} that $\prog_p \circ \prog_t \in \textsc{SynthesizeVis}((\rtsym_p, \rtsym_t), D)$. 
\end{proof}

\section{Complete Typing Rules}

%\subsection{Semantic of the Syntactic Constraints}

\subsection{Intersection of Types}
Figure \ref{fig:sqcap} presents inference rules that describe how we computing the intersection type of two refinement types. 
\begin{figure}
    \centering
    \small
    \begin{mathpar}
    \inferrule*[Left=Subtype]{\vdash \btsym_1 \subtype \btsym_2 }{\vdash \btsym_1 \intersects \btsym_2: \btsym_1} \and
    \inferrule*[Left=Symmetry]{\vdash \btsym_1 \intersects \btsym_2 : \btsym' }{\vdash \btsym_2 \intersects \btsym_1: \btsym'} \\
    %\inferrule*[Left=List-glb]{\vdash \btsym_1 \subtype \btsym_2}{ \vdash List[\btsym_1] \intersects List[\btsym_2] : List[\btsym_1]}\\
    \inferrule*[Left=Table]{\btsym_1 = {\sf Table}(\sigma_1) \ \ \ \btsym_2 = {\sf Table}(\sigma_2) \\\\
    \sigma_{\texttt{shared}} = \{c : \btsym_1 \wedge \btsym_2 \mid c : \btsym_1 \in \sigma_1, \  c : \btsym_2 \in \sigma_2 \} \\\\
    \tau = {\sf Table}(\sigma_{{\tt shared}} \cup (\sigma_1 \ \Delta \ \sigma_2))
    }{\vdash \btsym_1 \intersects \btsym_2 : \tau  } \\\\
    \inferrule*[Left=Refinement]{\vdash \btsym_1 \intersects \btsym_2: \btsym'}{\Env \vdash \rtype{\btsym_1}{\phi_1} \intersects \rtype{\btsym_2}{\phi_2}) : \rtype{\btsym'}{\phi_1 \wedge \phi_2}}
    \end{mathpar}
    \caption{Intersection of base and refinement types. $\sigma_1 \ \Delta \ \sigma_2$ is the "symmetric difference" between two schemas i.e., the union of different columns}
    \label{fig:sqcap}
\end{figure}

\subsection{Typing Rules}

Figures \ref{fig:tr1} and \ref{fig:tr2} present our complete set of typing rules.

\begin{figure}
    \centering
    \begin{mathpar}
    \inferrule*[Left=Sub]{\Env \vdash \rtsym_1 \subtype \rtsym_2 \\\\
    \Env \vdash e : \rtsym_1
    }{\Env \vdash e: \rtsym_2}\\
    
    \inferrule*[Left=Bar]{
    \Env(T) = \rtype{\tau_T}{\phi_T}\\\\
    \Env \vdash \tau_T : \mathsf{Table}(\{ \colx: \tau_x, \coly: \textsf{Quantitative}, \colco: \tau_{color}, \colsub: \tau_{subplot} \})\\\\
      \tau_x \neq \sf{Continuous} \ \ \  \tau_{color} \neq \sf{Continuous} \ \ \  \tau_{subplot} \neq \sf{Continuous}
    \\\\
    \textsf{Encode}(\Env) \land \textsf{Encode}(\phi_T) \Rightarrow |(\nu, \{\colx, \colco, \colsub\})| \geq |(\nu, \{\coly\})|)
    }{\tjudg{\Env}{{\sf Bar}(T, \colx, \coly, \colco, \colsub): \rtype{ \textsf{BarPlot}}{\bigwedge_{e \in \{{\sf x}, {\sf y}, {\sf color}, {\sf subplot}\} } \prov(\nu.e, T.c_{e})}}} \\

    \inferrule*[Left=Scatter]{
    \Env(T) = \rtype{\tau_T}{\phi_T}\\\\
    \Env \vdash \tau_T : \mathsf{Table}(\{ \colx: \tau_x, \coly: \tau_y, \colco: \top, \colsub: \tau_{subplot} \}) \\\\
      \tau_x \neq \sf{Nominal} \ \ \ \tau_y \neq \sf{Temporal} \ \ \ \tau_y \neq \sf{Nominal} \ \ \ \tau_{subplot} \neq \sf{Continuous}
    }{\tjudg{\Env}{{\sf Scatter}(T, \colx, \coly, \colco, \colsub): \rtype{ \textsf{ScatterPlot}}{\bigwedge_{e \in \{{\sf x}, {\sf y}, {\sf color}, {\sf subplot}\} } \prov(\nu.e, T.c_{e})}}} \\
    
    \inferrule*[Left=Line]{
    \Env(T) = \rtype{\tau_T}{\phi_T}\\\\
    \Env \vdash \tau_T : \mathsf{Table}(\{ \colx: \tau_x, \coly: \textsf{Quantitative}, \colco: \tau_{color},  \colsub: \tau_{subplot} \})\\\\
    \tau_x \neq \sf{Nominal} \ \ \ 
    \tau_{color} \neq \sf{Continuous} \ \ \ \tau_{subplot} \neq \sf{Continuous}
    \\\\
    \textsf{Encode}(\Env) \land \textsf{Encode}(\phi_T) \Rightarrow |\nu, \{\colx, \colco, \colsub\})| \geq |(\nu, \{\coly\})|
    }{\tjudg{\Env}{{\sf Line}(T, \colx, \coly, \colco, \colsub): \rtype{ \texttt{LinePlot}}{\bigwedge_{e \in \{{\tt x}, {\tt y}, {\tt color}, {\tt subplot}\} } \prov(\nu.e, T.c_{e})}}} \\

    \inferrule*[Left=Area]{
    \Env(T) = \rtype{\tau_T}{\phi_T}\\\\
    \Env \vdash \tau_T <: \mathsf{Table}(\{ \colx: \tau_x, \coly: \textsf{Quantitative}, \colco: \tau_{color},  \colsub: \tau_{subplot} \})\\\\
    \tau_x \neq \sf{Nominal}  \ \ \ 
    \tau_{color} \neq \sf{Continuous} \ \ \ \tau_{subplot} \neq \sf{Continuous}
    \\\\
    \textsf{Encode}(\Env) \land \textsf{Encode}(\phi_T) \Rightarrow |\nu, \{\colx, \colco, \colsub\})| \geq |(\nu, \{\coly\})|
    }{\tjudg{\Env}{{\sf Area}(T, \colx, \coly, \colco, \colsub): \rtype{ \textsf{AreaPlot}}{\bigwedge_{e \in \{{\tt x}, {\tt y}, {\tt color}, {\tt subplot}\} } \prov(\nu.e, T.c_{e})}}} \\

    \inferrule*[Left=Bin]{
    \Env \vdash e: \rtype{\tau_t}{\phi} \ \ \ \ \text{where $\tau_t = \textsf{Table}(\{\ldots, \coltarg: \tau_{\tt tgt}, \ldots\})$} \\\\
    \vdash \tau_{\tt tgt} <: \sf{Quantitative}\ \ \  \tau' = \tau_t[\coltarg \mapsto \sf{Discrete}] \\\\
    \phi_1 = \phi \ \forget \ \textsf{Terms}(\phi, \coltarg) \ \ \ \  \phi_2 = \phi_1 \  \forget \ \prov(\nu.\coltarg, \mathtt{bin}) \\\\
    \phi' = \phi_2 \land |(\nu, \{\coltarg\})| = n \land \prov(\nu.\coltarg, \mathsf{bin})
    }{\tjudg{\Env}{\textsf{bin}(e, n, \coltarg)} : \rtype{\tau'}{\phi'}}\\\\

    \end{mathpar}
    \vspace*{-1.0cm}
    
    \caption{Typing Rules }
    \label{fig:tr1}
\end{figure}
\begin{figure}
    \centering 
    \begin{mathpar}
    
    \inferrule*[Left=Filter]{
    \Env \vdash e: \rtype{\tau_t}{\phi} \ \ \ \ \text{where} \ \tau_t = \textsf{Table}(\{c_0: \tau_0, \ldots, c_n: \tau_n \}) \\\\
    \phi' = \phi \ \forget \ \textsf{Terms}(\phi, \{c_1, \ldots, c_n\})
    %\phi_1 = \phi \ \forget \ \textsf{Terms}(\phi, \coltarg) \ \ \ \  \phi_2 = %\phi_1 \  \forget \ \prov(\nu.\coltarg, \mathtt{filter}) \\\\
    %\phi' = \phi_2 \land
   % |\nu| = |{\sf Filter}(\nu, val_1 \ op \ val_2)| %\land 
    %\prov(\nu.\coltarg, \mathtt{filter})
    }{\tjudg{\Env}{\textsf{filter}(e, val_1 \ op \ val_2)} : \rtype{\tau}{\phi'}}\\\\

    \inferrule*[Left=Select]{
    \Env \vdash e: \rtype{\tau_t}{\phi} \ \ \ \ \text{where} \ \tau_t = \textsf{Table}(\{c_0: \tau_0, \ldots, c_n: \tau_n \})
    \\\\
    \textsf{size}(\overline{\colkey}) = k \ \ \ 
    \tau' = \textsf{Table}(\{c_0': \tau_0', \ldots, c_k': \tau_k' \}) , c_i' \in \overline{\colkey}
    }{\tjudg{\Env}{\textsf{select}(e, \overline{\colkey})} : \rtype{\tau'}{\phi}}\\\\

    \inferrule*[Left=Mutate]{
    \Env \vdash e: \rtype{\tau_t}{\phi} \ \ \ \ \text{where} \ \tau_t = \textsf{Table}(\{c_0: \tau_0, \ldots, c_n: \tau_n \}), \ 
    \forall \ 0 \leq i \leq n, \ c_i \neq \coltarg
    \\\\
    \coltarg \not\in \overline{\colkey} \ \ \ \tau' = {\sf Table}(\{c_0': \tau_0', \ldots, c_k': \tau_k', \coltarg : \top \}) \ \ c_i' \in \overline{\colkey} \\\\
    \phi_1 = \phi \ \forget \ \textsf{Terms}(\phi, \coltarg) \ \ \ \  \phi_2 = \phi_1 \  \forget \ \prov(\nu.\coltarg, \mathsf{mutate}) \\\\
    \phi' = \phi_2 \land |(\nu, \{\coltarg\})| \leq |(\nu, \overline{\colkey})|  \land \prov(\nu.\coltarg, \mathsf{mutate})
    }{\tjudg{\Env}{\textsf{mutate}(e, \coltarg, op, \overline{\colkey})} : \rtype{\tau'}{\phi'}}\\\\

    \inferrule*[Left=Summ-Mean]{
    \Env \vdash e: \rtype{\tau_t}{\phi} \ \ \ \ \text{where $\tau_t = \textsf{Table}(\{\ldots, \coltarg: \tau_{\sf tgt}, \ldots\})$} \\\\
    \coltarg \not\in \overline{\colkey} \ \ \ \tau' = {\sf Table}(\{c_0': \tau_0', \ldots, c_k': \tau_k', \coltarg : \btsym_{\sf tgt} \}) \ \ c_i' \in \overline{\colkey} \\\\
    \vdash \tau_{\sf tgt} : \textsf{Quantitative}\ \ \  \tau' = \tau'[\coltarg \mapsto \sf{Continuous}] \\\\
    \phi_1 = \phi \ \forget \ \textsf{Terms}(\phi, \coltarg) \ \ \ \  \phi_2 = \phi_1 \  \forget \ \prov(\nu.\coltarg, \mathsf{mean}) \\\\
    \phi' = \phi_2 \land |(\nu, \{\coltarg\})| \leq |(\nu, \overline{\colkey})| \land \prov(\nu.\coltarg, \mathsf{mean})
    }{\tjudg{\Env}{\textsf{summarize}(e, \overline{\colkey}, \mathsf{mean}, \coltarg)} : \rtype{\tau'}{\phi'}}\\\\
    
    \inferrule*[Left=Summ-Count]{
    \Env \vdash e: \rtype{\tau_t}{\phi} \ \ \ \ \text{where $\tau_t = \textsf{Table}(\{\ldots, \coltarg: \tau_{\sf tgt}, \ldots\})$} \\\\
    \coltarg \not\in \overline{\colkey} \ \ \ \tau' = {\sf Table}(\{c_0': \tau_0', \ldots, c_k': \tau_k', \coltarg : \btsym_{\tt tgt} \}) \ \ c_i' \in \overline{\colkey} \\\\
      \tau' = \tau'[\coltarg \mapsto \sf{Discrete}] \\\\
    \phi_1 = \phi \ \forget \ \textsf{Terms}(\phi, \coltarg) \ \ \ \  \phi_2 = \phi_1 \  \forget \ \prov(\nu.\coltarg, \mathsf{count}) \\\\
    \phi' = \phi_2 \land |(\nu, \{\coltarg\})| \leq  |(\nu, \overline{\colkey})| \land \prov(\nu.\coltarg, \mathsf{count})
    }{\tjudg{\Env}{\textsf{summarize}(e, \overline{\colkey}, \mathsf{count}, \coltarg)} : \rtype{\tau'}{\phi'}}\\\\
    
    \inferrule*[Left=Summ-Sum]{
    \Env \vdash e: \rtype{\tau_t}{\phi} \ \ \ \ \text{where $\tau_t = \sf{Table}(\{\ldots, \coltarg: \tau_{\sf tgt}, \ldots\})$} \\\\
    \coltarg \not\in \overline{\colkey} \ \ \ \tau' = {\sf Table}(\{c_0': \tau_0', \ldots, c_k': \tau_k', \coltarg : \btsym_{\sf tgt} \}) \ \ c_i' \in \overline{\colkey} \\\\
     \vdash \tau_{\tt tgt} : \sf{Quantitative} \ \ \  \tau' = \tau'[\coltarg \mapsto \sf{Continuous}]  \\\\
    \phi_1 = \phi \ \forget \ \textsf{Terms}(\phi, \coltarg) \ \ \ \  \phi_2 = \phi_1 \  \forget \ \prov(\nu.\coltarg, \mathsf{sum}) \\\\
    \phi' = \phi_2 \land |(\nu, \{\coltarg\})| \leq |(\nu, \overline{\colkey})| \land \prov(\nu.\coltarg, \mathtt{sum})
    }{\tjudg{\Env}{\textsf{summarize}(e, \overline{\colkey}, \mathsf{sum}, \coltarg)} : \rtype{\tau'}{\phi'}}\\\\
    
    \end{mathpar}
    \vspace*{-1.0cm}
    
    \caption{Typing Rules }
    \label{fig:tr2}
\end{figure}

\section{Formula Encoding} \label{sec:encoding}

In this section, we describe our {\sf Encode} procedure which encodes qualifiers in our refinement type system as formulas in the combined theory of equality, uninterpreted functions, and integers.

\paragraph{Formula Language} Figure \ref{fig:formulalang} presents the syntax of our encoded formulas as a context free grammar. Note that many of the terminals in our refinement type system also appear in this grammar but now have different semantics. The symbols $|\cdot|$, ${\sf Filter}$ and ${\sf Proj}$ refer to relational operators in our refinement type system, but correspond to uninterpreted functions in the formula language. We also import column names $c$ into our formula language, but they refer to object constants. Finally, in our refinement type system, ${\sf Proj}$ takes a list of column names as its second input, whereas in this formula language, ${\sf Proj}$ takes two inputs where the second argument is an object constant. 

We formalize our encoding procedure as inference rules incorporating judgments of the form.

\[
\encenv \vdash t \rightsquigarrow e, \encenv'
\]

where $\encenv$ is an environment that maps terms in our refinement type language to terms in the formula language. This judgment means:  given an environment $\encenv$ and a term $t$ in the refinement type language, {\sf Encode} returns the corresponding term $e$ in our formula language along with an updated environment $\encenv'$. This formalization is presented in Figure \ref{fig:encode}. Most of the rules are straightforward, so we highlight the most interesting. 
\begin{itemize}
    \item {\bf Syntactic.} Our encoding scheme associates each {syntactic} constraint $\pi(x.\enc, \provop)$  with a unique propositional variable ({\sc Syn-1}, {\sc Syn-2})
    \item {\bf Projection.} For each set of columns $\overline{c}$, we associate a fresh object constant $s$. In particular, $\{c_1, c_2\}$ and $\{c_2, c_1\}$ are associated with the same fresh constant as they represent the same sets ({\sc Proj-1}, {\sc Proj-2}). 
    \item {\bf Filter.} Every filter operator $op$ is assigned a fresh function constant $f_i$, and each value $val$ is assigned a fresh object constant $s$. Thus,  the filter operation $c \ op \ val$ is treated as the function application $f_i(c, s)$.
\end{itemize}

\paragraph{{\bf Correctness of {\sc SynthesizeVis}}} Since the combined theory of integers and equality with uninterpreted functions over-approximates the semantics of our qualifiers, one may wonder whether {\sc SynthesizeVis} is still complete i.e., doesn't prune correct programs. We now argue that the procedure is still complete. First, we note that our compatibility check will not prune feasible programs since (1) our encoding is the conjunction of formulas that over-approximate their corresponding qualifiers, and (2) our compatibility rule checks that the formula is satisfiable. Second, our subtyping checks are used to prune prune plotting programs whose output type is not a subtype of the goal type specification produced by the parser (line 3 {\sc SynthesizeVis}). In that case, the logical qualifiers in the goal type specification are only {syntactic} constraints, which are \emph{precisely encoded} as boolean constraints. As such, our subtyping check will not prune correct programs there.

\begin{figure}[!htb]
\centering
\small
\[
\begin{array}{r l l l }
    \textbf{Formula Language} \\ 
    {\sf Formula} \  F := & \oplus(F_1, \ldots, F_n) \ | \ E \between E \ | \  p  \\
    {\sf Expression} \  E := & |T| \ | \ {\sf Max}(T) \ | \ {\sf Min}(T) \ | \ x \ | \ S\\
    {\sf Table Function} \ T := & {\sf Proj}(T, S) \ | \ {\sf Filter}(T, G(c, S)) \ | \  x \\
    {\sf Object Constants} \ S := & a \ | \  b \ | \  \ldots \\
    {\sf Function Constants} \ G := & f_1 \ | \  f_2 \ | \  \ldots \\
\end{array}
\]
\caption{Our formula language where $\oplus \in \{\land, \lor, \neg\}$, and $\between \in \{=, \leq, >\}$ . Formulas in this language are in the combined theory of equality, uninterpreted functions, and integers. In particular, $|\cdot|$, {\sf Proj}, {\sf Filter}, {\sf Max}, and {\sf Min} are uninterpreted functions, and $p$ represents propositional variables.}
\label{fig:formulalang}
\end{figure}

\begin{figure}
    \centering 
    \begin{mathpar}
    \inferrule*[Left=Logical Operators]{
      \encenv \vdash \phi_1 \rightsquigarrow F_1, \encenv_1 \\\\
      \encenv_1 \vdash \phi_2 \rightsquigarrow F_2, \encenv_2 \\\\
      \ldots \\\\
      \encenv_{n-1} \vdash \phi_n \rightsquigarrow F_n, \encenv_n
    } {
     \encenv \vdash \oplus(\phi_1, \ldots, \phi_n) \rightsquigarrow \oplus(F_1, \ldots, F_n), \encenv_n
    } \\\\
    \inferrule*[Left=Semantic Term]{
        \encenv \vdash \theta_1 \rightsquigarrow E_1, \encenv_1 \ \ \ \ \encenv_1 \vdash \theta_2 \rightsquigarrow E_2, \encenv_2
    } {
     \encenv \vdash \theta_1 \between \theta_2 \rightsquigarrow E_1 \between E_2, \encenv_2
    } \ \ \ \ \ \ \ \ \ \
    \inferrule*[Left=Card]{
        \encenv \vdash \tablesym \rightsquigarrow T, \encenv'
    } {
     \encenv \vdash |\tablesym| \rightsquigarrow | T |, \encenv'
    } \\\\
    \inferrule*[Left=Syn-1]{
        \prov(x.\enc, \provop) \in dom(\encenv)
    } {
     \encenv \vdash \prov(x.\enc, \provop) \rightsquigarrow \encenv(\prov(x.\enc, \provop))
    } \ \ \ \ \ \ \ \ \ \ \ \ 
     \inferrule*[Left=Syn-2]{
        \prov(x.\enc, \provop) \not\in dom(\encenv) \ \ \ \ {\sf fresh}\ p' 
    } {
     \encenv \vdash \prov(x.\enc, \provop) \rightsquigarrow p', \encenv[\prov(x.\enc, \provop) \gets p']
    } \\\\
    \inferrule*[Left=Var-1]{
        x \in dom(\encenv)
    } {
     \encenv \vdash x \rightsquigarrow \encenv(x)
    } \ \ \ \ \ \ \ \ \ \ \ \ \ \ \ \ \ \ 
    \inferrule*[Left=Var-2]{
        x \not\in dom(\encenv) \ \ \ \ {\sf fresh} \ v
    } {
     \encenv \vdash x \rightsquigarrow v, \encenv[x \gets v]
    } \\\\
    \inferrule*[Left=Filter Op-1]{
        op \in dom(\encenv)
    } {
     \encenv \vdash op \rightsquigarrow \encenv(x)
    } \ \ \ \ \ \ \ \ \ \ \ \ \ \ \ \ \ \ \ \ 
     \inferrule*[Left=Filter Op-2]{
        op \not\in dom(\encenv), {\sf fresh} \ f_i
    } {
     \encenv \vdash op \rightsquigarrow f_i, \encenv[op \gets f_i]
    }
    \\ \\
    \inferrule*[Left=Max]{
        \encenv \vdash \tablesym \rightsquigarrow T, \encenv'
    } {
     \encenv \vdash {\tt max}(\tablesym) \rightsquigarrow {\sf Max}(T), \encenv'
    } \ \ \ \ \ \ \ \ \ \ \ \ 
    \inferrule*[Left=Min]{
        \encenv \vdash \tablesym \rightsquigarrow T, \encenv'
    } {
     \encenv \vdash {\tt min}(\tablesym) \rightsquigarrow {\sf Min}(T), \encenv'
    } \\ \\
    \inferrule*[Left=Proj-1]{
        \encenv \vdash \tablesym \rightsquigarrow T, \encenv' \ \ \ \ \overline{c} \not\in dom(\encenv') \ \ \ \ {\sf fresh} \  s
    } {
     \encenv \vdash {\sf Proj}(\tablesym, \overline{c}) \rightsquigarrow {\sf Proj}(T, s), \encenv'[\overline{c} \gets s]
    } \\ \\
    \inferrule*[Left=Proj-2]{
        \encenv \vdash \tablesym \rightsquigarrow T, \encenv' \ \ \ \ \overline{c} \in dom(\encenv')
    } {
     \encenv \vdash {\sf Proj}(\tablesym, \overline{c}) \rightsquigarrow {\sf Proj}(T, \encenv'(\overline{c})), \encenv'
    }
    \\ \\
    \inferrule*[Left=Filter-1]{
        \encenv \vdash \tablesym \rightsquigarrow T, \encenv_1 \ \ \ \ \encenv_1 \vdash op \rightsquigarrow f_i, \encenv_2 \ \ \ val \in dom(\encenv)
    } {
     \vdash {\sf Filter}(\tablesym, c \ op \ val) \rightsquigarrow {\sf Filter}(T, f_i(c, \encenv(val))), \encenv_2
    } \\ \\
    \inferrule*[Left=Filter-2]{
        \encenv \vdash \tablesym \rightsquigarrow T, \encenv_1 \ \ \ \ \encenv_1 \vdash op \rightsquigarrow f_i, \encenv_2 \ \ \ val \not\in dom(\encenv) \ \ \ {\sf fresh} \ s
    } {
     \vdash {\sf Filter}(\tablesym, c \ op \ val) \rightsquigarrow {\sf Filter}(T, f_i(c, s)), \encenv_2[val \gets s]
    }
    \end{mathpar}
    \caption{Inference rules describing the {\sf Encode} procedure}
    \label{fig:encode}
\end{figure}

\section{Semantics of $\forget$ Operator}

In this section, we describe the $\forget$ operator introduced in Section \ref{sec:typingrules} in more detail. Here we assume the operator takes qualifiers of the form $\phi_s \land \phi_p$, where $\phi_s$ (resp $\phi_p$) is a boolean combination of semantic (resp. syntactic) constraints. We give the semantics of $\forget$ as a procedure shown in Figure \ref{fig:forget}. Given a qualifier $\phi_s \land \phi_p$  along with a set of terms $T$ to remove, we have $\phi_s \land \phi_p \forget T \equiv \textsc{Remove}(\phi_s \land \phi_p, T)$. 

We now describe this procedure in more detail. In line 2 we call {\sf Encode} (described in Section \ref{sec:encoding}) to encode our qualifiers as logical formulas. Then in lines 4-10, we iterate over all the terms $t$ in $T$. If $t$ is a semantic term (line 5), we generate a fresh variable $x$ and replace all occurrences of $t$ in $\Phi_s$ with $x$. We then update $\Phi_s$ to be its existential generalization (line 7). On the other hand, if $t$ is a syntactic term we perform a similar procedure, except $\Phi_p$ becomes a QBF formula. Finally, since the fresh variables we introduced represent the terms we want to forget, we compute the strongest formulas entailed by $\phi_s$ (resp. $\phi_p$) that don't contain the fresh variables. Since, QBF admits quantifier elimination \cite{quantelimqbf}, we can derive the strongest QFF formula entailed by $\phi_p$ by applying quantifier elimination (line 11). However, since our semantic constraints are expressed in the combined theory of equality, uninterpreted functions, and integers, which does not admit quantifier elimination, we instead the Cover algorithm \cite{cover} to compute the strongest formula (line 12). Finally, we convert $\Phi_s$ and $\Phi_p$ back into qualifiers and return their conjunction (line 13).

\paragraph{Optimization} In our implementation of {\sc Remove}, we apply two optimizations based on the following observations. First, we observe that nearly all our logical qualifiers are conjunctions of literals, and so we represent our qualifiers as sets of literals. As such, when removing a {syntactic} constraint, we simply remove all corresponding literals from the set. Second, as all our semantic constraints are linear inequalities, we encode our semantic constraints as formulas in Presburger Arithemtic. We then apply Fourier-Motzkin variable elimination when removing semantic terms. We illustrate these optimizations in the examples below.

\begin{example} (Forgetting {Syntactic Constraints})
Suppose we call $\phi \forget \prov(\nu.\coltarg, {\sf mean})$ where $\phi$ is  $|\nu, \{c_1\}| = 30 \land \neg \prov(\nu.\coltarg, {\sf mean}) \land \prov(\nu.c_2, {\sf count})$. Since this qualifier is a conjunction of literals, {\sc Remove} maintains a set of constraints $\{|\nu, {c_1}| = 30, \neg \prov(\nu.\coltarg, {\sf mean}), \prov(\nu.c_2, {\sf count})\}$. The only literal that corresponds to $\prov(\nu.\coltarg, {\sf mean})$ is $\neg \prov(\nu.\coltarg, {\sf mean})$ and so we drop that literal from the set. Thus, the formula returned by {\sc Remove} is $|\nu, {c_1}| = 30 \land \prov(\nu.c_2, {\sf count})$
\end{example}

\begin{example} (Forgetting Semantic Terms)
Suppose we call $\phi \forget |\nu, \{c_1\}|$ where $\phi$ is  $|\nu, \{c_1\}| \leq  |\nu, \{c_2\}| \land |\nu, \{c_1\}| = 30$. {\sc Remove} internally constructs an equisat formula $x_1 \leq x_2 \land x_1 = 30$ in the theory of integers where $x_1$, and $x_2$ are fresh variables that occur freely and $x_1 \to |\nu, \{c_1\}|$, $x_2 \to |\nu, \{c_2\}|$. It then applies Fourier-Motzkin variable elimination on $x_1$ to get the constraint $30 \leq x_2$, and then decodes the formula back to the qualifier $30 \leq |\nu, \{c_2\}|$.
\end{example}
\begin{figure}[!t]
\small
\vspace{-10pt}
\begin{algorithm}[H]
\begin{algorithmic}[1]
\Procedure{Remove}{$\phi_s \land \phi_p$, $T$}
\vspace{0.05in}
\Statex\Input{Logical qualifier of the form $\phi_s \land \phi_p$ where $\phi_s$ is a boolean combination of semantic constraints, and $\phi_p$ is a boolean combination of syntactic constraints. }
\Statex\Input{A set of terms $T$ to forget}
\Statex\Output{A qualifier $\phi$ that does not contain any terms in $T$ and whose encoding is the strongest QFF formula that is entailed by the encoding of $\phi_s \land \phi_p$.}
\vspace{0.05in}
\State $\Phi_s \assign {\sf Encode}(\phi_s)$; $\Phi_p \assign {\sf Encode}(\phi_p)$
\State $V_s = \{\}$; $V_p = \{\}$
\ForAll{$t \in T$}
\If{{\sf IsSemantic}($t$)}
\State $x \assign {\sf GetFreshVar}()$
\State $\Phi_s \assign \exists x. \Phi_s[x/{\sf Encode}(t)]$; $V_s = V_s \cup \{x\}$
\ElsIf{{\sf IsSyntactic}($t$)}
\State $p \assign {\sf GetFreshPropVar}()$
\State $\Phi_p \assign \exists p. \Phi_p[p/{\sf Encode}(t)]$; $V_p = V_p \cup \{p\}$
\EndIf
\EndFor
\State $\Phi_p \assign {\sf EliminateBoolVars}(\Phi_p, V_p)$
\State $\Phi_s \assign {\sf ComputeCover}(\Phi_s, V_s)$
\State \Return ${\sf Decode}(\Phi_s) \land {\sf Decode}(\Phi_p)$
\EndProcedure
\end{algorithmic}
\end{algorithm}
\vspace*{-1.0cm}
\caption{Procedure encoding semantics of $\forget$.}
\label{fig:forget}
\end{figure}

\section{Training Parameters for the Parser}

Each intent-and-slot-filling model in the parser is trained for 20 epochs using the AdamW optimizer \cite{loshchilov2018decoupled} with a batch size of 16. We use a learning rate of $2e-05$ for BERT, which is one of the standard suggested learning rates \cite{bert}, a warm-up ratio of $0.05$, and an input dropout rate of 0.2 to reduce overfitting \cite{dropout}. We trained the models on one NVIDIA Quadro RTX 8000 with 48GB of memory. Each training run of the model took around 10 minutes. 

\section{User Study Procedure}
In this section, we describe our user-study protocol in more detail. 

\paragraph{User study sessions} Our user study was completed in 12 sessions, one for each participant. The participants used the same laptop, which had Excel and \toolname installed, across all sessions. 

\paragraph{Participant Introduction} We started each user study session by first describing the task that the participant needed to accomplish. In particular, we asked them to reproduce two plots shown in Figure~\ref{fig:user_study} using both Excel and \toolname. We chose Excel as the baseline because it is a common data visualization tool that is designed to be accessible to non-expert users. In order to minimize the effect of knowledge transfer, we randomly determined whether a participant was first given access to \toolname or to Excel. 

\paragraph{Plot Selection.} To avoid biasing the study in \toolname's favor, we randomly selected two plots of different types from the Cars domain in {\sc NLVCorpus} for the participants to reproduce. To ensure that the plots were reasonably challenging we only selected among plots that required data aggregation operations in the table transformation. 
\begin{figure}
    \centering
    \includegraphics[width=0.45\textwidth, trim= 30 550 350 0, clip]{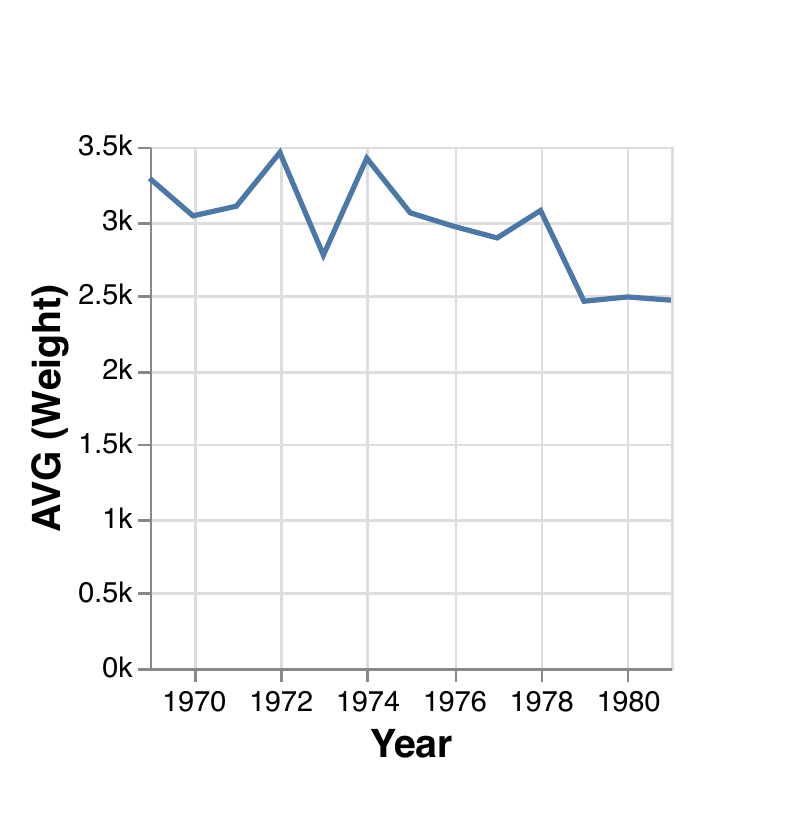}
    \includegraphics[width=0.45\textwidth, trim= 30 550 350 0, clip]{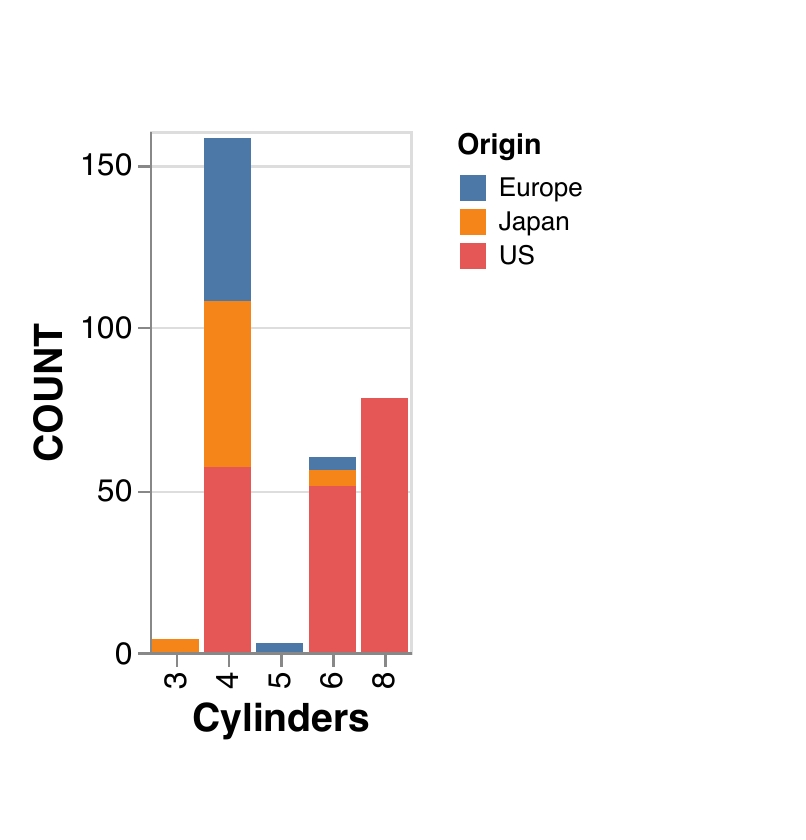}
    \caption{Two plots used in the user study}
    \label{fig:user_study}
\end{figure}

\paragraph{Dataset Introduction} After instructing the participants on what they needed to do, we showed them the relational table {\sc Cars} that the plots are based on. We gave them 2 minutes to get familiar with the data set and ask any questions related about it. After they were familiarized, we gave each user a training session for each tool.

\paragraph{Excel Training} We first introduced the Excel spreadsheet interface and showed the participants (1) how to make a scatter plot, and (2) how to do data aggregation using PivotTable, a feature in Excel that enables users to do data aggregation without any coding knowledge. Afterwards, we gave the users 10 minutes to play around with the tool. We encouraged them to try and produce a line chart and a bar chart. In addition to the training session, we also provided a ``cheat-sheet'' that included Excel documentation that we thought would be helpful to the user when performing the task. Finally, we allowed them to search online for help during the study.

\begin{figure}
    \centering
    % \fbox{\includegraphics[width=0.6\textwidth, trim=40 320 100 30, clip]{figures/Graphy-query4-screenshot.PNG}}
    \fbox{\includegraphics[width=1.0\textwidth]{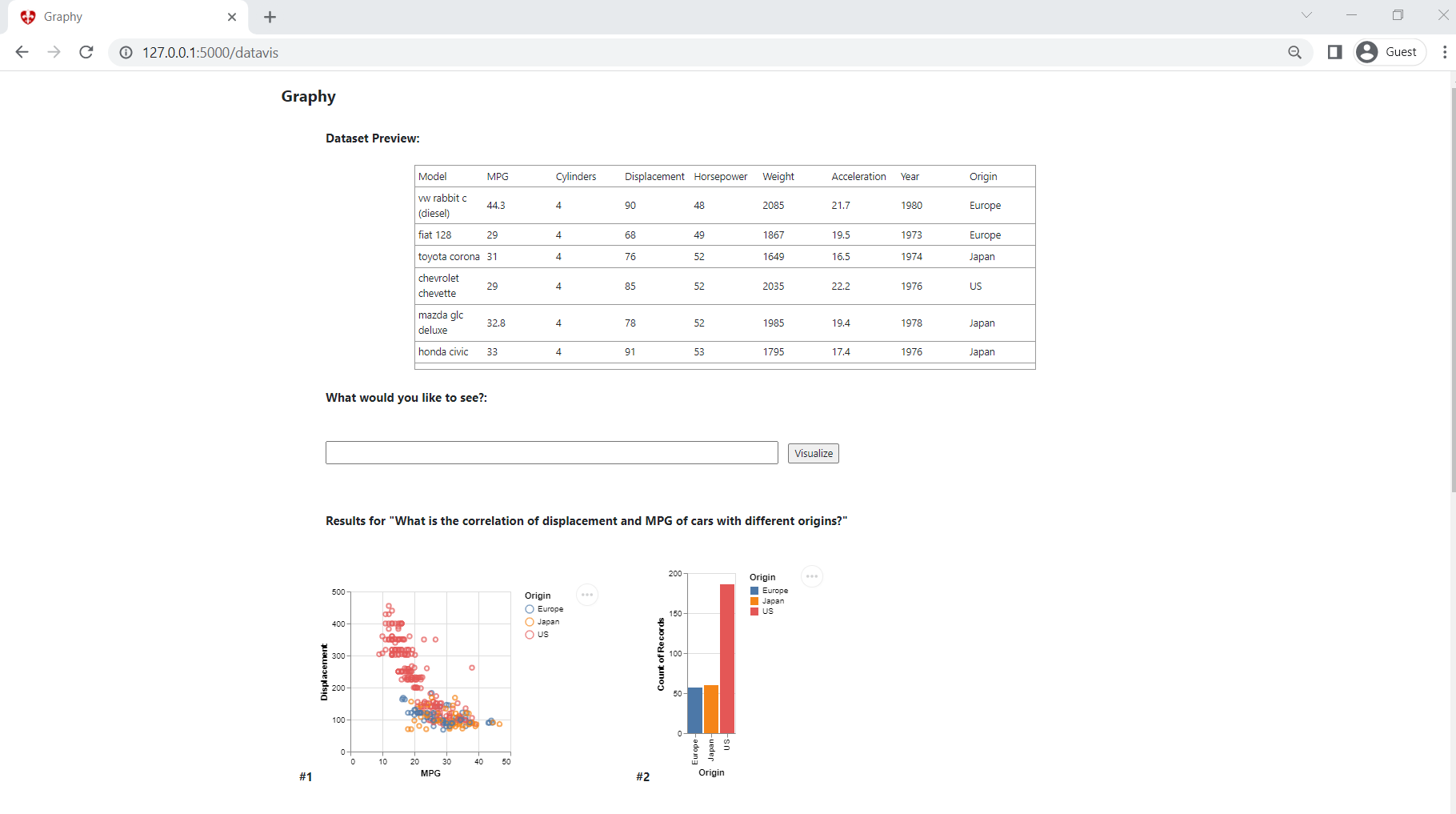}}
    \caption{\toolname Interface}
    \label{fig:graphy_interface}
\end{figure}

\paragraph{\toolname Training} Like with Excel, we started the training by introducing the participant to \toolname's UI, shown in Figure \ref{fig:graphy_interface}.  To avoid biasing the user in any way, we did not present any examples about how to use \toolname but simply asked the user to try it themselves for 5 minutes.

\paragraph{User Study Workflow} Once the participant was familiar with the data set and the tools, we gave them 2 minutes to familiarize themselves with the plots they needed to reproduce. We told the user to let us know when we could start timing and when they thought they had finished the task. In total, each participant had an hour to complete all the tasks using all the tools (15 minutes per task per tool). We made it clear to the participants that they were not required to reproduce \emph{exactly} the same plot as shown in the ground truth, and they could consider themselves to be finished as long as they thought they had produced a plot that conveys the same meaning. 

When using \toolname, users would enter a natural language query, and \toolname would return the top-$10$ results as visualizations back to the user. The user would then skim through the graphs and choose a visualization if they thought was equivalent to the ground truth. If the participant decided none of the visualizations shown was the one they wanted, they could try again by entering a different query. For Excel, we provided the user a spreadsheet that contained the table to be visualized so they did not need to import the data to Excel. During their time working on the plot, the user was allowed to use any resources such as searching the Internet or using the ``cheat-sheet'' we provide.

At the end of the session, we went over the participants' solutions and collected data on how many of the tasks they successfully solved, as well as the time it took to solve them with Excel and \toolname.

\end{document}